\documentclass[12pt]{amsart} 
\pdfoutput=1
\usepackage{amssymb,amsmath,amscd}
\usepackage[pdftex]{graphicx}  

\usepackage{rotating}
\usepackage[title]{appendix}
\usepackage{wasysym}
\usepackage[hyphens]{url}
\usepackage{hyperref}
\input cyracc.def

\newcommand{\mycomment}{}   

\theoremstyle{plain} 
\newtheorem{Lem}{Lemma}[section] 
\newtheorem{Prop}[Lem]{Proposition} 
 
\newtheorem{Cor}[Lem]{Corollary} 
\theoremstyle{definition} 
\newtheorem{Def}[Lem]{Definition} 
\newtheorem{Rem}[Lem]{Remark} 
\newtheorem{Con}[Lem]{Conjecture}

\newtheorem{Rule}[Lem]{Rule} 
\errorcontextlines=0  \numberwithin{equation}{section}

\newcommand{\bpi}{\begin{picture}} 
\newcommand{\epi}{\end{picture}}

\begin{document}
\sloppy

\title[An Astronomical Interpretation of the Nebra Sky Disc] 
{An Astronomical Interpretation of the Nebra Sky Disc} 
\date{November 2024} 
\author[B. Fiedler]{Bernd Fiedler}
\address{Bernd Fiedler \\ Eichelbaumstr. 13 \\ D-04249 Leipzig \\ Germany}
\urladdr{http://www.fiemath.de/}  
\email{bfiedler@fiemath.de}  

\begin{abstract} 
We agree with the interpretation of W. Schlosser, that the Nebra Sky Disc is a reminder of a method of determining a start date (and possibly also an end date) of the farming year. We extend this interpretation. We think that we
found the constellation Taurus on the Disc, which forms by addition of
three stars from the constellation Gemini the pattern of a plough of
Bronze Age. Moreover we found a line on the disc consisting of the
stars $\epsilon$ Gem, $\theta$ Aur, $\beta$ Aur and $\alpha$ Aur, which we called the Auriga line. We think that the Nebra people used the Auriga line to determine the day (which we call beta day) on which the Pleiades are vertically below $\beta$ Aur at dusk in February. We found a second representation of the Auriga line on the Disc where the distance ratios between the stars are very precisely equal to the distance ratios in the sky, and where the Pleiades are vertically below $\beta$ Aur. This proves that the Nebra people must have measured the distances, and that our hypothesis is correct.

The beta day could have been used to harmonize a lunisolar calendar with the solar year. However, the most likely possibility seems to us that a good sowing date could be determined by setting the sowing on the second round lunar phase after the beta day. (By round lunar phases we mean the full moon and the new moon.) Such a sowing date makes it possible to start a week count based on the lunar phases with the sowing in order to determine other agricultural dates.

The astronomical
knowledge for this procedure can be gained by astronomical
observations only. No mathematical calculations and import of knowledge from a Mediterranean culture are necessary.

One can determine a point on the ecliptic (beta point) such that the beta day occurs when the sun is at the beta point. The beta point moves only slowly relative to the vernal equinox. However, since the Bronze Age, the beta point and vernal equinox have swapped their order on the ecliptic.
\end{abstract}

\maketitle 

%
%

\section{Introduction}

The Nebra Sky Disc was found in 1999 on the Mittelberg in Saxony-Anhalt in Germany.
The first astronomical interpretations of the Nebra Sky Disc were
given by W. Schlosser \cite{schlosser1, schlosser2, schlosser3} and
R. Hansen \cite{hansen1, hansen2}. W. Schlosser saw the Disc as a
reminder of a method of determining the start and end dates for the
farming year. In his opinion, the heliacal setting (evening setting) of the Pleiades in
spring and the morning setting of the Pleiades in autumn came into
question, which the people in the Nebra region would have been able to
observe. The days of these settings were supplemented by conjunctions
between the moon and the Pleiades, which appear as a conjunction
between the Pleiades and the waxing moon shortly after the new moon in
spring and as a conjunction between the Pleiades and full moon in
autumn.

R. Hansen sees in the Disc the pictorial representation of a rule for inserting a leap month in order to harmonize the solar and lunar years. The rule is
known from the much later Babylonian calendar and is also based on a
conjunction between the moon and the Pleiades. E. P\'asztor and
C. Roslund \cite{pasztorroslund} provide a third interpretation of the
Sky Disc. They see it as an object for cultic purposes that has no
astronomical background.

We agree with the interpretation of W. Schlosser. In particular, we
interpret the large objects on the disc in the same way as
W. Schlosser (see Fig. \ref{bigobj}).

\mycomment{
\begin{figure}[!ht]
  \begin{center}
    \includegraphics[width=\textwidth]{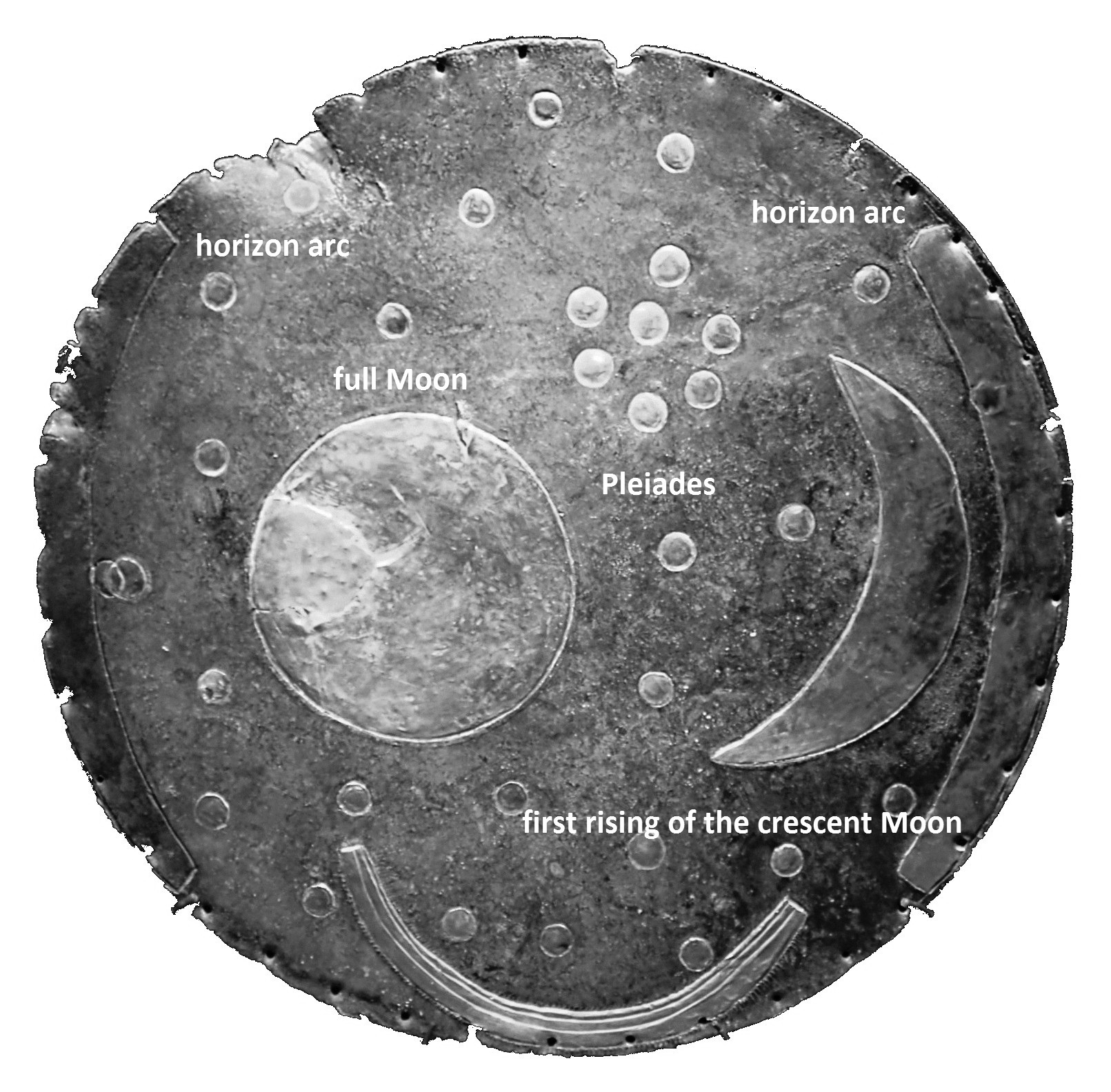}
\end{center}
\caption{The big objects on the Sky Disc} 
\label{bigobj}
\end{figure}
}
According to W.Schlosser, the circle seen as the moon could also be the sun and the crescent moon could also represent an eclipse, but the meanings shown in Figure \ref{bigobj} are primarily discussed by W.Schlosser.

We add new details to the theory of W. Schlosser. We believe
that we have found a star constellation and a line consisting of
stars from the constellations Auriga and Gemini on the Disc (Auriga line). The Nebra people\footnote{According to Meller and Michel, the creators of the Sky Disc were members of the {\it\'Un$\check{e}$tice Culture} (see \cite[p.45]{melmich1}).} could have used
the Auriga line stars to predict the imminence of the heliacal setting of
the Pleiades from the position of the Pleiades in relation to these
stars.

They determined the day on which the Pleiades, when they appear at nautical twilight, are exactly perpendicular to $\beta$ Aur (beta day). In the Bronze Age, this day occurred 24 to 28 days before the heliacal setting of the Pleiades. By determining this day, the Nebra people were able to predict that the heliacal setting of the Pleiades was imminent (Section \ref{sec4}).

In addition, with the help of the beta day they could have
started a lunisolar calendar every year, which is well adapted to a solar
calendar and to the seasons (Section \ref{subsec6.1}).

However, we consider the possibility that the Nebra people used the beta day to set a sowing date to be the most likely. A favorable sowing date at the time of the Sky Disc would have been the second round phase of the moon (full moon or new moon) that followed the beta day. If the start of sowing was set on such a day, then the sowing could also start a week count that was based on the phases of the moon and which could then be used to determine further agricultural dates over the course of the year (Section \ref{subsec6.2}).

The Nebra people depicted the Auriga line twice on the Disc. The second representation is a scaled reproduction of the celestial image of the Auriga line on the Disc (Section \ref{sec524a} and Appendix \ref{appB}). This suggests the following:
\begin{itemize}
\item It is very likely that the Nebra people measured distances between stars in the sky.
\item Our ideas about the use of the Auriga line presented in this paper are probably correct.
\end{itemize}

Our investigations are based, among other things, on own
astronomical observations from March to May 2020 and October 2022 and on calculations with the {\tt Stellarium} software 0.21.0 \cite{zottihoffmann, zottiwolf}, with {\tt Mathematica} \cite{mma} and {\tt Peakfinder} \cite{peakf1, peakf2} and on measurements on photos using {\tt Gimp} \cite{gimp}.

\section{Star Constellations}
\subsection{Can the image of the Sky Disc show a star constellation?}
In all of his papers on the Sky Disc, W. Schlosser tried to show that
there were no star constellations on the Disc. 1.) He generated random
star arrangements by computer. All these computer images showed
clusters of stars in certain areas, which he called
"constellations". In the image on the Disc, however, there are no such
clusters of stars (if one disregards the object of the Pleiades).

2.) If people are asked to draw dots on a piece of paper without creating any meaning, they place the dots far apart and, if possible, in such a way that no connection is suggested.

However, both of these do not mean that no constellation can be contained in the point cloud.

In order
to orientate oneself in the sky, humans connect certain, often bright
stars with lines, so that a pattern arises that one can memorize and
recognize. Such patterns are constellations. The stars in them almost
never cluster together (see Fig. \ref{bear}).
\mycomment{
\begin{figure}[!ht]
  \begin{center}
    \includegraphics[width=\textwidth]{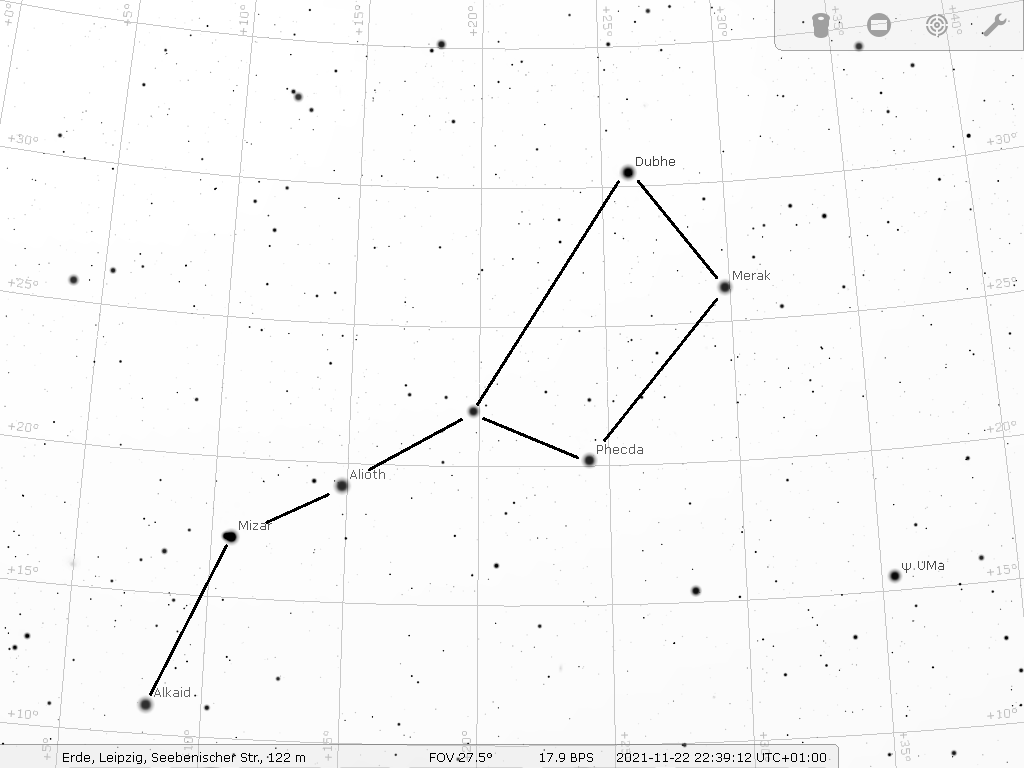}
\end{center}
\caption{The constellation Great Bear is not a cluster of stars but a pattern in a large area.} 
\label{bear}
\end{figure}
}

If one has recognized a constellation, one can also identify and
use the individual stars in it, for example to navigate a ship.

The absence of clusters or a dot placement that does not reveal any connections does not mean that an arrangement of dots does not contain a constellation.

One could also ask the question: Why is the Sky Disc round? One answer could be: Because something is depicted on the Disc that, by rotating the Disc, can be made to coincide with something that can also be seen in nature. And that can only be constellations.

\subsection{The constellations Taurus and ''Plough''} \label{subsec2.2}
The star cluster of the Pleiades belongs to the constellation
Taurus. The main part of the constellation Taurus consists of two
v-shaped lines that connect the star $\gamma$ Tau with the stars
$\beta$ Tau and $\zeta$ Tau (see Fig. \ref{pflug}).
\mycomment{
\begin{figure}[!ht]
  \begin{center}
    \includegraphics[width=\textwidth]{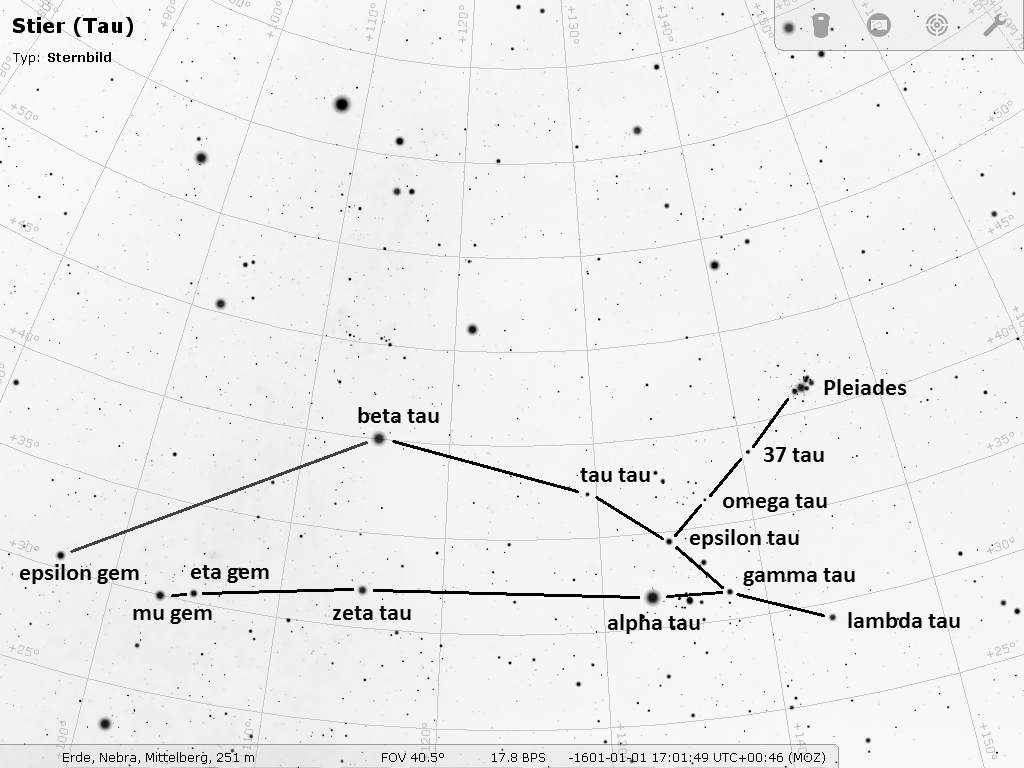}
\end{center}
\caption{The constellations Taurus and ''Plough''. The abbreviations tau and gem denote the affiliation of the stars to the constellations Taurus or Gemini.}
\label{pflug}
\end{figure}
}
The Pleiades are connected to these lines by a line running over the
stars 37 Tau and $\omega$ Tau.

The same image can be found on the Sky Disc (see Fig. \ref{pflug2}).
\mycomment{
\begin{figure}[!ht]
  \begin{center}
    \includegraphics[width=\textwidth]{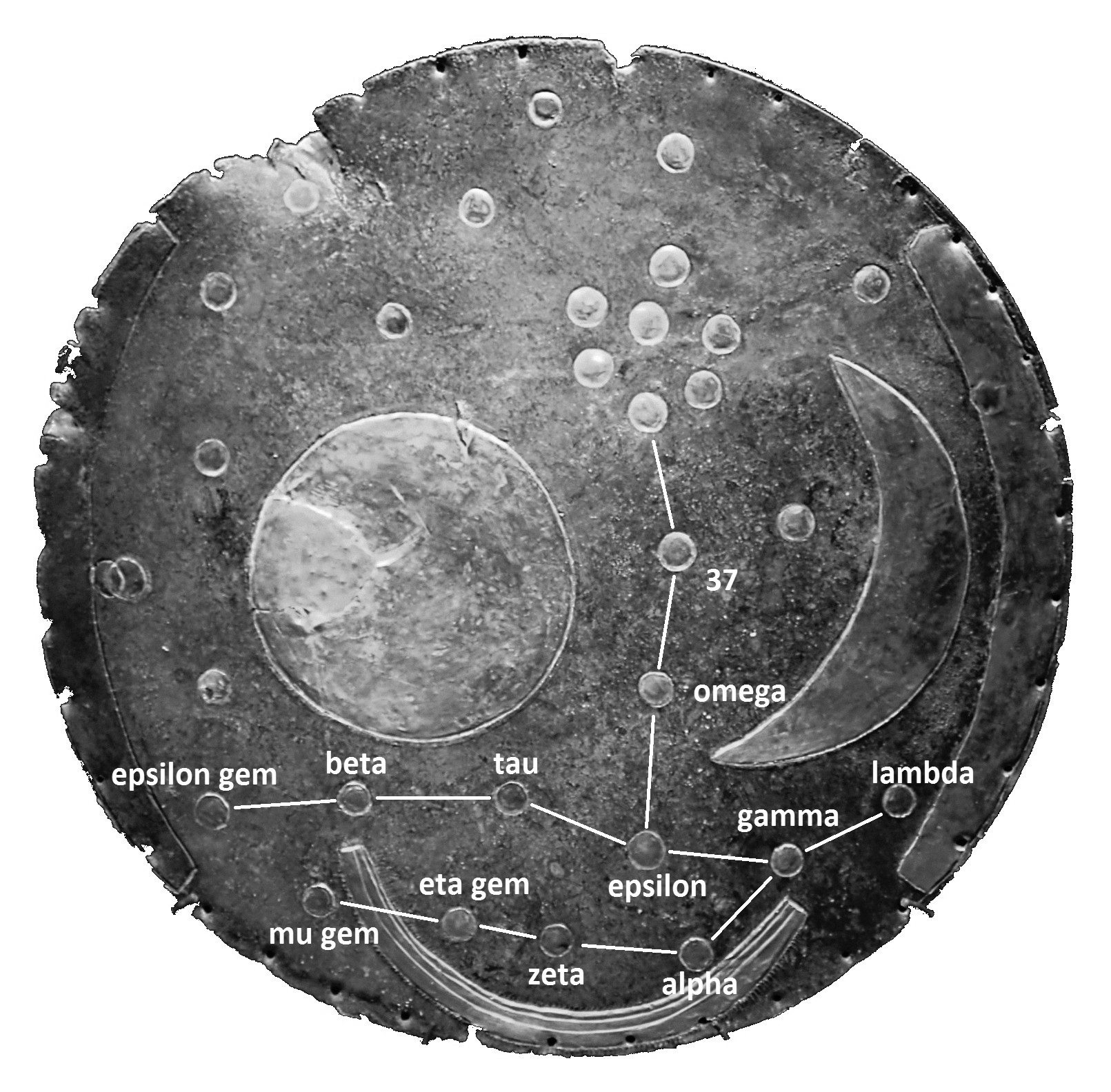}
\end{center}
\caption{The constellations Taurus and ''Plough'' on the Sky Disc} 
\label{pflug2}
\end{figure}
}
We see a line running over 2 stars that leads from the Pleiades to the lower part of the
Disc, where it meets two v-shaped lines on which there are the same
number of stars of Taurus as on the star map Fig. \ref{pflug}. On the
Disc Fig. \ref{pflug2}, the v-lines are extended by 3 more stars. The
stars $\mu$ Gem, $\eta$ Gem and $\epsilon$ Gem of the constellation
Gemini correspond well to these stars (see Fig. \ref{pflug}). $\mu$
and $\eta$ fit because they extend the ''$\alpha$ Tau'' - ''$\zeta$
Tau'' line as dead straight as on the Disc. $\epsilon$ fits because it
is slightly to the left above $\mu$. The resulting constellation is
very similar to a plough from the Bronze Age (see Fig. \ref{walle}).
\mycomment{
\begin{figure}[!ht]
  \begin{center}
    \includegraphics[width=\textwidth]{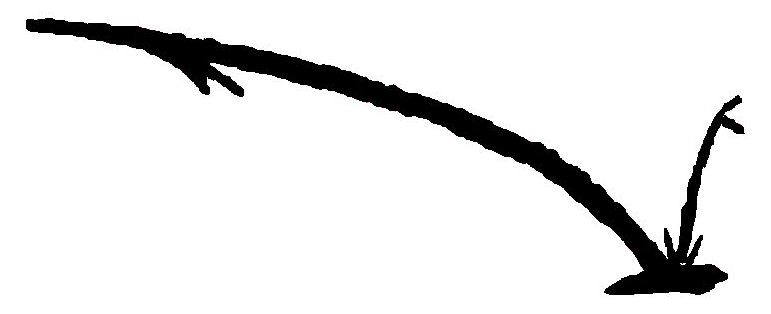}
\end{center}
\caption{Silhouette of the Plough of Walle} 
\label{walle}
\end{figure}
}
This impression arises not so much from the image on the Disc
(Fig. \ref{pflug2}) but mainly from the image in the sky itself
(Fig. \ref{pflug}).

There are indications that the Nebra people actually saw a
plough in this constellation. In the nature, there is no equivalent to
the star to the right of 37 Tau on the Disc (Fig. \ref{pflug2}). But
exactly there would be the handle with which the farmer steers the
plough. Furthermore, the line ''$\gamma$ Tau'' - ''$\mu$ Gem'' has a
kink at $\alpha$ Tau and does not run straight as in
Fig. \ref{pflug}. This is where the ploughshare would be, which is
inclined to the ground and not tangential as in
Fig. \ref{pflug}. Because of the kink at $\alpha$ Tau, a kink also had
to be added at $\epsilon$ Tau.

There are also some problems with the mapping we made between the
stars on the Disc and the stars in the sky.\\*[0.2cm]
(i) The stars $\omega$ Tau and 37 Tau, which connect the Pleiades with the constellation Taurus, are faint. They have magnitudes
of 4.90 and 4.35 (values from {\tt Stellarium} \cite{zottihoffmann,
  zottiwolf}). (The smaller the magnitude, the
brighter the star.) Could the Nebra people actually see these stars?

It is commonly said that the naked human eye can see stars up to
magnitude 5 (see \cite[p.31]{mueller}). R. Brandt \cite[p.126]{brandt}
specifies this. If the human eye is in the dark for 15 minutes, it can
see stars up to magnitude 5. If it is 1 hour in the dark, it sees
stars up to magnitude 6. So the Nebra people could see the two faint stars.\\*[0.2cm]
(ii) Between the stars $\gamma$ Tau and $\alpha$ Tau there is the star
$\theta$2 Tau and between $\gamma$ Tau and $\epsilon$ Tau there is the
star $\delta$1 Tau. These stars seem to be missing from the Disc.

The star we identified as $\epsilon$ Tau cannot be any other star,
since the line coming from the Pleiades ends at $\epsilon$. And then
the star opposite $\epsilon$ on the second Taurus line must also be
the star $\alpha$, since both stars are symmetrical. Hence, $\theta$2
and $\delta$1 are the missing stars.

Probably the Nebra people deliberately left out $\theta$2 and
$\delta$1, perhaps because such closely spaced stars could not be
attached to the disc, or because they only wanted to mark the
beginning and end of these star-rich line segments.\\*[0.2cm]
(iii) On the disc, the distances between the stars on the Taurus'
v-shaped lines are approximately equidistant (Fig.
\ref{pflug2}). However, this is not the case in nature
(Fig. \ref{pflug}).

Probably the Nebra people just wanted to indicate which stars are on a
common line and how many there are. Perhaps in such a situation people
tend to draw equidistant distances rather than depicting the actual
distance relationships. That this could be true is shown by the
example of the Sky Tablet of Tal-Qadi (Fig. \ref{talqadi}).

\mycomment{
\begin{figure}[!ht]
  \begin{center}
    \includegraphics[width=\textwidth]{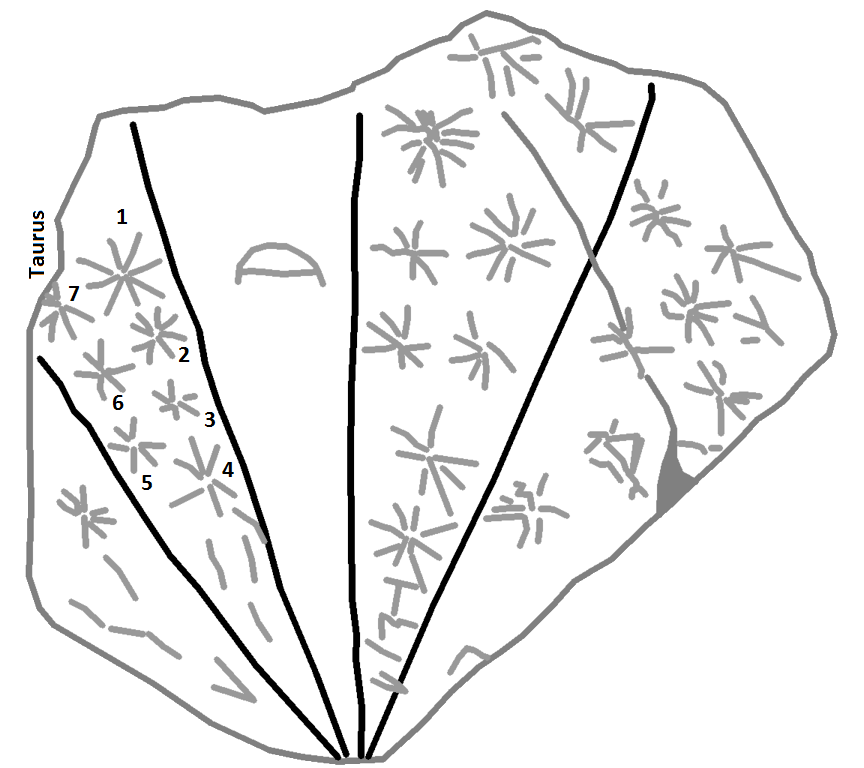}
\end{center}
\caption{The star constellation Taurus on the Sky Tablet of Tal-Qadi:
  1=beta, 2=epsilon, 3=delta1, 4=gamma, 5=theta2, 6=alpha, 7=zeta.} 
\label{talqadi}
\end{figure}
}

P. Kurzmann \cite{kurzmann1, kurzmann2} has assigned real
constellations to all star arrangements on the Tablet. The second
segment from the left shows the constellation Taurus. There all stars
are also arranged equidistantly.

Interestingly, the stars $\theta$2 and $\delta$1 are present on the
Tablet. Instead, the star $\tau$ Tau has been omitted.

The deviation from equidistance is particularly strong for the stars
$\mu$ Gem and $\eta$ Gem (see Fig. \ref{pflug} and \ref{pflug2}). Why
did the Nebra people include these stars in the line anyway?

The stars of Taurus are not very bright except for $\alpha$ Tau
(Aldebaran) and $\beta$ Tau. If Aldebaran or the Pleiades are near the western horizon (about to set), then
they are darkened even further by the extinction and are difficult to
find. In this case, however, the Plough is very steep and the line ''$\mu$ Gem'' - ''$\alpha$ Tau'' is an important
orientation line that leads from the brighter stars of Gemini to
Aldebaran and the neighboring Pleiades (Fig. \ref{pflug_steil}). 

\mycomment{
\begin{figure}[!ht]
  \begin{center}
    \includegraphics[width=\textwidth]{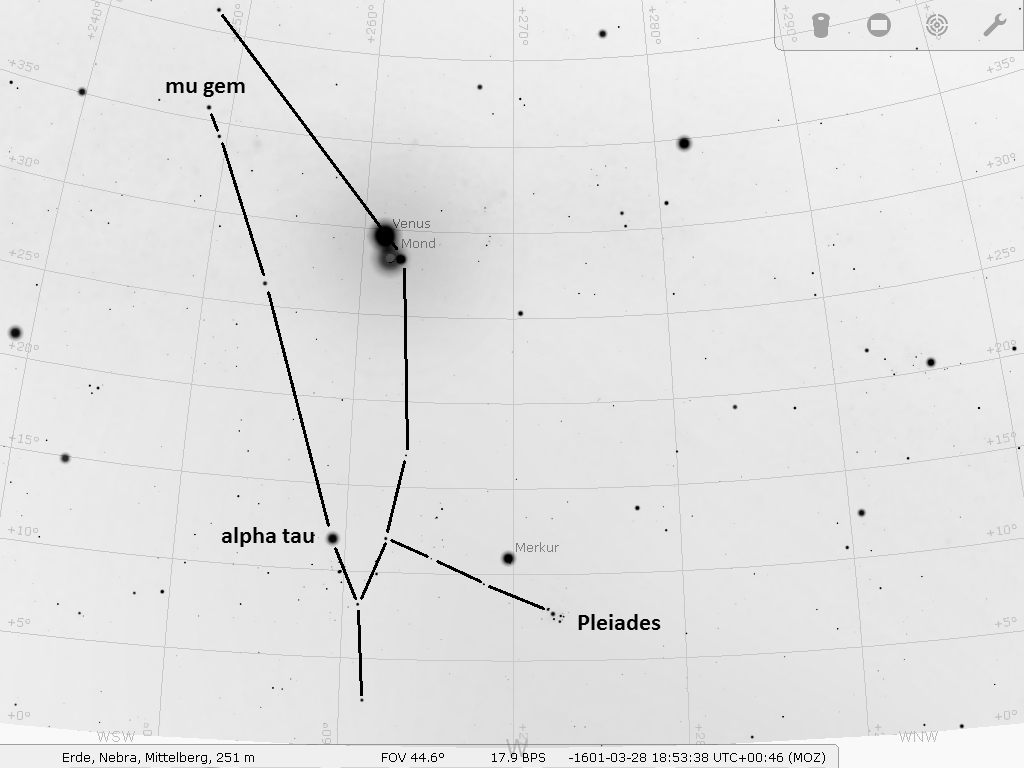}
\end{center}
\caption{On the day of the heliacal setting of the Pleiades, the Plough is steep.} 
\label{pflug_steil}
\end{figure}
}

If the beginning of this line
is marked by two stars that are close together, it is easier to see
that one is on the right line.\\*[0.2cm]
If we accept the explanations given here for the above critical points
(i)-(iii), then our interpretation of the stars discussed could be
correct.
\begin{Rem}
In the English language, the main part of the constellation Great Bear
(Fig. \ref{bear}) is also referred to as the ''Plough''. In our paper, however, by ''Plough'' we always mean the constellation in Fig. \ref{pflug}.
\end{Rem}

\subsection{The ''Auriga line''} \label{subsec2.3}
We can identify 3 more stars on the Disc with real stars. On the Disc,
a line leads away from the star $\epsilon$ Gem, which runs over the stars we have designated as $\theta$ Aur, $\beta$ Aur
and $\alpha$ Aur (Fig. \ref{auxiliaryline2}). 

\mycomment{
\begin{figure}[!ht]
  \begin{center}
    \includegraphics[width=\textwidth]{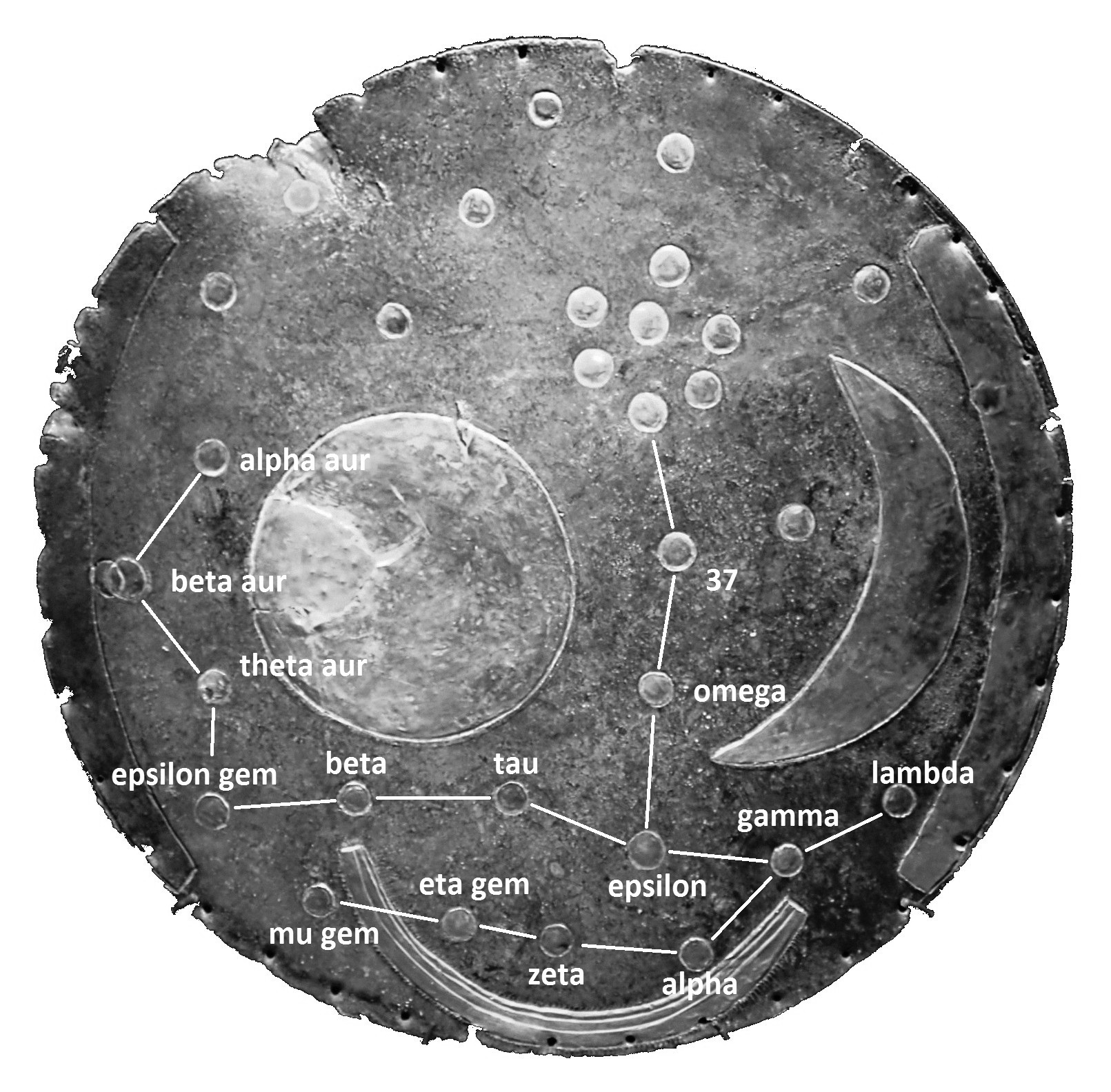}
\end{center}
\caption{The constellation ''Plough'' and the ''Auriga line'' on
  the Sky Disc} 
\label{auxiliaryline2}
\end{figure}
}
The real stars $\theta$
Aur, $\beta$ Aur and $\alpha$ Aur (Capella) in the constellation
Auriga form the same pattern in the sky as the stars of the same name
on the Disc (Fig. \ref{auxiliaryline}).
\mycomment{
\begin{figure}[!ht]
  \begin{center}
    \includegraphics[width=\textwidth]{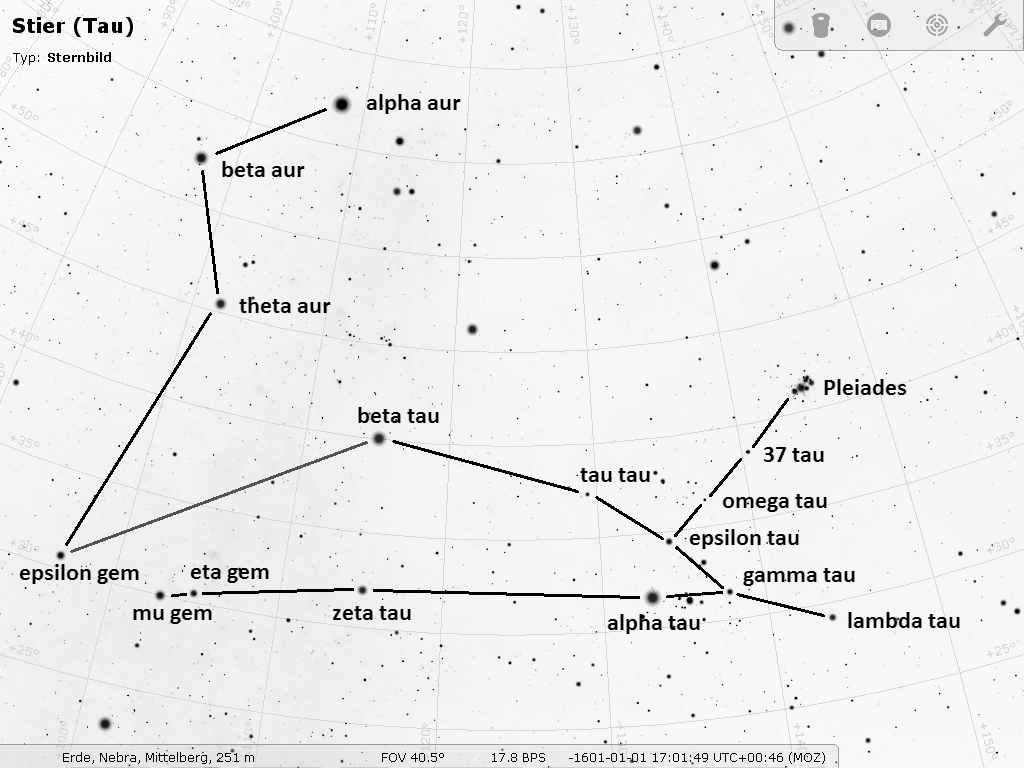}
\end{center}
\caption{The constellation ''Plough'' and the ''Auriga line''} 
\label{auxiliaryline}
\end{figure}
}
This identification is based on the fact that the stars $\epsilon$ Gem,
$\theta$ Aur and $\alpha$ Aur lie on a common straight line, from which the star
$\beta$ Aur deviates in the characteristic peak. Only the pattern and the
number of stars matter. The distances between the stars on the disc or
in the sky can differ as in Subsection \ref{subsec2.2}.

We call the pattern from the stars considered here ''Auriga
line''. It will be seen that the Nebra people presumably used these stars to predict the imminence of the Pleiades' heliacal setting.

\section{Sidereal time, stellar time and solar time} \label{sec3a}
In astronomy there are the following time units, which we also use in our paper.
\begin{Def}
 Let an observation location $P$ be given.
  \begin{itemize}
   \item The time that the sun needs to travel from its transit across the local meridian of $P$ back to the local meridian of $P$ after 1 rotation of the Earth is called {\it solar day}.
    \item The time that the vernal equinox needs to travel from its transit across the local meridian of $P$ back to the local meridian of $P$ after 1 rotation of the Earth is called {\it sidereal day}.
  \item  The time that an infinitely distant fixed star without proper motion needs to travel from its transit across the local meridian of $P$ back to the local meridian of $P$ after 1 rotation of the Earth is called {\it stellar day}.
  \end{itemize}
  All these periods are divided into 24 hours x 60 minutes x 60 seconds.
  The time measured by the movement of the sun, the vernal equinox or a fixed star is called ({\it apparent or true}) {\it solar time, sidereal time, stellar time}.
\end{Def}



At a fixed observation site, it is 12 o'clock local true
solar time when the sun is at the local meridian.

Sidereal time is also
a local time, i.e. for each observation location it is defined that it
is 0 o'clock sidereal time when the vernal equinox is on the local meridian.

True solar time is not a uniform flow of time. Therefore, a {\it mean
  solar time} has been introduced that represents a uniform time flow
with a day length of 24 hours.

There is also a mean sidereal time. But we do not use this.\\*[0.2cm]
\fbox{
  \parbox{0.9\linewidth}{
In this paper, we only use:
\begin{itemize}
\item mean solar time,
\item apparent sidereal time and apparent stellar time.
\end{itemize}}
}\\*[0,2cm]



The sidereal day is approximately 3 minutes 56 seconds shorter than the solar day.
This is illustrated by Fig. \ref{stellarday}.

Position 1 shows the earth when the Sun or a distant star is on the
local meridian of an observer. In position 2, the earth has rotated
360 degrees and the local meridian is again pointing in the same
direction as in position 1. If there had been a distant star on the
local meridian in position 1, it would now be on the local meridian
again. One sidereal day would have passed.

The Sun, however, would not yet be on the local meridian. The time to
position 3 would have to pass before the local meridian points to the
Sun and a solar day is over.

Very precise values for the duration of a solar day and a sidereal day are given in \cite{eop} and \cite[p.66]{ahnert2}.
\begin{eqnarray}
1 \text{ mean solar day} & = & 86400 \text{ seconds} \label{3.1}\\
1 \text{ sidereal day} & = & 86164,09053083288 \text{ seconds} \label{3.2}\\
                       & = & 0,997269566329084 \text{ mean solar day}\nonumber 
\end{eqnarray}
From (\ref{3.1}) and (\ref{3.2}) we can now calculate the difference between solar day and sidereal day.
\begin{eqnarray}
\hspace*{0.8cm} 1 \text{ mean solar day} - 1 \text{ sidereal day} & = & 235,90946916712
\text{ seconds} \\
 &=& 3 \text{ minutes } 55,90946916712 \text{ seconds}\nonumber
\end{eqnarray}
Viewed from the same location, a star seen at one position in the sky
will be seen at the same position on another night at the same sidereal
time. This is particularly true of its positions in rising, transit and setting.
\mycomment{
\begin{figure}[!ht]
  \begin{center}
    \includegraphics[width=8cm]{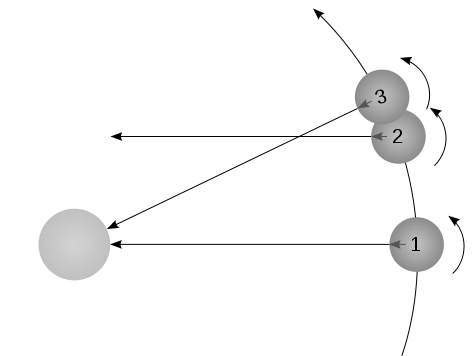}
\end{center}
\caption{Solar day and sidereal day.} 
\label{stellarday}
\end{figure}
}

Let the celestial sphere be seen from a fixed observation point at a
point in sidereal time $ t_0 $. Then the celestial sphere shows exactly the
same sight on every other day at the same sidereal time. Only the positions of the sun, moon and planets
have changed.

The fact that a given view of the sky always recurs at the same
sidereal time is only valid in shorter periods of time. In large
periods of time there is a shift in the sidereal time at which a given
view of the sky appears (see Table \ref{tab10}).

The stellar day is about 0.008 seconds longer than
the sidereal day because the vernal equinox moves slowly due to the precession and nutation of the Earth's axis. As a result, the shift mentioned above is smaller for stellar time than for sidereal time.

\cite{eop} indicates for the duration of the stellar day:
\begin{eqnarray}
1 \text{ stellar day} &=& 86164,098903691 \text{ seconds} \label{3.4}
\end{eqnarray}
From (\ref{3.4}) and (\ref{3.2}) we obtain
\begin{eqnarray}
1 \text{ stellar day} - 1 \text{ sidereal day} & = & 0,00837285812 \text{ seconds}
\end{eqnarray}






\section{Heliacal settings} \label{sec3}
According to the interpretation of W. Schlosser, the Nebra people of the Bronze Age chose the heliacal setting of the Pleiades as the starting date of the farming year and determined it through observations. In this section we put together some basic astronomical concepts that can be used to explain and calculate heliacal sets.

\subsection{What is a heliacal setting?}
If we use sidereal time, the rise, transit, and setting of a star
occur at the same times every day. Measured in solar time, however,
the rise, transit and setting of a star are always 3 minutes 56
seconds earlier than the day before. The reason for this is that a
sidereal day is 3 minutes 56 seconds shorter than a solar day. As a
result, the setting of the Pleiades takes place 3 minutes 56 seconds
earlier each day than the day before.

In the first half of the year up to the summer solstice, the days get
longer, because the Sun sets a little later each day than the day
before (about 1 minute). The setting times of the Sun and the Pleiades
thus converge and could meet.

Table  \ref{tab1} shows examples of set times from the present and from the time of the Nebra Disc.
\begin{table}[t] 
\begin{center}
\begin{tabular}{|l|l|l|}
\hline
 setting Sun &  setting $\eta$ Tau & apparent sidereal time \\
\hline
 2020-1-1 Greg. 16:12  & 2020-1-2 Greg 5:34  & 12:05:15 \\
\hline
 2020-5-30 Greg. 21:14  & 2020-5-30 Greg. 20:44 & 12:05:02 \\
\hline
\end{tabular}
$\;$\vspace{10pt}
\begin{tabular}{|l|l|l|}
\hline
 setting Sun &  setting $\eta$ Tau & apparent sidereal time \\
\hline
 -1601-1-1 Jul. 15:55  & -1601-1-2 Jul. 1:21 & 7:08:11 \\
\hline
 -1601-4-28 Jul. 18:44  & -1601-4-28 Jul. 17:41 & 7:08:13  \\
\hline
\end{tabular}
\vspace{3mm}
\caption{Some settings of the Sun and of $\eta$ Tau (Pleiades) in 2020
and -1601. The sidereal times are the times for the setting of the
Pleiades. If a Gregorian date is given, ''time'' is given in CET. In the case of a Julian date, ''time'' is Local Mean Solar Time.}
\label{tab1}
\end{center}
\end{table}
At the beginning of the year the setting of the Pleiades takes place
deep in the night. It is not affected by the light of the Sun. The
Pleiades can be seen from sunset until its own setting.

In May or April, however, the Pleiades set before the Sun. They can
then no longer be seen after sunset. In both examples there must be a
day between the two dates in Table \ref{tab1} on which the Pleiades can be seen for the
last time after sunset. This day is the {\it heliacal setting} of the
Pleiades (see \cite[p.36]{weigzimm}).

On this day the Pleiades only appear shortly after sunset, when dusk
is so far advanced that they can be seen. After a few minutes they
disappear again because then they set themselves. The Pleiades
disappear in the rays of the Sun. ''Heliacal'' means "in the rays of
the Sun". But actually there is only a remnant of sunlight. The day after the heliacal setting, the Pleiades no longer appear.

Briefly something about the phases of twilight (see \cite{astro}). Twilight is divided
into 3 phases (see Table \ref{tab2}).
\begin{table}[t] 
\begin{center}
\begin{tabular}{|l|c|}
\hline
 Twilight phase &  depth of the Sun below the horizon\\
\hline
 civil twilight  & $0^\circ$ - $6^\circ$ \\
\hline
 nautical twilight  & $6^\circ$ - $12^\circ$ \\
\hline
 astronomical twilight  & $12^\circ$ - $18^\circ$ \\
\hline
\end{tabular}
\vspace{3mm}
\caption{The phases of twilight.}
\label{tab2}
\end{center}
\end{table}
In civil twilight, stars with a magnitude of 1 or brighter are only
visible. Constellations cannot yet be recognized because too many
stars are missing. In the nautical twilight stars of a magnitude 3
or brighter are visible and thus also constellations. But there are
also objects on earth, such as the horizon, recognizable, so that
measurements for ship navigation are possible. The Pleiades cannot be seen until nautical twilight at the earliest, as their brightest stars are of the magnitude of about 3.

\subsection{The Arcus Visionis}
The traditional method of determining the heliacal setting of a star
is based on what is known as the arcus visionis, in German
''Sehungsbogen''.

If we want to observe a star of a certain magnitude shortly after
sunset, the Sun must be a certain depth below the horizon so that the
star can be seen at all. If we want to watch the setting of this star,
which happens shortly after sunset, the depth of the Sun will have to
be greater than normal because the star will then approach the horizon
and its magnitude will be further reduced by the extinction of the
atmosphere. If in the following days the time of the star's setting
approaches the time of the Sun's setting, the depth of the Sun will
decrease because the Sun has less time until the star's setting to
gain depth. The smallest value at which the setting of the star can
still be seen is the depth of the Sun on the day of the heliacal
setting of the star.

Because of its extinction, the star will usually not reach the horizon
when it sets, but will instead disappear at a certain height. For mathematical reasons one does not use the depth of the Sun at the time of the actual disappearance of the star, but one waits for the time until the star reaches the horizon and works with the depth that the Sun then has.
\begin{Def}[{\cite[p.731]{schoch}}]
The {\it arcus visionis} $\sigma$ of a star is the depression of the Sun below
the horizon, measured in the vertical circle for the moment when the
star sets on the last evening when it is visible or rises on the first
morning when it is visible, refraction being disregarded in the case
of both bodies.
\end{Def}
Figure \ref{arcvis} shows the arcus visionis of the Pleiades in the year -1601.
\mycomment{
\begin{figure}[!ht]
  \begin{center}
    \includegraphics[width=\textwidth]{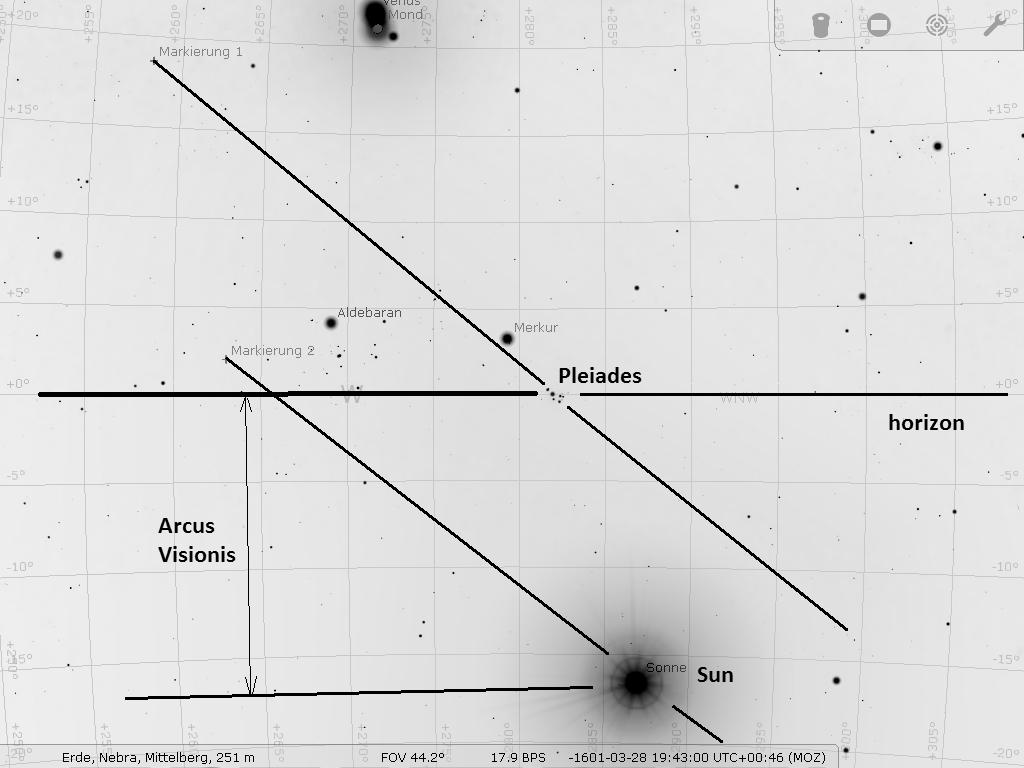}
\end{center}
\caption{-1601: Sun, Pleiades and Arcus Visionis on the day of the heliacal setting of the Pleiades, when the Pleiades are on the horizon.} 
\label{arcvis}
\end{figure}
}
On the days before the heliacal setting of the Pleiades, the depth of
the Sun is greater than $\sigma$ when the Pleiades are on the
horizon. The Pleiades can then be seen longer than at their heliacal
setting. After the heliacal setting of the Pleiades, the depth of the
Sun is less than $\sigma$. There is then too much residual light from the Sun and the Pleiades can no longer be seen.


Since the determination of the arcus visionis is an old topic, it is
difficult to find information about values of $\sigma$
today. Wikipedia \cite{wikibogen} gives the following $\sigma$-values for the Pleiades without citing the source:
\begin{center}
\begin{tabular}{|l|c|}
\hline
 $\sigma$ Pleiades &  transparency of the air\\
\hline
 $14.5^\circ\leq\sigma\leq 15.5^\circ$  & good \\
\hline
  $19.5^\circ\leq\sigma\leq 20.5^\circ$  & poor \\
\hline
\end{tabular}
\end{center}
The value for good transparency of the air matches the value in
\cite[Table E 64]{neugebauer} where $\sigma = 15.5^\circ$ is given for
a star with magnitude 3.

W. Schlosser \cite{schlosser3} did not use an arcus visionis, but
explained that the heliacal setting of the Pleiades occurs when the
Pleiades are $5^\circ$ high and at the same time the Sun is $15^\circ$
below the horizon (both on the western horizon). Since the Pleiades
and the Sun are close to the ecliptic, the Sun will deepen another
$5^\circ$ as the Pleiades move to the horizon. Thus W. Schlosser's
approach corresponds to the use of an arcus visionis of $\sigma =
20^\circ$.

Since we later also have to consider the heliacal setting of the
Aldebaran ($\alpha$ Tau), we also need its arcus visionis. B.L. van der
Waerden \cite{waerden} gives $\sigma = 11.5^\circ$ for Aldebaran.

Finally, Table \ref{tab3} shows the $\sigma$-values that we will use.

\begin{table}[t] 
\begin{center}
\begin{tabular}{|l|l|}
\hline
 $\sigma = 16^\circ$ &  for Pleiades with good transparency of the air\\
\hline
 $\sigma = 20^\circ$  & for Pleiades with poor transparency of the air
and\\
 &  for comparisons with values from W. Schlosser\\
\hline
 $\sigma = 12^\circ$  & for Aldebaran \\
\hline
\end{tabular}
\vspace{3mm}
\caption{The  $\sigma$-values that we use.}
\label{tab3}
\end{center}
\end{table}

\subsection{Heliacal settings of the Pleiades and Aldebaran}
If the arcus visionis $\sigma$ of a star is known, the heliacal
setting of the star can be calculated using spherical trigonometry
(see \cite[p.581-586]{ideler2}). However, we use a method
\cite[p.57-58]{ideler1} that was
used in earlier times to determine heliacal settings with the help of
an armillary sphere made of metal. We carried out this method within
the framework of the {\tt Stellarium} software \cite{zottihoffmann,zottiwolf} (see Appendix \ref{appD}) and
thus determined the heliacal settings of the Pleiades and the
Aldebaran $\alpha$ Tau) compiled in Table \ref{tab4}. The accuracy of {\tt Stellarium} (1 arc second) is much greater than that of an armillary sphere.

\begin{table}[t] 
\begin{center}
\begin{tabular}{|c|c|c|}
\hline
 Pleiades $\sigma = 16^\circ$ & Pleiades $\sigma = 20^\circ$ &
 Aldebaran $\sigma = 12^\circ$\\
\hline
 -1-4-8 & -1-4-3  & -1-4-16\\
 -601-4-4 & -601-3-30 & -601-4-12\\
 -1001-4-2 &-1001-3-28  & -1001-4-11\\
\hline
 -1601-3-28& -1601-3-24& -1601-4-7\\
 -1801-3-27& -1801-3-22& -1801-4-5\\
 -1941-3-26& -1941-3-21& -1941-4-5\\
 -2101-3-25& -2101-3-20& -2101-4-4\\
\hline
\end{tabular}
\vspace{3mm}
\caption{Julian dates of the heliacal settings of the Pleiades and
  Aldebaran ($\alpha$ Tau), calculated for the Mittelberg.}
\label{tab4}
\end{center}
\end{table}

We have chosen 1 as the last digit of each year so that the year in
question is not a leap year. In a leap year, all dates after February
are shifted by one day. This could give the impression of inaccuracies.

The period -2101 to -1601 is the period in which the Sky Disc was
dated by H. Meller \cite{meller1,meller2}, E. Pernicka \cite{pernicka2, pernicka}
and others (see \cite[p.115]{pernicka}). We also consider the period
-1001 to -1 because there was a suggestion by R. Gebhard and R. Krause
\cite{gebhardkrause} to date the Disc to this period. We calculate
heliacal sets for both periods to see in which period the heliacal
sets of the Pleiades are
better suited as starting dates for the farming year.

Using the method described in Appendix \ref{appD}, we have determined
the Julian dates compiled in Table \ref{tab4} for the heliacal
settings of the Pleiades and Aldebaran. In order to be able to assess
what the seasonal conditions at the time of a heliacal set actually
are, the Julian dates in Table \ref{tab4} must be converted into
Gregorian dates. This can be done with the help of Table \ref{tab7}
and the formula
$$\text{Gregorian date} = \text{Julian date} - \Delta .$$
Table \ref{tab5} shows the results. In the first line of Table
\ref{tab5} we have also given the heliacal sets for 2020 in order to
be able to relate them to our own sky observations made this year.

The values for the heliacal settings of the Pleiades for -1601 and
-1941 at $\sigma = 20^\circ$ are the same as the dates given by
W. Schlosser in \cite{schlosser3}.

\begin{table}[t] 
\begin{center}
\begin{tabular}{|c|c|c|}
\hline
 Pleiades $\sigma = 16^\circ$ & Pleiades $\sigma = 20^\circ$ &
 Aldebaran $\sigma = 12^\circ$\\
\hline
2020-4-30 & 2020-4-23 & 2020-5-6 \\
\hline
   -1-4-6 & -1-4-1 & -1-4-14 \\
 -601-3-28 & -601-3-23 & -601-4-5 \\
-1001-3-23 & -1001-3-18 & -1001-4-1 \\
\hline
-1601-3-14 & -1601-3-10 &  -1601-3-24 \\
-1801-3-11 & -1801-3-6  & -1801-3-20 \\
-1941-3-9  & -1941-3-4  & -1941-3-19 \\
-2101-3-7 & -2101-3-2 & -2101-3-17 \\
\hline
\end{tabular}
\vspace{3mm}
\caption{Gregorian dates of the heliacal settings of the Pleiades and
  Aldebaran ($\alpha$ Tau). The heliacal sets in 2020 were calculated for Leipzig, Seebenischer Stra{\ss}e, all others for the Mittelberg.}
\label{tab5}
\end{center}
\end{table}
Regarding the question in which of the periods ''-2101 to -1601'' or
''-1001 to -1'' the heliacal settings of the Pleiades provide better
start dates for the farming year, it can now be said that the heliacal
settings in the time -1001 to -1 are not very suitable. They fall in
the second half of March or even April, which would be too late for a
sowing date. If one wants to see the purpose of the Disc, among other things, in determining the heliacal setting of the Pleiades for agriculture, then that would not work in the period -1001 to -1.

\section{Signs that the heliacal setting of the Pleiades is imminent
  or has just ended} \label{sec4}
The Bronze Age Nebra people could not yet predict the heliacal setting
of the Pleiades. They could only determine it through observation. To make
this work a little easier, the following additional information would
be helpful:
\begin{enumerate}
\item A sign that the heliacal setting of the Pleiades is imminent. If
  one does not have such a sign, one must observe the Pleiades daily
  for many more days than necessary in order to find the heliacal
  setting.
\item A sign that the heliacal setting of the Pleiades is definitely over. Such a sign is necessary if no observations were possible on the day of the heliacal set due to heavy cloud cover.
\end{enumerate}
In their observational practice, the Nebra people will certainly have sensed the need for such signs and looked for them.

\subsection{A sign that the heliacal setting of the Pleiades is
  definitely over}
At the setting of the Aldebaran ($\alpha$ Tau) one can see that the heliacal setting of the Pleiades is definitely over.
\mycomment{
\begin{figure}[!ht]
  \begin{center}
    \includegraphics[width=\textwidth]{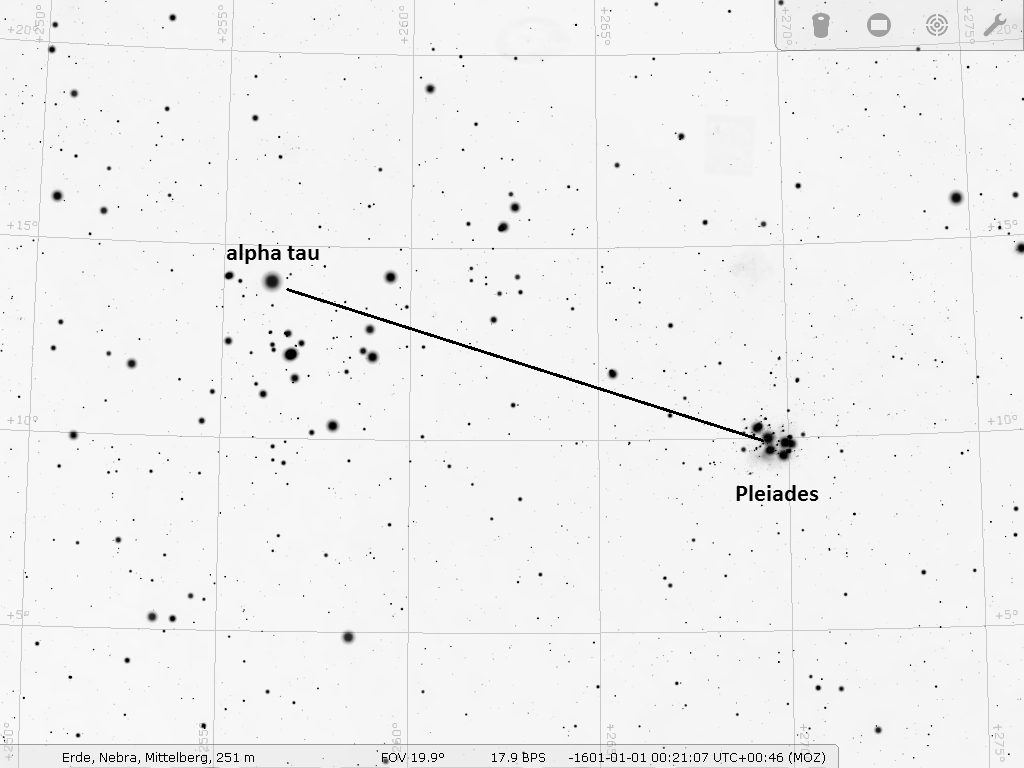}
\end{center}
\caption{The line from Aldebaran ($\alpha$ Tau) to the Pleiades ($\eta$ Tau) at -1601-1-1 above the Mittelberg when the Pleiades are $10^\circ$ high in the evening.} 
\label{lineAP}
\end{figure}
}

In the years -2101 to -1601 the Pleiades and the Aldebaran lay on an
almost horizontal line shortly before the setting of the
Pleiades. Figure \ref{lineAP} shows this line on -1601-1-1 on the
Mittelberg, when the Pleiades are only $10^\circ$ degrees high in the
evening. The Aldebaran then has a height of about $14^\circ$.

The magnitude (visual brightness) of the Aldebaran is 0.85 while the
brightest stars of the Pleiades have magnitudes of around 3. This
means that on setting the Pleiades disappear due to the extinction at
a higher height than the Aldebaran. During our own observations, the
Pleiades always disappeared at an altitude of $7^\circ$ to $8^\circ$,
while Aldebaran could still be seen up to $2^\circ$ altitude (see
Appendix \ref{F.2}, Tables \ref{tab18} and \ref{tab20}. According to
W. Schlosser, the Pleiades can be seen up to $5^\circ$. But even then,
Aldebaran can still be seen at lower altitudes.

Since Aldebaran is a little higher than the Pleiades and on the other hand can be seen a little longer than the Pleiades every day, it is a good control star for the setting of the Pleiades.
If the heliacal setting of the Pleiades cannot be observed due to cloudiness and a few days later, when the visibility is good again, the Aldebaran is already quite low when it appears after sunset (for example lower than $8^\circ$), then the heliacal setting of the Pleiades is over.
Since Aldebaran is pictured on the Sky Disc (Figure \ref{pflug2}), the
Nebra people had this star in their sights. They will probably have
used it as an aid to observing the heliacal setting of the
Pleiades. This is particularly possible because the heliacal setting
of Aldebaran occurs a sufficient number of days later than that of the
Pleiades (see Table \ref{tab5}).

\subsection{A sign that the heliacal setting of the Pleiades is
  imminent}
\subsubsection{A possible use of the ''Auriga Line''} \label{sec4.2.1}
The stars of the ''Auriga Line'' (see Subsection \ref{subsec2.3}) play an important role in a criterion for the imminence of the Pleiades' heliacal setting.
\mycomment{
\begin{figure}[!ht]
  \begin{center}
    \includegraphics[width=\textwidth]{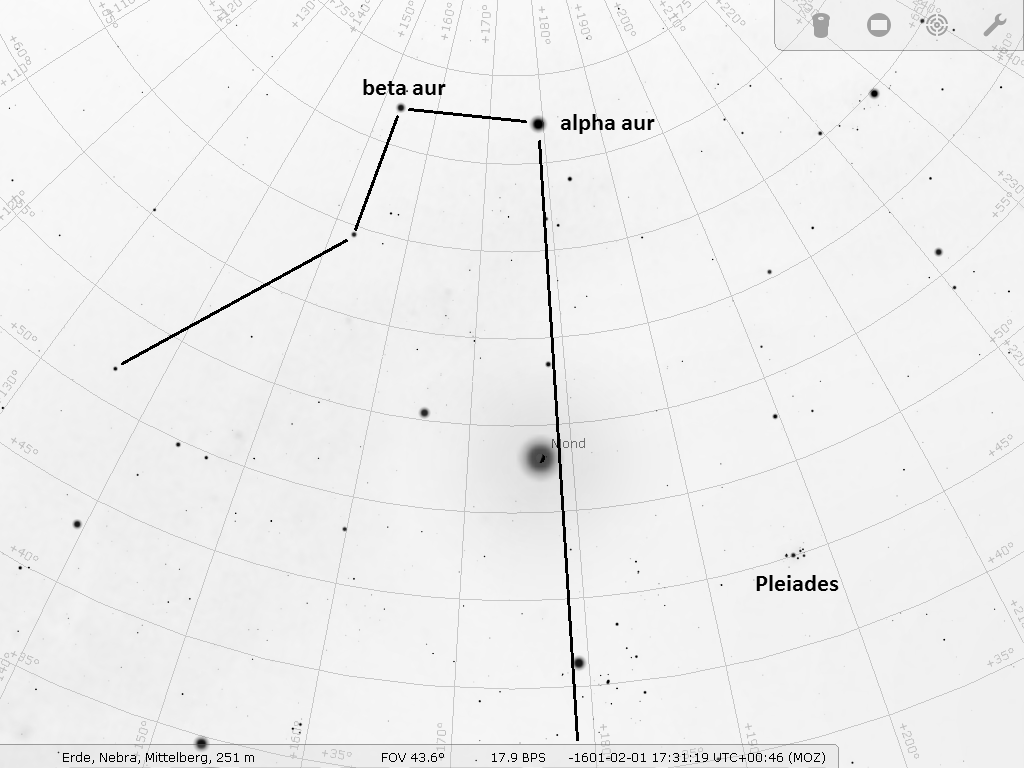}
\end{center}
\caption{-1601-02-01 Julian. Pleiades to the right of $\alpha$ Aur
  (Capella). Sun height $h = -9^\circ$.} 
\label{fig13}
\end{figure}
}
\mycomment{
\begin{figure}[!ht]
  \begin{center}
    \includegraphics[width=\textwidth]{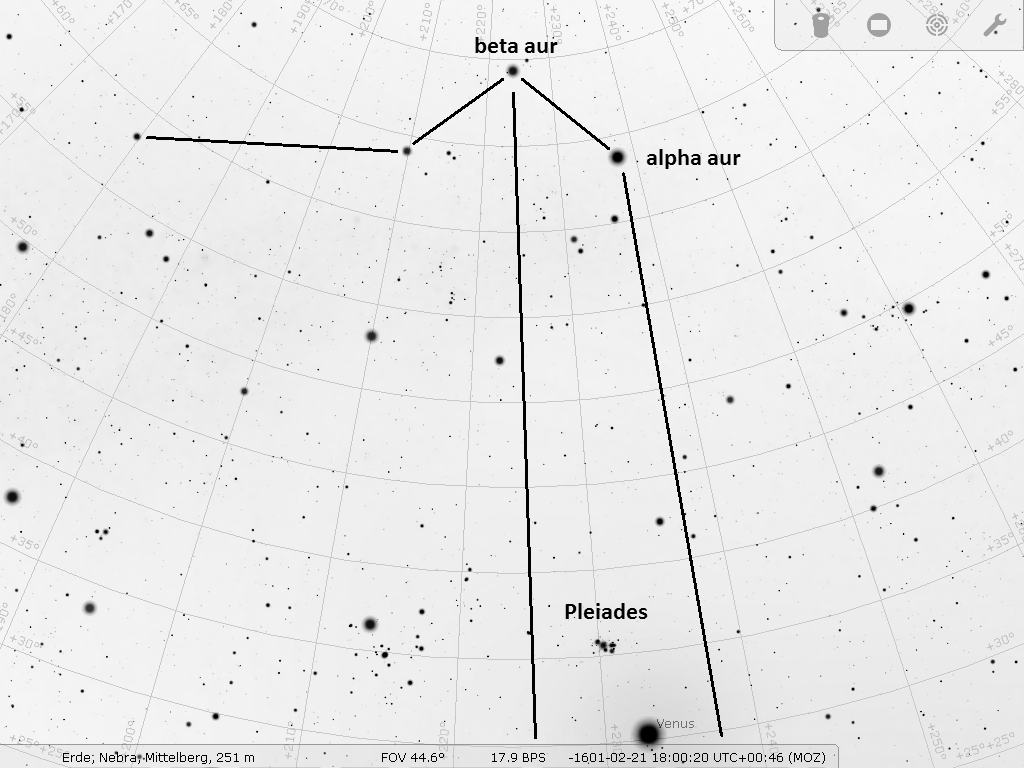}
\end{center}
\caption{-1601-02-21 Julian. Pleiades between $\beta$ Aur and $\alpha$
  Aur (Capella). Sun height $h = -9^\circ$.} 
\label{fig14}
\end{figure}
}
\mycomment{
\begin{figure}[!ht]
  \begin{center}
    \includegraphics[width=\textwidth]{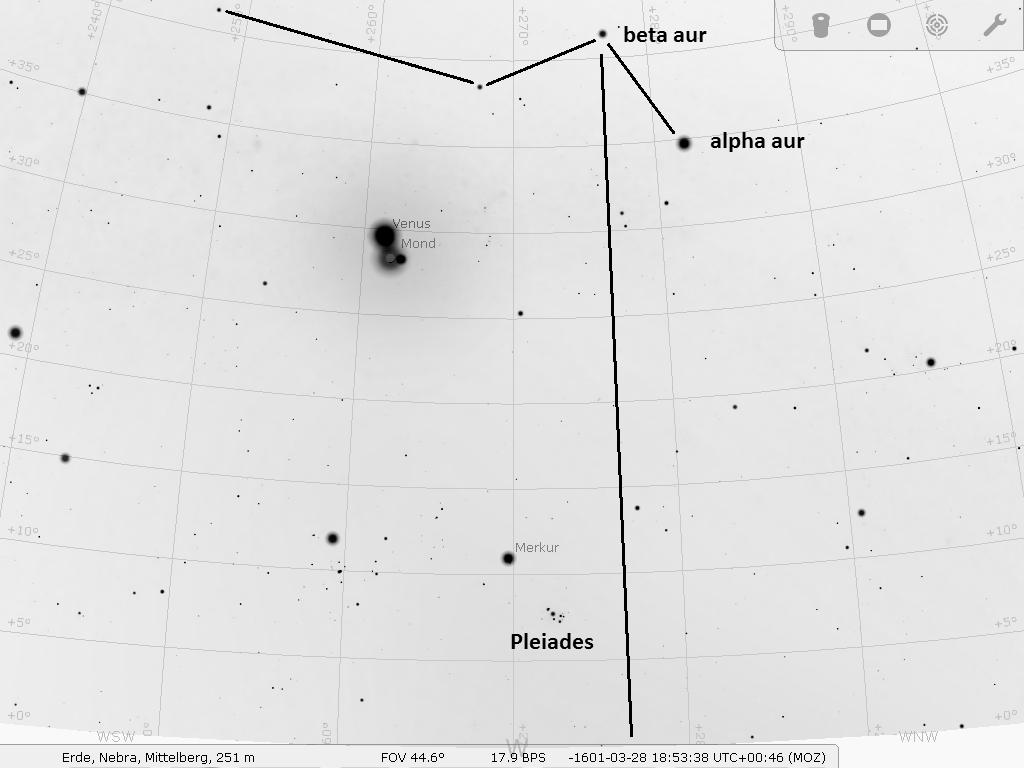}
\end{center}
\caption{-1601-03-28 Julian. Heliacal setting of the Pleiades at $\sigma =
  16^\circ$. Pleiades to the left of beta Aur. Sun height $h = -9^\circ$.} 
\label{fig15}
\end{figure}
}

The heliacal setting of the Pleiades always occurs at nautical
twilight after sunset. To find the date of this setting, the Nebra
people will have observed the sky at nautical twilight some time prior
to that. They will have noticed that the Pleiades are making a slow
movement under the stars of the Auriga line. Figures \ref{fig13}, \ref{fig14} and \ref{fig15}
show this movement in the year -1601. All three images were calculated
for the time when the sun is $-9^\circ$ below the horizon in the evening.

On February 1st Julian, the Pleiades are far to the right of the
perpendicular that runs from $\alpha$ Aur (Capella) to the horizon. On
February 21, the Pleiades will be between the perpendiculars of
$\alpha$ Aur and $\beta$ Aur to the horizon. And on March 28 Julian,
the day of their heliacal setting, the Pleiades are to the left of
$\beta$ Aur's perpendicular to the horizon.

This movement of the Pleiades under the stars of the Auriga Line is
only an apparent movement. In reality, the stars do not move against
each other in such a short time. But using $\alpha$ Aur or $\beta$ Aur
as fixed reference points gives the impression that the Pleiades are
moving along under these stars.

From the series of Figures \ref{fig13}-\ref{fig15} it follows that
there are also days when the Pleiades are exactly perpendicular under
$\alpha$ Aur or $\beta$ Aur in the nautical twilight. Of these days, the day when the Pleiades are perpendicular
under $\beta$ Aur in nautical twilight is particularly important. Only
from this day on there is a need to observe the sky daily to determine
the day of the heliacal setting of the Pleiades.

We believe that the Nebra people used the Auriga Line to find the
day when the Pleiades are perpendicular under $\beta$ Aur at nautical
twilight. This allowed them to greatly shorten the time they had to
observe the sky each day. It is now understandable why they moved the
star $\beta$ Aur on the Sky Disc when the horizon arcs were placed
(see Figure \ref{auxiliaryline2}). It was the most important star on
the Auriga Line and must not be lost at all during the forging
work.

We introduce the following concept:
\begin{Def}
Let $P$ be a given astronomical observing site. We call a day {\it beta day at $P$ to solar depth $h_0$} if
\begin{enumerate}
\item on that day the Pleiades appear from $P$ vertically below $\beta$ Aur when the Sun has a depth $h = h_0$ below the western horizon in the evening.
\item the Pleiades perform an apparent movement from right to left under $\beta$ Aur with respect to $\beta$ Aur.
\end{enumerate}
\end{Def}
$h = h_0$ is the sun depth at which the Pleiades become visible at nautical twilight.
\mycomment{
\begin{figure}[!ht]
  \begin{center}
    \includegraphics[width=\textwidth]{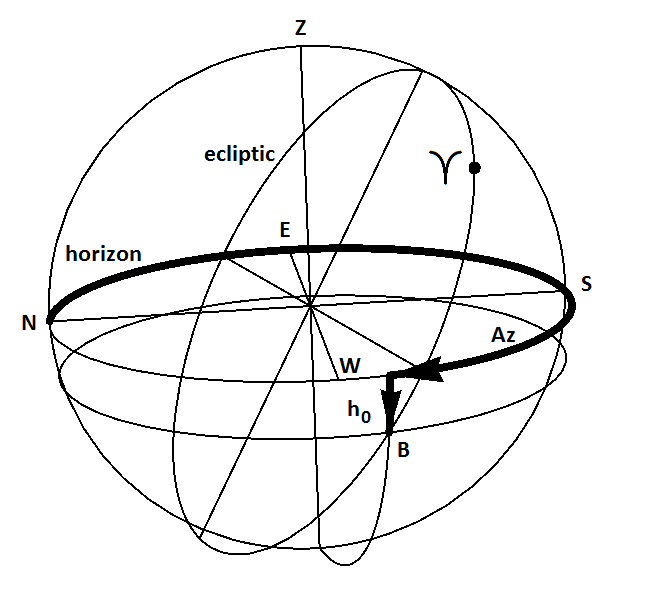}
\end{center}
\caption{Determination of the beta point.} 
\label{HorEklipt}
\end{figure}
}

\subsubsection{Details about beta days} \label{sec5.2.2}

\begin{table}[t] 
\begin{center}
\begin{tabular}{|c|c|c|c|c|}
\hline
 date & time & $\text{Az}_{\eta}$ & app. sid. time & description \\
\hline
-1601-1-1&  18:44:55&$  180^\circ0'$&      0:31:02 & transit
Pleiades\\
-1601-1-1& && 1:19:27 & sky view of Figure \ref{fig13} \\
-1601-1-1&     20:32:55&$  216^\circ36'$&  2:19:20 & $\text{Az}_{\eta} = \text{Az}_{\alpha}$\\
-1601-1-1& && 3:07:24 & sky view of Figure \ref{fig14} \\
-1601-1-1&     21:57:55&$  239^\circ44'$&  3:44:34 & $\text{Az}_{\eta} = \text{Az}_{\beta}$\\
-1601-1-1&     23:39:55&$  266^\circ55'$&  5:26:50 & hardly any
relative movement\\
&&&& to $\beta$ Aur (begin)\\
-1601-1-2& && 6:18:50 & sky view of Figure \ref{fig15} \\
-1601-1-2&      0:38:55&$  273^\circ42'$&  6:26:00 & hardly any
relative movement\\
&&&& to $\beta$ Aur (end)\\
-1601-1-2&      1:19:55&$  281^\circ41'$&  7:07:07 & set Pleiades\\
-1601-1-2&      2:46:55&$  299^\circ18'$&  8:34:21 & $\text{Az}_{\eta} = \text{Az}_{\beta}$\\
-1601-1-2&      4:27:55&$  322^\circ39'$&  10:15:30 & $\text{Az}_{\eta} = \text{Az}_{\alpha}$\\
-1601-1-2&     12:03:55&$   77^\circ55'$&  17:52:52 & rise Pleiades\\
-1601-1-2&     14:57:55&$  113^\circ14'$&  20:47:22 & hardly any
relative movement\\
&&&& to $\alpha$ Aur (begin)\\
-1601-1-2&     15:53:55&$  126^\circ47'$&  21:43:30 & hardly any
relative movement\\
&&&& to $\alpha$ Aur (end)\\
-1601-1-2&     18:40:59&$  180^\circ0'$&   0:31:02  & transit Pleiades\\
\hline
\end{tabular}
\vspace{3mm}
\caption{Important positions of the Pleiades during a full azimuth
  cycle on -1601-1-1 Jul. and -1601-1-2 Jul. on Mittelberg. ''Time''
  is Local Mean Solar Time. $\text{Az}_{\eta}$ = azimuth of $\eta$ Tau (Alcyone
/ Pleiades), $\text{Az}_{\alpha}$ = azimuth of $\alpha$ Aur
  (Capella), $\text{Az}_{\beta}$ = azimuth of $\beta$
Aur. ''app. sid. time'' means ''apparent sidereal time''.}
\label{tab7a}
\end{center}
\end{table}
\mycomment{
\begin{figure}[!ht]
  \begin{center}
    \includegraphics[width=\textwidth]{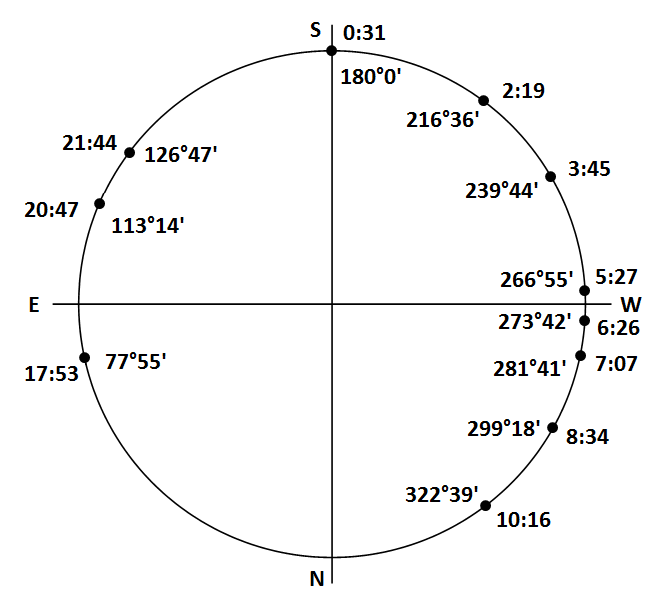}
\end{center}
\caption{Positions of the Pleiades from Table \ref{tab7a}.} 
\label{azimutkreis}
\end{figure}
}
Each position of the Pleiades shown in Figures \ref{fig13}-\ref{fig15}
is also traversed by the Pleiades on every single day within 24
hours. Table \ref{tab7a} contains important positions passed by the
Pleiades between their transit on -1601-1-1 and their subsequent
transit on -1601-1-2 (on the Mittelberg). This table was calculated
with {\tt Stellarium} \cite{zottihoffmann,zottiwolf} using the method
described in Appendix \ref{appE}.

If Az$_\eta$ = Az$_\alpha$ or Az$_\eta$ = Az$_\beta$ holds, the
Pleiades are perpendicular under $\alpha$ Aur (Capella) or $\beta$
Aur. ''Hardly any relative movement to $\beta$ Aur'' means
$$\text{Az}_\eta - \text{Az}_\beta \sim \text{const.} \;\;, \;\;\text{h}_\eta - \text{h}_\beta \sim \text{const.}$$
The stars $\eta$ Tau and $\beta$ Aur are moving, but seem to be at
rest with respect to each other. The analogous meaning has "hardly any
relative motion to $\alpha$ Aur".

"Set of the Pleiades" in Table \ref{tab7a} means that the Pleiades
have an altitude of $0^\circ$. In Figure \ref{fig15} the Pleiades
still have an altitude of $7^\circ$. Therefore the sky view of Figure
\ref{fig15} appears at an earlier sidereal time than the setting of
the Pleiades.

In the period 21:43:30 to 5:26:50 sidereal time the Pleiades perform a relative motion from right to left with respect to the stars $\alpha$ Aur and $\beta$ Aur and in the period 6:26:00 to 20:47:22 sidereal time a relative motion from left to right with respect to these stars.

Figure \ref{azimutkreis} shows the azimuth angles of $\eta$ Tau and sidereal times
at which the events listed in Table \ref{tab7a} occur at the observing
site. The points corresponding to Figures \ref{fig13}-\ref{fig15} have
not been plotted here. Note that the events between the setting of the
Pleiades and the rising of the Pleiades are not visible at the
observation site, because they take place below the horizon plane.

All events listed in Table \ref{tab7a} take place throughout the year at the apparent sidereal times given for them. Therefore the sequence of these events does not change.

Since sidereal time shifts by 3 minutes 56 seconds from solar time
each day, the events in Table \ref{tab7a} and Figure \ref{azimutkreis}
appear 3 minutes 56 seconds earlier each day than the previous
day. However, only a small fraction of these events were observed on
each day by the Nebra people, since we assume that they made their
observations when the Pleiades appeared at dusk at a solar depth
$h_0$. Thus a beta day only occurred when the time when the Pleiades
were perpendicular under $\beta$ Aur was very close to the time when the
Sun was at depth $h_0$ in the evening.

\begin{Def} \label{def5.2}
Let $P$ be a given astronomical observing site and $yr$ a given year. We consider on any day of this year in the horizon system of $P$ the coordinate circle $K_{h_0}$ of depth $h_0$ below the horizon of $P$ at the moment when
\begin{enumerate}
\item the Pleiades appear vertically under $\beta$ {Aur} from $P$.
\item the Pleiades perform an apparent motion from right to left under $\beta$ Aur with respect to $\beta$ Aur.
\end{enumerate}
Then we call the intersection $B$ of $K_{h_0}$ with the ecliptic, which has azimuth $180^\circ < Az_B < 360^\circ$, {\it beta point to depth $h_0$ at $P$} (see Figure \ref{HorEklipt}). The moment when conditions (1) and (2) are fulfilled, we call {\it beta event}.
\end{Def}
The condition $180^\circ < Az_B < 360^\circ$ guarantees that the intersection point under the west horizon is used as beta point and not the second also still existing intersection point under the east horizon.

It follows from the properties of the sidereal time (see Section \ref{sec3a}):
\begin{Prop}
  A view of the celestial sphere that satisfies conditions {\rm (1)} and {\rm (2)} in Def. {\rm\ref{def5.2}}, always takes place at the same sidereal time 
on every day of a given year.
\end{Prop}
Therefore, we can make the determination of the beta point for a given year $yr$ on any day of $yr$.

The synchronicity between beta event and sidereal time remains only for a certain time. The time of a beta event and the position of a beta point on the ecliptic depend on
\begin{itemize}
\item the location of the Earth's axis of rotation,
\item the proper motion of the pair of stars under consideration,
\item of the position of the ecliptic.
\end{itemize}
These three things perform slow motions in larger periods of time, which is why the sidereal time of a beta event and the position of a beta point shift with time. In section \ref{sec524} we will see that the sidereal time of a beta event already shows a shift of more than 1 minute after 25 years, while the stellar time of a beta event shifts by more than 1 minute after 500 years. Within one year, however, no shifts occur with certainty.

From Table \ref{tab7a} it follows that for an observation site $P$
there is only one beta point each year, because there is only one
event with $Az_{\eta} = Az_{\beta}$, which takes place above the
horizon. A beta day occurs when the sun is at the beta point during
its apparent course along the ecliptic. Thus also only one beta day happens each year.
\mycomment{
\begin{figure}[!ht]
  \begin{center}
    \includegraphics[width=\textwidth]{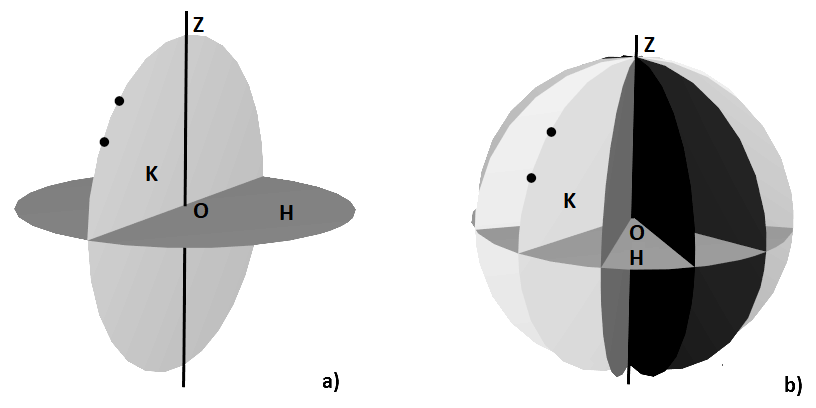}
\end{center}
\caption{Regarding the zenith criterion in Prop. \ref{prop5.4}.} 
\label{zenitkriterium}
\end{figure}
}

All previous considerations of this section were made only for $P =$ Mittelberg and the pair of stars $\beta$ Aur and $\eta$ Tau (Pleiades). We now clarify to what extent similar phenomena exist for arbitrary observing sites $P$ and arbitrary pairs of stars $S_1$, $S_2$. To do this, we first simplify the way of speaking.

\begin{Def}
  In the propositions \ref{prop5.4} to \ref{prop5.8} we only want to understand by the way of speaking "Two stars $S_1$, $S_2$ appear vertically above each other": $S_1$ and $S_2$ lie on a great circle of the celestial sphere whose plane is perpendicular to the plane of the horizon. 
\end{Def}

This way of speaking also includes the case in which one or both stars are below the horizon and thus the vertical positioning of the stars is not visible at all. Likewise, the case in which both stars lie on a great circle but are on two different sides of the zenith of the observation point. However, this simplifies our subsequent reasoning. One only has to be aware that when applying the propositions \ref{prop5.4} to \ref{prop5.8} one would also have to check whether one really sees the perpendicular positioning of the stars or whether it only appears in the form of the special cases listed here.

\begin{Prop} \label{prop5.4}
  Let $P$ be an observation site on Earth and $S_1$, $S_2$ two stars on the celestial sphere. We put a great circle $K$ of the celestial sphere through $S_1$ and $S_2$. Then it is valid: From $P$ $S_1$ and $S_2$ appear perpendicularly one above the other if and only if the zenith $Z$ of $P$ lies also on $K$.
\end{Prop}
\begin{proof}
  In the proof we rely on Figure \ref{zenitkriterium}.

  {\bf a)} Let $S_1$ and $S_2$ appear vertically above each other from $P$. Then $\vec{a} = \overrightarrow{OS_1}$, $\vec{b} = \overrightarrow{OS_2}$ span the circular plane of $K$ and the normal vector $\vec{n} = \vec{a} \times \vec{b}$ lies in the horizon plane $H$ of $P$.  Any vector starting at the coordinate origin $O$ and perpendicular to $\vec{n}$ lies in the plane of $K$. This is especially true for the vector $\overrightarrow{OZ}$ from $O$ to the zenith $Z$ of $P$, because it is perpendicular to $H$, thus also to $\vec{n}$. Thus $Z$ lies on the great circle $K$. (See Fig. \ref{zenitkriterium} a.)

{\bf b)}  Let $Z$ lie on $K$. All great circles containing $\overrightarrow{OZ}$ form the set of all altitude circles of the horizon system of $P$. They are all perpendicular to $H$. Since $K$ also contains the vector $\overrightarrow{OZ}$, $K$ is also perpendicular to $H$. Consequently, $S_1$ and $S_2$ appear perpendicular to each other from $P$. (See Fig. \ref{zenitkriterium} b.)
\end{proof}

\mycomment{
\begin{figure}[!ht]
  \begin{center}
    \includegraphics[width=\textwidth]{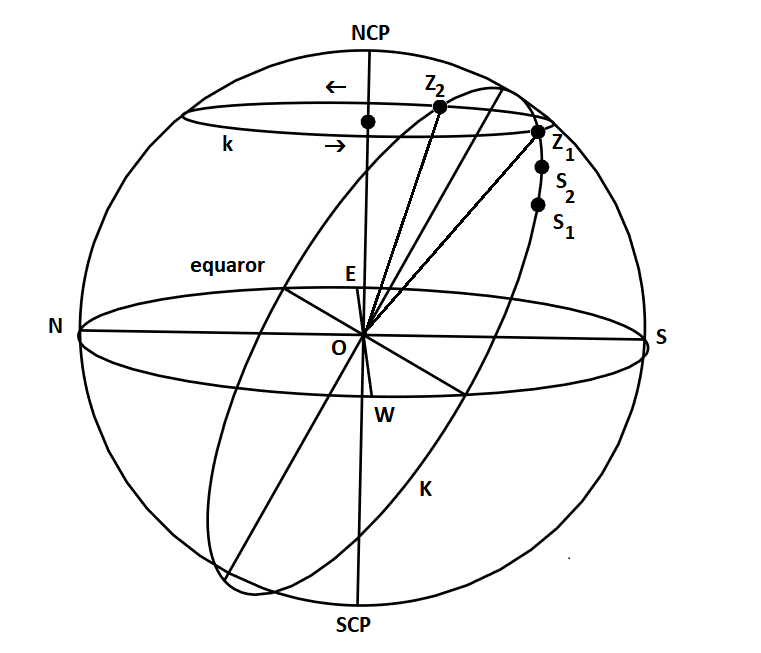}
\end{center}
\caption{} 
\label{zenit}
\end{figure}
}
\begin{Prop}
  For each observation point $P$, which is not on the earth equator or on one of the earth poles, and each pair of stars $S_1$ and $S_2$ exactly one of the following {\rm 3} possibilities occurs at each earth rotation :
\begin{itemize}
\item $S_1$ and $S_2$ appear from $P$ {\rm 2} times vertically above each other. 
\item$S_1$ and $S_2$ appear from $P$ {\rm 1} times vertically one above the other. 
\item $S_1$ and $S_2$ never appear vertically one above the other from $P$.
\end{itemize}
\end{Prop}

\begin{proof}
  If $P$ is not on the equator or on one of the Earth's poles, then its zenith $Z$ during the daily rotation of the Earth runs along a small circle $k$ whose circular plane is parallel to the equatorial plane (see Fig. \ref{zenit}). Then, if the great circle $K$ is inclined steeply enough, it will intersect the small circle $k$ at exactly two points $Z_1$ and $Z_2$. $S_1$ and $S_2$ appear perpendicularly one above the other from $P$ exactly when the zenith $Z$ of $P$ passes through the points $Z_1$ and $Z_2$.

With a smaller inclination of $K$, $K$ touches the circle $k$ only in a single point. Then $S_1$ and $S_2$ appear vertically one above the other from $P$ only once in 24 hours. If the inclination of $K$ is even smaller, then $K$ and $k$ no longer intersect, i.e. $S_1$ and $S_2$ never appear vertically from $P$.
\end{proof}

\begin{Prop}
Let $P$ be one of the Earth's poles and $S_1$ and $S_2$ be any two fixed stars. Then $S_1$ and $S_2$ either never appear or always appear vertically above each other from $P$.
\end{Prop}

\begin{proof}
  If the observation point $P$ is located at one of the Earth's poles, then the small circle $k$ becomes the associated celestial pole. If then the great circle $K$ does not pass through the celestial poles, then $S_1$ and $S_2$ never appear vertically above each other from $P$, since $Z$ never lies on $K$. If $K$ passes through the celestial poles, then $S_1$ and $S_2$ always appear vertically one above the other, because $Z$ lies on $K$ during an entire rotation of the Earth.
\end{proof}

\begin{Prop} \label{prop5.8}
  Let $P$ be a point on the Earth's equator and $S_1$ and $S_2$ be any two fixed stars. Then $S_1$ and $S_2$ appear either twice or always vertically one above the other from $P$.
\end{Prop}

\begin{proof}
  If $P$ lies on the Earth's equator, then during one revolution of the Earth its zenith also runs on a great circle, namely on the celestial equator. If $K$ is inclined with respect to the celestial equator, then $K$ and the celestial equator intersect at exactly two points and thus $S_1$ and $S_2$ appear vertically above each other twice during one revolution of the Earth. If, on the other hand, $K$ and the celestial equator coincide, then $S_1$ and $S_2$ always appear vertically one above the other, because the zenith $Z$ is always on $K$.
\end{proof}

\subsubsection{Calculation of beta days for the Bronze Age and the Iron Age}
\begin{table}[t] 
\begin{center}
\begin{tabular}{|c|c|c|c|c|c|c|}
\hline
\multicolumn{7}{|l|}{Sun h $= -9^\circ$} \\
\hline
 date (Jul.) & time & \multicolumn{2}{|c|}{$\eta$ Tau (Pleiades)}&\multicolumn{2}{|c|}{$\beta$ Aur} & app. sidereal time \\
\hline
 & & Az & h & Az & h & \\
\hline
 -601-2-28& 18:22:20&$  241^\circ8'$&$   36^\circ38'$&$     240^\circ47'$&$   70^\circ32'$&
 4:26:56\\
-1601-2-28& 18:10:58   &$240^\circ0'$   &$30^\circ44'$     &$240^\circ6'$
&$64^\circ38'$&    3:45:39 \\
-1801-2-28& 18:08:42   &$239^\circ45'$  &$29^\circ32'$     &$239^\circ46'$   &$63^\circ26'$&    3:37:26\\
-1941-2-28& 18:07:09   &$239^\circ35'$  &$28^\circ41'$     &$239^\circ32'$   &$62^\circ36'$&    3:31:44\\
-2101-3-1&  18:06:49   &$240^\circ39'$  &$27^\circ0'$      &$241^\circ3'$    &$60^\circ54'$&    3:30:36\\
\hline
\multicolumn{7}{|l|}{Sun h $= -12^\circ$} \\
\hline
 date (Jul.) & time & \multicolumn{2}{|c|}{$\eta$ Tau (Pleiades)}&\multicolumn{2}{|c|}{$\beta$ Aur} & app. sidereal time \\
\hline
 & & Az & h & Az & h & \\
\hline
 -601-2-25& 18:36:55   &$241^\circ50'$  &$36^\circ15'$     &$241^\circ54'$   &$70^\circ9'$ &    4:29:44\\
-1601-2-24& 18:24:24   &$239^\circ26'$  &$31^\circ2'$      &$239^\circ16'$   &$64^\circ56'$ &   3:43:21\\
-1801-2-24& 18:22:15   &$239^\circ13'$  &$29^\circ50'$     &$238^\circ59'$   &$63^\circ44'$ &   3:35:15\\
-1941-2-25& 18:22:12   &$240^\circ21'$  &$28^\circ15'$     &$240^\circ39'$   &$62^\circ9'$  &   3:34:59\\
-2101-2-25& 18:20:27   &$240^\circ10'$  &$27^\circ17'$     &$240^\circ21'$   &$61^\circ11'$  &  3:28:30\\
\hline
\end{tabular}
\vspace{3mm}
\caption{Julian dates of the beta day. We determined this for the Mittelberg when the sun is
  $-9^\circ$ or $-12^\circ$ below the horizon in the evening. Az and h
  denote the azimuth and the height of a celestial body. ''Time'' is Local Mean Solar Time.}
\label{tab8}
\end{center}
\end{table}
\begin{table}[t] 
\begin{center}
\begin{tabular}{|c|c|}
\hline
 Gregorian dates, sun's height $-9^\circ$ & Gregorian dates, sun's height $-12^\circ$ \\
\hline
 -601-2-21&  -601-2-18\\
-1601-2-14& -1601-2-10\\
-1801-2-12& -1801-2-8\\
-1941-2-11& -1941-2-8\\
-2101-2-11&  -2101-2-7\\
\hline
\end{tabular}
\vspace{3mm}
\caption{Gregorian dates of the beta day. We determined this for the Mittelberg when the sun is
  $-9^\circ$ or $-12^\circ$ below the horizon in the evening.}
\label{tab9}
\end{center}
\end{table}

\begin{table}[t] 
\begin{center}
\begin{tabular}{|c|c|c|c|c|c|}
\hline
 year & -601 & -1601 & -1801 & -1941 & -2101\\
\hline
D & 35 & 28 & 27 & 26 & 24 \\
\hline
\end{tabular}
\vspace{3mm}
\caption{Number D of days after which the heliacal setting with $\sigma = 16^\circ$ follows the beta day with $h=-9^\circ$.}
\label{tab10c}
\end{center}
\end{table}

If we use {\tt Stellarium} \cite{zottihoffmann, zottiwolf}, we can compute
beta days with a procedure similar to the procedure for heliacal
settings in Appendix \ref{appD}. We give the detailed version of this
procedure in Appendix \ref{appC3}. The main steps of the procedure are
as follows:\\*[0.3cm]
We want to determine the beta day for a solar depth $h_0$.
\begin{enumerate}
\item Select any day in the middle of the year for which the beta day
  is to be calculated.
\item Rotate the celestial sphere so that the Pleiades ($\eta$ Tau)
  are perpendicular under $\beta$ Aur (steps (6), (7), (8) in
  \ref{appC3}).
\item With this positioning of the celestial sphere, determine the
  intersection $B$ of the ecliptic with the $h_0$ line in the
  observer's horizon system. $B$ is an approximation of the beta point
  (see Figure \ref{HorEklipt}).
\item Move the sun on the ecliptic by changing the date and time until
  it is above the point $B$. The day when this happens is the beta day.
\end{enumerate}

Using this procedure, described in more detail in \ref{appC3}, we
determined Table \ref{tab8}. Besides the Julian dates and the time of appearance of $\eta$ Tau vertically below $\beta$ Aur we have also given the azimuth and altitude of the two stars, so that one can see that they have approximately the same azimuth (which must be the case for a perpendicular position),
and that both stars are above the horizon.

The solar depth $h = -9^\circ$ is the most likely value for $h$ at
which the Pleiades become visible at nautical twilight in the
evening. This agrees with the $h$-values for solar depth found in my
own observations of the Pleiades (see Appendix \ref{F.2}, Tables
\ref{tab15}, \ref{tab16}, \ref{tab17},
values of the solar depth $h$ in the first observation of $\eta$
Tau). The second part of Table \ref{tab8} shows beta days that would result from an appearance of the Pleiades at $h = -12^\circ$.

Table \ref{tab9} shows the Gregorian dates of the beta days from Table
\ref{tab8}, which we calculated again using Table \ref{tab7} in
Appendix \ref{appA}.

If we compare Tables \ref{tab4} and \ref{tab8}, we can determine the
number of days which must still pass after the beta day with $h = -9^\circ$ until the heliacal setting of the Pleiades with the arcus visionis $\sigma = 16^\circ$ occurs. Table \ref{tab10c} shows the result.

Here $h = -9^\circ$ and $\sigma = 16^\circ$ fit together. Because the Pleiades already disappear in about $7^\circ$ altitude due to extinction at their setting, the sun also has a depth of $h = -9^\circ$ at the heliacal setting of the Pleiades
with $\sigma = 16^\circ$.

The value $h = -12^\circ$ is more consistent with $\sigma = 20^\circ$ (for Schlosser's calculations), because if the Pleiades disappeared at a height of $7^\circ$ the sun would be at a depth of $h = -13^\circ$.

\subsubsection{The second ''Auriga Line'' on the Sky Disc} \label{sec524a}

\mycomment{
\begin{figure}[!ht]
  \begin{center}
    \includegraphics[width=\textwidth]{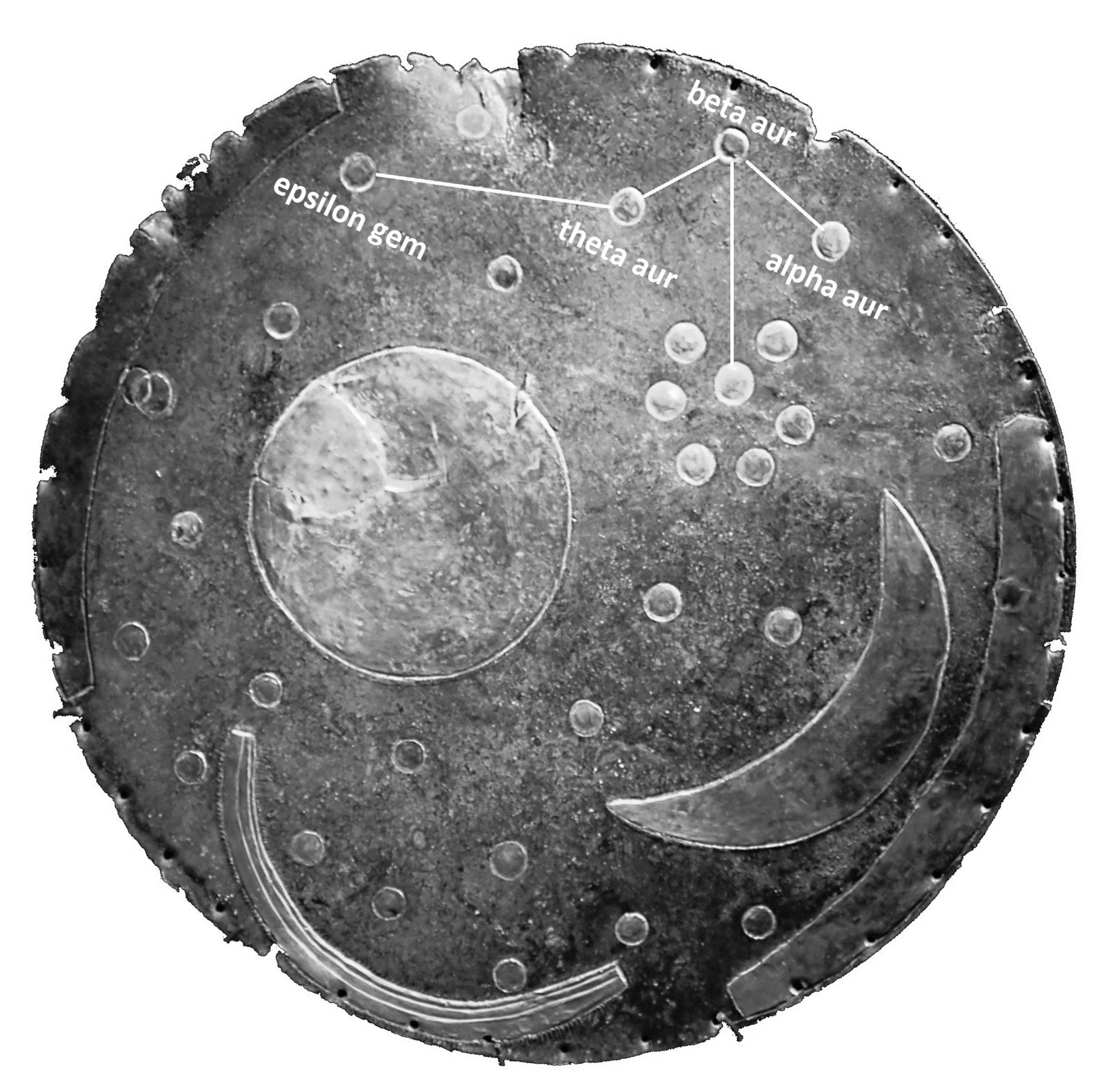}
\end{center}
\caption{The second Auriga line on the Sky Disc} 
\label{AuxLine4}
\end{figure}
}
\mycomment{
\begin{figure}[!ht]
  \begin{center}
    \includegraphics[width=\textwidth]{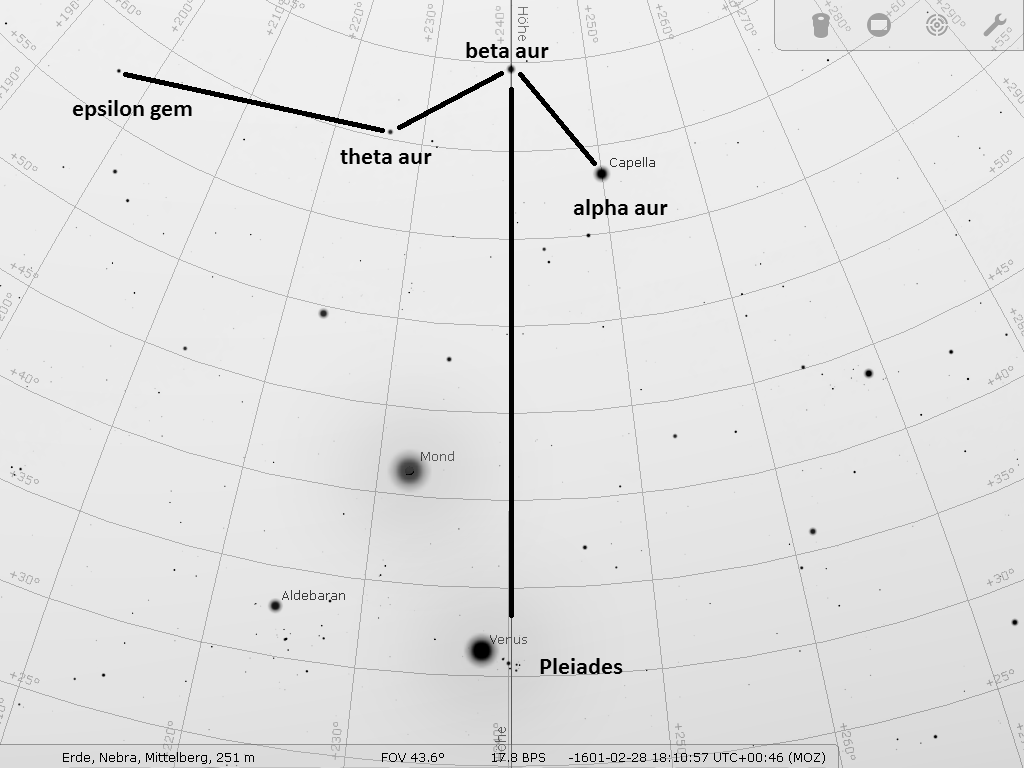}
\end{center}
\caption{The skewness of the Auriga line when the Pleiades are
  perpendicular to $\beta$ Aur in the Bronze Age on the Mittelberg.} 
\label{vertical}
\end{figure}
}

There is a hint on the Sky Disc that the purpose of the Disc is actually to determine beta days. It appears that the Auriga line is shown a second time on the disc (see Fig. \ref{AuxLine4}).

If we now rotate the Sky Disc so that the middle star of the Pleiades object on the Disc is perpendicular under the star $\beta$ Aur, we get an inclination of this Auriga line which is very similar to that which occurs when the Pleiades are perpendicular under $\beta$ Aur in the Bronze Age on the Mittelberg (see Fig. \ref{vertical}). In Fig. \ref{AuxLine4} we have made such a rotation. To all appearances, the Nebra people mapped the view of the sky that is present on a beta day onto the Sky Disc.

In the second representation of the Auriga line, the distance lengths between the stars also correspond more closely to reality.

We made length and angle measurements in the Figures \ref{AuxLine4} and \ref{vertical} using GIMP \cite{gimp}. Tables \ref{tab11} and \ref{tab12a} show the results. Furthermore, by drawing a test line, from $\epsilon$ Gem to $\alpha$ Aur, one can determine that $\theta$ Aur lies on the line $\epsilon$ Gem -- $\alpha$ Aur in both figures.

\begin{table}[t] 
\begin{center}
\begin{tabular}{|l|l|l|}
\hline
 line & length (pixels) & relative length\\
\hline
$\beta$ Aur -- $\alpha$ Aur & 198,0 & 1,0 \\
$\beta$ Aur -- $\theta$ Aur & 177,0 & 0,89393939 \\
$\epsilon$ Gem -- $\theta$ Aur & 396,0 & 2,0 \\
$\epsilon$ Gem -- $\alpha$ Aur & 691,0 & 3,4898989 \\
\hline
\multicolumn{2}{|l|}{$\sphericalangle(\alpha \text{ Aur} \rightarrow\epsilon\text{ Gem} \,,\, \beta\text{ Aur}\rightarrow\text{center Pleiades})$} & $98,67^\circ$\\
\hline
\end{tabular}
\vspace{3mm}
\caption{Measurements in Figure \ref{AuxLine4}. The symbol $\sphericalangle(A\rightarrow B\,,\,C\rightarrow D)$ means the angle between the vector from A to B and the vector from C to D.}
\label{tab11}
\end{center}
\end{table}

\begin{table}[t] 
\begin{center}
\begin{tabular}{|l|l|l|}
\hline
 line & length (pixels) & relative length\\
\hline
$\beta$ Aur -- $\alpha$ Aur & 139,8 & 1,0 \\
$\beta$ Aur -- $\theta$ Aur &  135,1 &  0,966381 \\
$\epsilon$ Gem -- $\theta$ Aur & 279,2 & 1,997139 \\
$\epsilon$ Gem -- $\alpha$ Aur & 496,0 & 3,547926 \\
\hline
\multicolumn{2}{|l|}{$\sphericalangle(\alpha \text{ Aur} \rightarrow\epsilon\text{ Gem} \,,\, \beta\text{ Aur}\rightarrow\eta\text{ Tau})$} & $101,86^\circ$\\
\hline
\end{tabular}
\vspace{3mm}
\caption{Measurements in Figure \ref{vertical}. The symbol $\sphericalangle(A\rightarrow B\,,\,C\rightarrow D)$ means the angle between the vector from A to B and the vector from C to D.}
\label{tab12a}
\end{center}
\end{table}

The relative lengths on the Disc and on our present-day star map agree amazingly well. Also the deviation of the inclination angle of the line $\epsilon$ Gem -- $\alpha$ Aur in both representations is not very large. This is actually only possible if the distances shown on the Disc were measured in the sky. We therefore formulate:

\begin{Con} \label{conj5.8}
 The Nebra people were able to measure lengths and angles of inclination at the ''Auriga line'' and transfer them to the Sky Disc true to scale.
\end{Con}

\begin{Cor} \label{cor5.10}
  Figure {\rm\ref{AuxLine4}} shows an image of the Pleiades at the moment when they are vertically below $\beta$ Aur. The representation of the Auriga line is to scale.
\end{Cor}

We have deliberately formulated the statement \ref{conj5.8} as a conjecture because there are still uncertainties.
\begin{enumerate}
  \item The measurements for Tables \ref{tab11} and \ref{tab12a} were made quickly. Each quantity was measured only once and the location of the centers of the small star plates was only estimated.
\item  We do not know exactly whether the image of the Sky Disc we used contains distortions, e.g. because the camera was tilted in relation to the Sky Disc.
\end{enumerate}

When we have better image material, we will repeat the above measurements again.
But in our present paper we also make a first attempt in Appendix \ref{appB} to obtain more precise values than those in Tables \ref{tab11} and \ref{tab12a} using better mathematical methods.
Appendix \ref{appB} confirms the values in Tables \ref{tab11} and \ref{tab12a}. In particular, the correctness of Corrolary \ref{cor5.10} becomes very likely.

How could the Nebra people, who had no writing and thus probably no number signs, measure lengths in the sky? One possibility would be to use rods that were as straight as possible. We think that reed stalks, for example, are suitable for this.

To determine the distance between two stars, hold a stalk against the sky with both hands and arms outstretched. The hands should be far enough apart so that both stars fit between them along the length of the stick. Align one end of the stick with one of the stars. Then move the other hand towards the second star until the thumb is at the second star. You must now check that the end of the stick is still at the first star. If it is, then you can mark the position of the second star with the thumbnail and cut the stalk there.

With a little practice, this procedure should work. It is important that the arms are outstretched so that there is always the same distance between the culm and the eyes in all such measurements. And, of course, all measurements on one object should be taken by one and the same person.

When in this way sticks have been obtained for all the distances given in the Tables \ref{tab11}, \ref{tab12a}, one can relate the sticks to each other. One can easily prove that $\epsilon$ Gem -- $\theta$ Aur is twice as long as $\beta$ Aur -- $\alpha$ Aur.

The relative length of $\epsilon$ Gem -- $\alpha$ Aur is obtained by placing a two-stick and two one-sticks next to the stick of $\epsilon$ Gem -- $\alpha$ Aur. On the second stick of ones, which overhangs $\epsilon$ Gem -- $\alpha$ Aur, you mark the end of $\epsilon$ Gem -- $\alpha$ Aur.

Then turn this stick over and place it with the other end on the row of sticks next to $\epsilon$ Gem -- $\alpha$ Aur. The mark on this stick comes to rest again next to the end of $\epsilon$ Gem -- $\alpha$ Aur. Thus, the mark is in the middle of this stick and $\epsilon$ Gem -- $\alpha$ Aur has relative length 3,5.

The Nebra people could probably only estimate the angle of inclination of the line $\epsilon$ Gem - $\alpha$ Aur. However, the results of Appendix \ref{appB} give reason to hope that the Nebra people were also able to measure this angle of inclination.

When creating a reduced image of the ''Auriga line'' one can rely on the fact that all relative lengths 1, 2, 3.5 can be built up from the relative length 0.5. For example, a stick of relative length 3.5 is created by joining 7 sticks of relative length 0.5. For the line $\theta$ Aur -- $\beta$ Aur one could use a stick of relative length 1, which was shortened a little.

The Nebra people only had to find such a length for a stick of relative length 0.5 that all the longer sticks formed from it and the image of the Auriga line composed of them fit on the Disc. Then they could make the sticks and put the image of the Auriga line on the Disc.

\subsubsection{With what accuracy can a person estimate that a star is perpendicular to another star?}
\begin{table}[t] 
\begin{center}
\begin{tabular}{|c|c|c|c|c|c|}
\hline
 time & \multicolumn{2}{|c|}{$\alpha$ Tau} &
 \multicolumn{2}{|c|}{$\beta$ Tau} & $\text{Az}_{\alpha} - \text{Az}_{\beta}$ in degrees\\
\hline
 & Az & h & Az & h & \\
\hline
\multicolumn{6}{|c|}{2020-3-24, Seebenischer Str.} \\
\hline
22:12 CET&     $275^\circ51'$&  $16^\circ39'$&    $275^\circ26'$&  $33^\circ24'$&   275,85 - 275,4 = 0,45\\
23:29 CET&     $290^\circ27'$&  $4^\circ56'$&     $289^\circ12'$&  $21^\circ39'$&   290,45 - 289,2 = 1,25\\
23:46 CET&     $293^\circ43'$&  $2^\circ28'$&     $292^\circ11'$&  $19^\circ9'$&    293,717 - 292,183 = 1,533\\
\hline
\multicolumn{6}{|c|}{2020-4-5, Sch\"onauer Str} \\
\hline
20:47 CEST&  $256^\circ8'$&   $31^\circ52'$&   $255^\circ4'$&   $48^\circ36'$&   256,133 -
255,067 = 1,067 \\
21:14 CEST& $261^\circ53'$&  $27^\circ43'$&   $261^\circ16'$&  $44^\circ28'$&   261,883 - 261,267 = 0,617\\
22:02 CEST& $271^\circ29'$&  $20^\circ13'$&   $271^\circ6'$&   $36^\circ58'$&   271,483 - 271,1 = 0,383\\
\hline
\multicolumn{6}{|c|}{2020-4-26, Seebenischer Str.} \\
\hline
21:47 CEST&   $284^\circ20'$&  $9^\circ45'$&   $283^\circ32'$&  $26^\circ30'$&   284,333 - 283,533 = 0,8\\
22:36 CEST&   $293^\circ40'$&  $2^\circ30'$&   $292^\circ8'$&   $19^\circ11'$&   293,667 - 292,133 = 1,533\\
\hline
\end{tabular}
\vspace{3mm}
\caption{Azimuth differences between $\alpha$ Tau (Aldebaran) and $\beta$ Tau during our own observations.}
\label{tab10a}
\end{center}
\end{table}
\begin{table}[t] 
\begin{center}
\begin{tabular}{|c|c|c|c|c|c|}
\hline
 time & \multicolumn{2}{|c|}{$\eta$ Tau} &
 \multicolumn{2}{|c|}{$\alpha$ Aur} & $\text{Az}_{\eta} - \text{Az}_{\alpha}$ in degrees\\
\hline
 & Az & h & Az & h & \\
\hline
\multicolumn{6}{|c|}{2020-3-24, Seebenischer Str.} \\
\hline
22:12 CET&     $289^\circ58'$&  $15^\circ6'$&      $294^\circ0'$&   $43^\circ11'$&   289,967 - 294,0 = -4,033\\
22:50  CET&    $296^\circ53'$&  $9^\circ38'$&      $299^\circ10'$&  $37^\circ52'$&   296,883 - 299,167 = -2,283\\
22:56  CET&    $297^\circ59'$&  $8^\circ48'$&      $299^\circ59'$&  $37^\circ3'$&    297,983 - 299,983 = -2,0\\
\hline
\multicolumn{6}{|c|}{2020-4-5, Sch\"onauer Str} \\
\hline
21:14 CEST&  $277^\circ2'$&   $25^\circ52'$&   $284^\circ12'$&  $53^\circ39'$&   277,033 - 284,2 = -7,167\\
22:02 CEST&  $285^\circ50'$&  $18^\circ30'$&   $290^\circ53'$&  $46^\circ29'$&   285,833 - 290,883
= -5,05\\
\hline
\multicolumn{6}{|c|}{2020-4-26, Seebenischer Str.} \\
\hline
21:47 CEST&   $298^\circ7'$&   $8^\circ42'$&      $300^\circ5'$&   $36^\circ57'$&  298,117 - 300,083 = -1,967\\
21:57 CEST&   $299^\circ58'$&  $7^\circ20'$&      $301^\circ27'$&  $35^\circ36'$&  299,967
- 301,45 = -1,483$^*$\\
\hline
\end{tabular}
\vspace{3mm}
\caption{Azimuth differences between $\eta$ Tau (Alcyone/Pleiades) and
  $\alpha$ Aur (Capella) during our own observations. $^*$ This value was not seen in observations, since by this time the Pleiades had already set.}
\label{tab10b}
\end{center}
\end{table}
\begin{table}[t] 
\begin{center}
\begin{tabular}{|c|c|c|c|c|}
\hline
 date (Jul.) & time & Az $\eta$ Tau (Pleiades) & Az $\beta$ Aur &
 $\text{Az}_{\eta} - \text{Az}_{\beta}$ in degrees\\
\hline
-1601-2-19&  18:16:59    &$232^\circ37'$       &$228^\circ9'$&   4,467\\
-1601-2-20&  18:18:28    &$234^\circ1'$                &$230^\circ35'$&  3,433\\
-1601-2-21&  18:19:57    &$235^\circ24'$               &$232^\circ54'$&  2,5 \\  
{\bf -1601-2-22}&  18:21:26  &$236^\circ46'$
&$235^\circ7'$&     {\bf 1,65}\\
{\bf -1601-2-23}&  18:22:54    &$238^\circ7'$
&$237^\circ14'$&  {\bf 0,883}\\
{\bf -1601-2-24}&  18:24:24    &$239^\circ26'$
&$239^\circ16'$&  {\bf 0,167}\\
{\bf -1601-2-25}&  18:25:52    &$240^\circ45'$
&$241^\circ13'$&  {\bf -0,467}\\
{\bf -1601-2-26}&  18:27:21    &$242^\circ3'$
&$243^\circ5'$&   {\bf -1,033}\\
{\bf -1601-2-27}&  18:28:51  &$243^\circ19'$
&$244^\circ52'$&    {\bf -1,55}\\
-1601-2-28&  18:30:21    &$244^\circ35'$               &$246^\circ36'$&  -2,017\\
-1601-3-1&   18:31:50    &$245^\circ50'$             &$248^\circ16'$& -2,433\\
\hline
\end{tabular}
\vspace{3mm}
\caption{Azimuth deviations between $\eta$ Tau and $\beta$ Aur on the
  days around -1601-2-24 Julian on the Mittelberg, when the sun is
  $-12^\circ$ below the horizon in the evening. ''Time'' is Local Mean Solar Time.}
\label{tab12}
\end{center}
\end{table}

We now turn to the question of how accurately the Nebra people could determine whether or not the Pleiades are perpendicular under $\beta$ Aur. Although we saw in the previous section that the Nebra people probably also made measurements, we will assume that they only estimated this, since the Disc contains no evidence that measuring devices were used for this purpose.

Two stars $S_1$ and $S_2$ appear to an observer at a time $t_0$ as standing vertically one above the other if they have the same azimuth angle $$\alpha_1 = \alpha_2$$ in the observer's horizon system at time $t_0$. Based on our own observations, we can say something about the accuracy with which a person can determine whether $\alpha_1 = \alpha_2$ .

We made 2020 on the days 24.3., 5.4. and 26.4. with the naked eye observations of the surroundings of the Pleiades and the constellations Taurus and Auriga (see Tables \ref{tab18}, \ref{tab19}, \ref{tab20}). During all these observations we always had the impression that
\begin{enumerate}
\item $\alpha$ Tau is always vertically below $\beta$ Tau,
\item $\eta$ Tau (Pleiades) is always to the left of the perpendicular of $\alpha$ Aur on the horizon.
\end{enumerate}
At the time of these observations we did not know that we would have to judge how accurately one can estimate whether stars are vertically aligned or not.

From the observation times given in Tables \ref{tab18}, \ref{tab19}, \ref{tab20} we have now calculated the azimuth angles of the respective celestial bodies and the azimuth differences using Stelarium \cite{zottihoffmann, zottiwolf}. Tables \ref{tab10a} and \ref{tab10b} show the results.

From the Tables \ref{tab10a}, \ref{tab10b} and statements (1), (2) about the observed stars follows immediately
\begin{Prop} \label{prop5.10}
Let $S_1$ and $S_2$ be two stars, which at a time $t_0$ have the
azimuths $\alpha_1$ and $\alpha_2$ and $S_2$ stand at least $15^\circ$
higher than $S_1$. Then there holds true:
\begin{enumerate}
\item To a person, $S_1$ appears to be perpendicular below $S_2$ if
$$|\alpha_1 - \alpha_2| < 1.6^\circ\;.$$
\item To a person, $S_1$ does not appear to be perpendicular below
  $S_2$ if
$$|\alpha_1 - \alpha_2| > 2^\circ\;.$$
\end{enumerate}
\end{Prop}

We now apply Prop. \ref{prop5.10} to the observation situation of the Nebra people in the year -1601. We assume that the Pleiades always became visible at sun depth $h = -12^\circ$ at nautical twilight. Then the values of the azimuths of $\beta$ Aur, $\eta$ Tau and the associated azimuth differences given in Table \ref{tab12} result from {\tt Stellarium} calculations. The Pleiades will thus have appeared vertically below $\beta$ Aur to the Nebra people on the days February 22 to February 27 (Julian) when they became visible in the nautical twilight. A similar result is obtained at sun depth $h = -9^\circ$.

\subsubsection{Does the beta point have an approximately constant
  distance to the vernal equinox?} \label{sec524}

We now investigate the question whether the beta day provides a date
which is a suitable starting point for the determination of a start
date of the agricultural year. This would be the case if the
associated beta point had a distance to the vernal equinox that is
constant or changes only very slowly.

\begin{table}  
\rotatebox{90}{\parbox{\textheight}{
\hspace{1cm}\hspace{40pt}
\begin{tabular}{|c|c|c|c|c|c|c|}
\hline
date (Jul.) & $\lambda$ of $B$  & Az $\beta$ Aur & Az $\eta$ Tau
&$\Delta \lambda$/year & App. sidereal time & Hour angle $\alpha$ Cyg  \\
\hline
-2101-08-01   08:01:24 &  322.6127$^\circ$ &         239.7388$^\circ$  &
239.7380$^\circ$ &           &         3:26:44.0 & 9:03:36.71\\
-2076-08-01   07:58:33 &  322.7624$^\circ$ &          239.7271$^\circ$ &       239.7265$^\circ$ &   21.5568" &          3:27:34.3 & 9:03:36.91\\
-2051-08-01   07:59:38 &  322.9123$^\circ$ &          239.7153$^\circ$ &       239.7152$^\circ$ &   21.5856" &          3:28:24.8 & 9:03:37.20\\
-2026-08-01   08:00:44 &  323.0653$^\circ$ &          239.7083$^\circ$ &       239.7079$^\circ$ &   22.0320" &          3:29:16.3 & 9:03:38.50\\
-2001-08-01   08:01:49 &  323.2183$^\circ$ &          239.7006$^\circ$ &       239.7004$^\circ$ &   22.0320" &          3:30:07.8&  9:03:39.80\\
-1976-08-01   07:59:00 &  323.3740$^\circ$ &          239.6973$^\circ$ &       239.6965$^\circ$ &   22.4208" &          3:31:00.2 & 9:03:42.01\\
-1951-08-01   08:00:07 &  323.5270$^\circ$ &          239.6880$^\circ$ &       239.6889$^\circ$ &   22.0320" &          3:31:51.8 & 9:03:43.32\\
-1926-08-01   08:01:14 &  323.6859$^\circ$ &          239.6894$^\circ$ &       239.6891$^\circ$ &   22.8816" &          3:32:45.3 & 9:03:46.61\\
-1901-08-01   08:02:23 &  323.8450$^\circ$ &          239.6902$^\circ$ &       239.6893$^\circ$ &   22.9104" &          3:33:38.8 & 9:03:49.95\\
-1876-08-01   07:59:34 &  324.0027$^\circ$ &          239.6875$^\circ$ &       239.6876$^\circ$ &   22.7088" &          3:34:31.9 & 9:03:52.82\\
-1851-08-01   08:00:43 &  324.1629$^\circ$ &          239.6887$^\circ$ &       239.6890$^\circ$ &   23.0688" &          3:35:25.9 & 9:03:56.50\\
-1826-08-01   08:01:52 &  324.3249$^\circ$ &          239.6928$^\circ$ &       239.6930$^\circ$ &   23.3280" &          3:36:20.4 & 9:04:00.83\\
-1801-08-01   08:03:01 &  324.4870$^\circ$ &          239.6961$^\circ$ &       239.6968$^\circ$ &   23.3424" &          3:37:15.0&  9:04:05.16\\
-1776-08-01   08:00:15 &  324.6518$^\circ$ &          239.7039$^\circ$ &       239.7041$^\circ$ &   23.7312" &          3:38:10.5 & 9:04:10.40\\
-1751-08-01   08:01:24 &  324.8169$^\circ$ &          239.7114$^\circ$ &       239.7118$^\circ$ &   23.7744" &          3:39:06.1 & 9:04:15.74\\
-1726-08-01   08:02:36 &  324.9850$^\circ$ &          239.7240$^\circ$ &       239.7233$^\circ$ &   24.2064" &          3:40:02.7&  9:04:22.09\\
-1701-08-01   08:03:45 &  325.1502$^\circ$ &          239.7298$^\circ$ &       239.7307$^\circ$ &   23.7888" &          3:40:58.4 & 9:04:27.43\\
-1676-08-01   08:01:02 &  325.3211$^\circ$ &          239.7462$^\circ$ &       239.7458$^\circ$ &   24.6096" &          3:41:55.9&  9:04:34.69\\
-1651-08-01   08:02:13 &  325.4892$^\circ$ &          239.7562$^\circ$ &       239.7570$^\circ$ &   24.2064" &          3:42:52.6 & 9:04:41.05\\
-1626-08-01   08:03:25&   325.6604$^\circ$ &          239.7714$^\circ$ &       239.7722$^\circ$ &   24.6528" &          3:43:50.2 & 9:04:48.41\\
-1601-08-01   08:04:38&   325.8346$^\circ$ &          239.7917$^\circ$ &       239.7913$^\circ$ &   25.0848" &          3:44:48.9 & 9:04:56.78\\
\hline
\end{tabular}
}}
\vspace{-0.8cm}\vspace{-40pt}
\caption{The ecliptic longitude $\lambda$ of the beta point $B$.
We determined this Table for the Mittelberg on August 1 (Julian) of each listed year.
 The time is Local Mean Solar Time. $\lambda$ is determined in the ecliptical system of the day. The hour angle of $\alpha$ Cyg (Deneb) is equal to the stellar time with Deneb as reference point.}
\label{tab10}
\end{table}
The calculation of the ecliptic longitude $\lambda$ of the beta point $B$
using Stelarium \cite{zottihoffmann, zottiwolf} can be performed with
a procedure similar to the procedure in Appendix \ref{appC3} for the
calculation of beta days. We give this procedure in Appendix
\ref{appC4}.

However, for the calculation of the coordinates of beta points we need
a greater accuracy than for the calculation of beta days. We have
therefore made the following modifications to the procedure in
Appendix \ref{appC3} to obtain the procedure in Appendix \ref{appC4}.

\begin{itemize}
\item We added an interval nesting method to find a time interval of 1
  second length in which the moment lies in which $\eta$ Tau
  (Pleiades) is exactly perpendicular under $\beta$ Aur (steps (8) -
  (14) in App. \ref{appC4}).
\item We increased the magnification of the sky image to its maximum
  value in several steps to obtain the best possible approximation of
  the beta point (steps (15) - (18) in App. \ref{appC4}).
\end{itemize}
Determining the exact time when $\eta$ Tau is perpendicular under
$\beta$ Aur is unfortunately not possible with {\tt Stellarium}. The
software {\tt Stellarium} produces the image of the sky at a given date and
time. This is its main purpose. The input of time units smaller than 1
second is not intended here. Therefore, for a point of time, which one
would like to determine, one can find at most an interval of 1 second
length, in which this point of time lies. However, since we only want
to reconstruct observations of the Nebra people , which were only made
with the naked eye and where the vertical positioning of $\eta$ Tau
under $\beta$ Aur was only estimated, an accuracy of 1 second is quite
sufficient.

Table \ref{tab10} was determined using the procedure in Appendix
\ref{appC4}. We calculated the ecliptic longitude $\lambda$ of the beta points at 25-year
intervals. August 1 (Julian) was chosen each year as the day on which the beta point was determined. The time given in Table \ref{tab10} is the time when on August 1 (Julian) $\eta$ Tau is perpendicular under $\beta$ Aur and moves from right to left under $\beta$ Aur.
$\Delta\lambda$/year is calculated according to the formula
\[ \Delta\lambda/\text{year} = (\lambda_\text{current line} -
\lambda_\text{previous line})/25\;. \]
The columns for Az $\beta$ Aur and Az $\eta$ Tau were included to
check that both stars are approximately perpendicular to each
other. The last two columns can be used to check that the
perpendicular positioning of the two stars occurs at approximately the
same sidereal time or stellar time. As expected the apparent sidereal
time already shows a shift of about 1 minute after 25 years.

We measure the stellar time by the hour angle of the star $\alpha$ Cyg
(Deneb). $\alpha$ Cyg has only a very small proper motion (2.8 mas/yr)
and thus changes its position only by 1.4" in 500 years. It is
therefore a good reference point. When $\alpha$ Cyg has its transit
(i.e. is at the local meridian), it is 0 o'clock stelar time.

Using stellar time, there are only shifts in the range of seconds in 25 years, which only add up to a shift of about 1 minute in 500 years.

\mycomment{
\begin{figure}[!ht]
  \begin{center}
    \includegraphics[width=\textwidth]{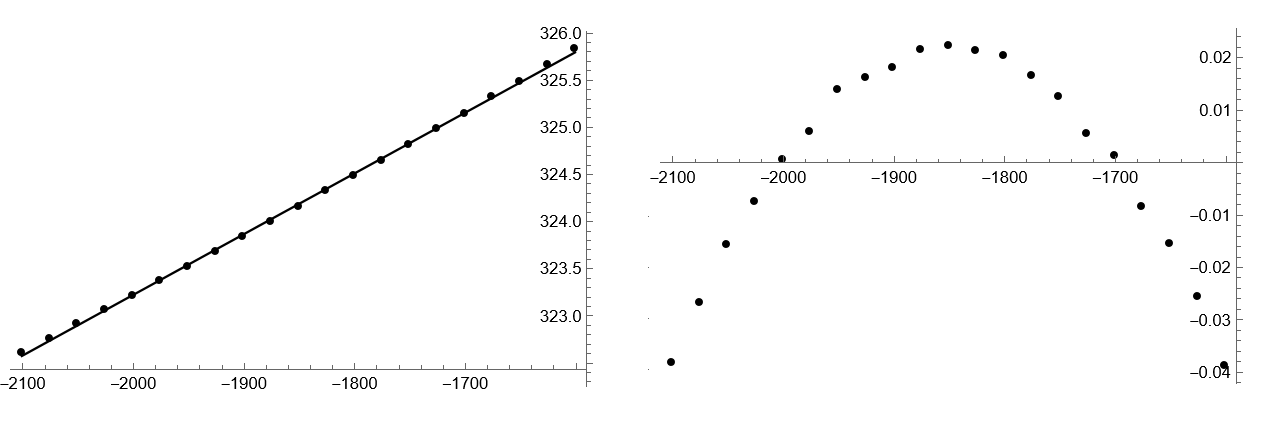}
\end{center}
\caption{The fit (\ref{5.1}) (left) and its residuals (right).} 
\label{ausgleich1}
\end{figure}
}
\mycomment{
\begin{figure}[!ht]
  \begin{center}
    \includegraphics[width=\textwidth]{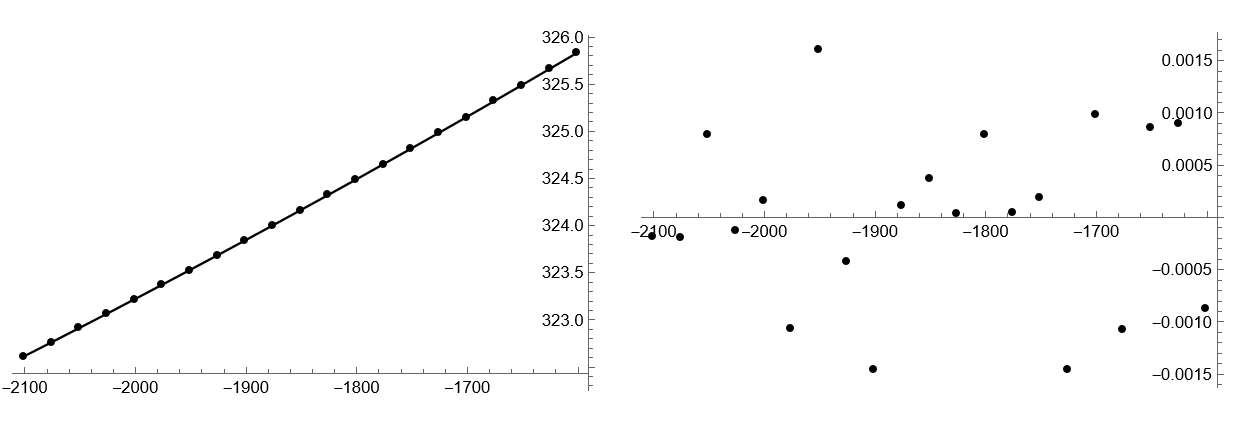}
\end{center}
\caption{The fit (\ref{5.2}) (left) and its residuals (right).} 
\label{ausgleich2}
\end{figure}
}
\mycomment{
\begin{figure}[!ht]
  \begin{center}
    \includegraphics[width=\textwidth]{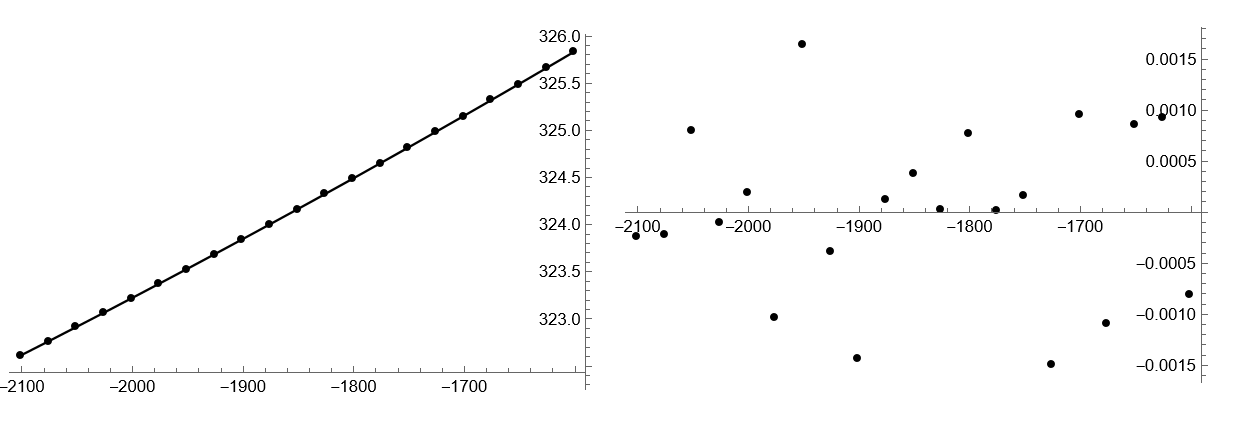}
\end{center}
\caption{The fit (\ref{5.3}) (left) and its residuals (right).} 
\label{ausgleich3}
\end{figure}
}

We can find even more precise functional relationships between the data from Table \ref{tab10} applying the least squares method. We use {\tt Mathematica} \cite{mma} for this.

For a functional relationship between the year $y$ and the ecliptic length $\lambda$ of the beta point $B$, we first form the set {\tt data} of points $(y_i, \lambda_i)$, where $y_i$ is the year in the $i$th row and $\lambda_i$ is the $\lambda$-value in the $i$th row of Table \ref{tab10}. Then we apply the Mathematica tool {\tt Fit} and calculate fitting polynomials of orders 1, 2 and 3 using {\tt Fit[data,\{1,x\},x]}, {\tt Fit[data,\{1,x, x\^{ }2\}, x]}, {\tt Fit[data,\{1,x, x\^{ }2, x\^{ }3\},x]}. The results are
\begin{eqnarray}
  \lambda &=& 336.11014952554126 + 0.006442425974025995 \cdot y \label{5.1}\\
  \lambda &=& 339.36915458854673 + 0.009987482913605685\cdot y \label{5.2}\\
  & & +\;9.576058723878078\cdot 10^{-7}\cdot y^2 \nonumber
\end{eqnarray}
\begin{eqnarray}
  \lambda &=& 339.44114470090994 + 0.010105105571940302\cdot y \label{5.3}\\
  & & +\; 1.0214066074646323\cdot 10^{-6}\cdot y^2 \nonumber\\
  & & +\; 1.1489417445834214\cdot 10^{-11}\cdot y^3 \nonumber
\end{eqnarray}
The Figures \ref{ausgleich1}, \ref{ausgleich2} and \ref{ausgleich3} show graphics of these fits and its residuals
\begin{equation}
  r_i = f(y_i) - \lambda_i\;.
\end{equation}

From Figure \ref{ausgleich1} we see that the first-order fit polynomial (\ref{5.1}) does not yet provide the best possible fit of the function to the points in {\tt data}. At the ends of the function line, the points from {\tt data} are not on the line but slightly above the line. Furthermore, the residuals are still relatively large and do not scatter but appear to lie on a parabolic curve.

In contrast, in the Figure \ref{ausgleich2} of the 2nd order fit polynomial (\ref{5.2}), all points from {\tt data} lie on the function curve. The residuals are smaller by the factor $\frac1{10}$ than that of the function (\ref{5.1}). Thus, the fit of the function curve to the data has improved.

When changing to the fit polynomial (\ref{5.3}), however, the residuals no longer decrease (see Figure \ref{ausgleich3}). We do not achieve any further improvement here and can therefore stop at (\ref{5.2}).
\begin{Prop}
Formula {\rm (\ref{5.2})} is the best possible fit in the form of a polynomial to the points from {\tt data}.
\end{Prop}
From (\ref{5.2}) we can now more precisely calculate the annual movement of the beta point on the ecliptic. Starting from $\lambda = f(y)$ according to (\ref{5.2}) we calculate using {\tt Mathematica}
\begin{eqnarray}
  k(y) := f(y+1) - f(y) &=& 0.009988440519478073 \label{5.5a}\\
  & & + 1.9152117447761297\cdot 10^{-6}\cdot y\;. \nonumber
\end{eqnarray}
Since the factor of $y$ is positive in (\ref{5.5a}), $k(y)$ is a strictly increasing function. We calculate the values of $k$ at the ends of the time interval $-2101 \leq y \leq -1601$ with {\tt Mathematica}.
\begin{eqnarray}
  k(-2101) &=& 0.0059645806437034245^\circ = 21.47249031733233'' \\
  k(-1601) &=& 0.006922186516091489^\circ = 24.91987145792936''
\end{eqnarray}
\begin{Prop} \label{prop5.12}
Between $-2101$ and $-1601$ the annual change in the ecliptic longitude $\lambda$ of the beta point increases from $21.4725''$ to $24.9199''$.
\end{Prop}
In the ecliptic coordinate system, the vernal equinox always has the length $\lambda = 0^\circ$ or $\lambda = 360^\circ$ because the measurement of $\lambda$ starts at the vernal equinox. As the length $\lambda$ of the beta point increases from $\lambda = 322.6127^\circ$, the beta point approaches the vernal equinox. However, due to the values given in Prop. \ref{prop5.12}, its movement is very slow, so that it is a suitable characteristic for the start of the agricultural year, as described at the beginning of this section.

However, the actual motion is not such that the vernal equinox rests on the ecliptic and the beta point moves towards the vernal equinox. Rather, the vernal equinox moves counter to the Sun's motion on the ecliptic at about $50''$ per year. The beta point moves in the same direction as the vernal equinox, but more slowly than the vernal equinox. As a result, the spring equinox is slowly catching up with him.

As the distance between the beta point and the vernal equinox decreases, these points will eventually coincide and the vernal equinox will overtake the beta point. We can calculate when that will happen. Table \ref{tab17a} shows the calculation for the Mittelberg.

To find the year in which $\lambda = 360^\circ$ applies to the beta point, we consider the function $\lambda = f(y)$ from (\ref{5.3}) and solve using {\tt Mathematica} the equation $f(y) = 360^\circ$. We choose (\ref{5.3}) instead of (\ref{5.2}) because the year we are looking for will probably be well outside the range $-2101 \leq y \leq -1601$ for which the functions (\ref{5.1}) -- (\ref{5.3}) were actually calculated. The summand with $y^3$ in (\ref{5.3}) will perhaps provide a little more precision than using (\ref{5.2}).

As a solution of $f(y) = 360^\circ$ we get the value $y = 1727.13$, i.e. the year $y = 1727$. Using {\tt Stellarium} \cite{zottihoffmann, zottiwolf} and the procedure from Appendix \ref{appC4} we now determine the ecliptic length $\lambda$ of the beta point in the year $y = 1727$. The result is $\lambda = 366.0944^\circ$. The vernal equinox has thus already overtaken the beta point and the beta point is now on the other side of the vernal equinox, which is always at $\lambda = 360^\circ$.

To find a year where the beta point is closer to $\lambda = 360^\circ$, we compute a new best fit polynomial by adding the point $(1727, 366.0944)$ to the set {\tt data} and then use the {\tt Mathematica} tool {\tt Fit} to determine fitting polynomials for {\tt data} as above. We gradually increase the order of these polynomials until we get a polynomial with the smallest possible residuals. In this way we get the polynomial
\begin{eqnarray}
  \lambda &=& 339.8692126709598 + 0.011041655112973253\cdot y \label{5.4}\\
  & & +\; 1.7875651660479225\cdot 10^{-6}\cdot y^2 \nonumber \\
  & & +\; 2.892455333353818\cdot 10^{-10} \cdot y^3 \nonumber \\
  & & +\; 3.765172453929443\cdot 10^{-14} \cdot y^4\;. \nonumber
\end{eqnarray}
With the polynomial (\ref{5.4}), the equation $f(y) = 360^\circ$ yields the value $y = 1412.66 $. In 1412, according to {\tt Stellarium}, the beta point has the ecliptic longitude $\lambda = 358.7039^\circ$. We have bounced back to the other side of the vernal equinox but only a good degree away from it.

We again add the point $(1412, 358.7039^\circ)$ to the set {\tt data} and again determine a fitting polynomial with the smallest possible residuals. We get the polynomial of order 5
\begin{eqnarray}
  \lambda &=& 339.1447310507547 + 0.009887022255267337\cdot y \label{5.5}\\
  & & + \; 1.420159266771764\cdot 10^{-6}\cdot y^2 \nonumber \\
  & & + \; 5.693052077566582\cdot 10^{-10}\cdot y^3 \nonumber \\
  & & +\; 2.422839250033822\cdot 10^{-13}\cdot y^4 \nonumber \\
  & & +\; 3.5898514817170673\cdot 10^{-17}\cdot y^5\;. \nonumber
\end{eqnarray}
With $f(y)$ according to (\ref{5.5}) it follows from the equation $f(y) = 360^\circ$ $y = 1473.06$ and {\tt Stellarium} delivers $\lambda = 359.9044^\circ $ for the year 1473. The vernal equinox and the beta point are only separated by $\frac1{10}$ degrees. We now determine an even more precise year by means of linear extrapolation.
\begin{eqnarray*}
  (1473-1412):(359.9044 - 358.7039) &=& x:(360- 358.7039) \\
  61:1.2005 &=& x:1.2961 \\
1.2005\cdot x &=& 79.0621 \\
x = 79.0621/1.2005 &=& 65.857642648896293211162015826739
\end{eqnarray*}
This $x$ value then results in the new year
\begin{eqnarray*}
  y &=& 1412 + 65.857642648896293211162015826739 \\
  &=& 1477.8576426488962932111620158267 \sim 1478\;.
\end{eqnarray*}
In 1478, according to {\tt Stellarium}, the beta point has an ecliptic length $\lambda = 360.0071^\circ$.
\begin{Prop}
In a year close to the year {\rm 1478} the beta point and the vernal equinox coincide. After that they lie on the ecliptic in the reverse order compared to the time before their coincidence.
\end{Prop}

\begin{table}[t] 
\begin{center}
\begin{tabular}{|c|c|l|c|c|c|c|c|}
\hline
 n & $\lambda = f(y)$ & $y$ from & date/time & $\lambda$ of B & new point\\
  & & $f(y) = 360^\circ$ & & & \\
 \hline
 1 & (\ref{5.3})& 1727.13& 1727.08.01  & $366.0944^\circ$ & $(1727, 366.0944^\circ)$\\
 & & & 10:55:06 & & \\
 \hline
  2 & (\ref{5.4})& 1412.66 & 1412.08.01  & $358.7039^\circ$ & $(1412, 358.7039^\circ)$\\
 & & & 9:36:13 & & \\
 \hline
 3 & (\ref{5.5})& 1473.06 & 1473.08.01  & $359.9044^\circ$ & \\
 & & & 9:42:12 & & \\
 \hline
\end{tabular}
\vspace{3mm}
\caption{Intermediate results for the calculation of the year in which the beta point and vernal equinox coincide. The data was determined for the Mittelberg. Time is the local time when the beta event occurs on August 1st of the given year.}
\label{tab17a}
\end{center}
\end{table}
\mycomment{
\begin{figure}[!ht]
  \begin{center}
    \includegraphics[width=\textwidth]{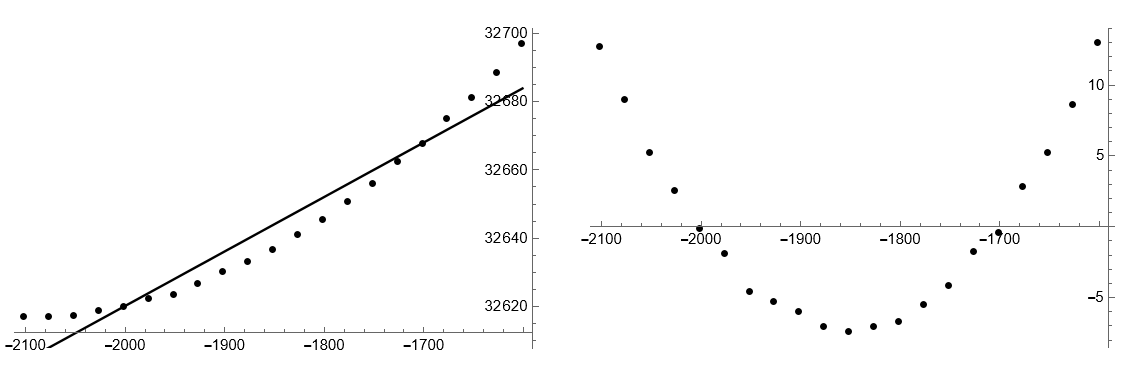}
\end{center}
\caption{The stellar time of the beta event between -2101 and -1601. The fit (\ref{pol1}) (left) and its residuals (right).}
\label{stellar1} 
\end{figure}
}
\mycomment{
\begin{figure}[!ht]
  \begin{center}
    \includegraphics[width=\textwidth]{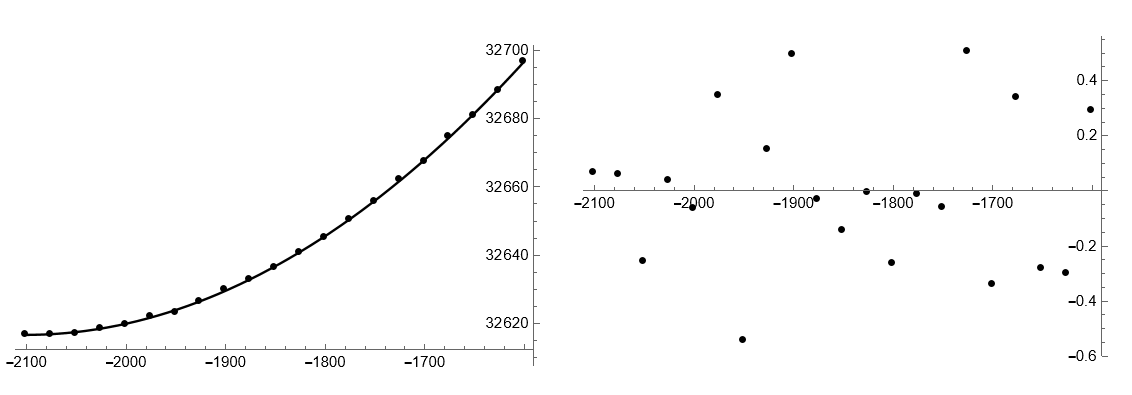}
\end{center}
\caption{The stellar time of the beta event between -2101 and -1601. The fit (\ref{pol2}) (left) and its residuals (right).} 
\label{stellar2}
\end{figure}
}
\mycomment{
\begin{figure}[!ht]
  \begin{center}
    \includegraphics[width=\textwidth]{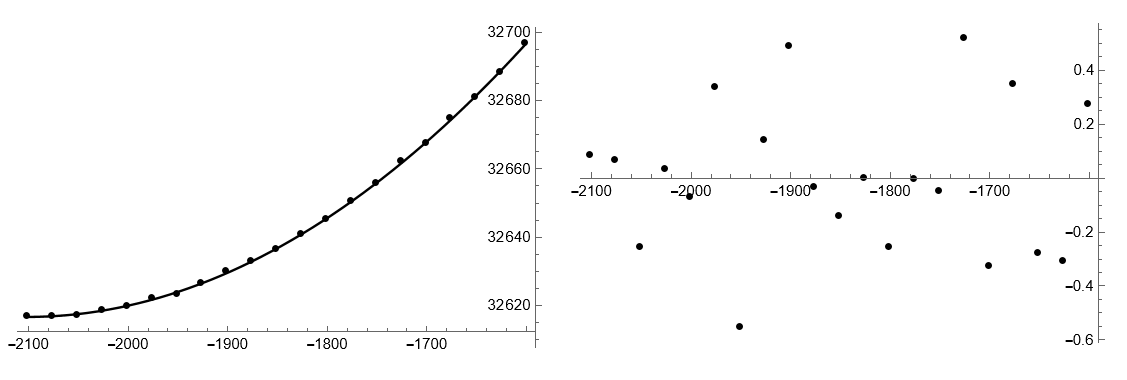}
\end{center}
\caption{The stellar time of the beta event between -2101 and -1601. The fit (\ref{pol3}) (left) and its residuals (right).} 
\label{stellar3}
\end{figure}
}

We now derive formulas for the shift of the stellar time $t^*$ at which a beta event occurs.

For this purpose we again form a set {\tt data} of points $(y_i, t^*_i)$, where $y_i$ is the year in the $i$th row and $t^*_i$ is the $t^*$-value (in seconds) in the $i$th row of Table \ref{tab10}. Then we again apply the Mathematica tool {\tt Fit} and calculate fitting polynomials of orders 1, 2 and 3 using {\tt Fit[data,\{1,x\},x]}, {\tt Fit[data,\{1,x, x\^{ }2\}, x]}, {\tt Fit[data,\{1,x, x\^{ }2, x\^{ }3\},x]}. The results are

\begin{eqnarray}
  t^* &=& 32939.53019861475 + 0.15968779220780813\cdot y \label{pol1}\\
  t^* &=& 34024.35350287913 + 1.3397290344431394\cdot y \label{pol2}\\
  & & +\;0.00031875776397496643\cdot y^2 \nonumber
\end{eqnarray}
\begin{eqnarray}
  t^* &=& 34046.437165024734 + 1.3758109207196094\cdot y \label{pol3}\\
  & & +\; 0.00033832925555652007\cdot y^2 \nonumber \\
  & & +\; 3.524489749962895\cdot 10^{-9}\cdot y^3 \nonumber
\end{eqnarray}

If we plot the fits (\ref{pol1}), (\ref{pol2}), (\ref{pol3}) and their residuals we obtain the Figures \ref{stellar1}, \ref{stellar2}, \ref{stellar3}. We see that (\ref{pol1}) still has relatively large residuals, which lie approximately on a parabolic curve and correspond to a missing term of the fitting polynomial we are looking for. Fit (\ref{pol2}) contains this missing term, which leads to much smaller residuals. However, a transition to fit (\ref{pol3}) does not result in any further reduction of the residuals. This gives us

\begin{Prop}
Formula {\rm (\ref{pol2})} is the best possible fit in the form of a polynomial to the set of points $(y_i , t^*_i)$.
\end{Prop}


\section{On the interpretation of the large moon symbols on the Sky Disc}
Actually, neither a full Moon nor a new Moon is needed to determine
heliacal sets or beta days. Then why do Moon symbols appear on the Sky Disc? On
this question it is said in the interpretation of W. Schlosser \cite{schlosser3}
that the Moon symbols indicate spring or autumn. In the vicinity of
the heliacal setting (evening setting) of the Pleiades in spring there is a conjunction
between the Pleiades and a narrow waxing Moon, in autumn, near the
cosmic setting (morning setting) of the Pleiades, there is a conjunction of the Pleiades
and a almost full Moon.

There is no doubt that these conjunctions exist. However, there are some objections to the use of the conjunction in autumn to indicate the fall or harvest season.

In Appendix \ref{appE2} we show that two conjunctions of the Pleiades with a nearly full moon often occur one month apart (Tables \ref{tab28} and \ref{tab29}). The Pleiades-Moon conjunction would therefore not be clearly determined.

Even more serious is the objection that conjunctions of the Pleiades with a nearly full moon cannot be seen with the naked eye. If the moon is only a few degrees away from the Pleiades in the sky with a surface illuminated by at least 90\%, then its bright light outshines the Pleiades and the Pleiades cannot be seen. The Nebra people could therefore not have known on the night of the conjunction that a conjunction was taking place. We have verified the non-visibility of the Pleiades by our own observations (Appendix \ref{appE2}, Tables \ref{tab30} and \ref{tab30b}).

If the lunar symbols (full moon and waxing crescent) are ruled out as memograms for determining spring and autumn, what significance do they have? The moon's phases provide a division of the year into both months and weeks. The Nebra people will have used this division option, where several possibilities are conceivable.

\subsection{A lunisolar calendar that is harmonized with the solar year} \label{subsec6.1}

If one sets the rule
\begin{Rule} \label{rule6.1}
A lunar year (common year) consists of 12 months, each beginning with a new moon.
Determine the beta day and then start the new lunar year on the first new moon after the beta day.
\end{Rule}
then a lunisolar calendar is created in which a leap month is automatically inserted when the difference between the solar and lunar year becomes too large.

The synodic month has a length of 29.53 days. The lunar year is therefore 354.36 days long, $10.89 \sim 11$ days shorter than the solar year of 365.25 days.

The beta day is a phenomenon of the solar year. In a period of a few years, the distance between two beta days is always the length of a solar year, because according to Proposition \ref{prop5.12} the position of the beta point hardly changes within a few years.

If we now consider a lunar year that we start within 29 days of the beta day, then the start of the next lunar year moves 11 days closer to the beta day from year to year. After a short time there will be a lunar year that begins before the beta day. However, since Rule \ref{rule6.1} requires that lunar years must always start after the beta day, an additional month, a leap month, is inserted in such a case. Whenever there is a lunar year that begins before the beta day, a leap month will automatically occur if Rule \ref{rule6.1} is followed. No special mathematical knowledge is required for this.

R. Hansen also proposed such a leap rule in \cite{hansen1, hansen2}, where he assumed that it might have been imported from the Babylonians. If they knew about the beta day, the Nebra people could have figured it out on their own.

However, rule \ref{rule6.1} would only suggest depicting the waxing crescent moon on the disk. So why also depict the full moon?

One could also consider lunar years in which the months always begin with a full moon. Rule \ref{rule6.1} would have to be replaced by

\begin{Rule} \label{rule6.2}
A lunar year (common year) consists of 12 months, each beginning with a full moon.
Determine the beta day and then start the new lunar year on the first full moon after the beta day.
\end{Rule}
But this new rule would then only suggest depicting the full moon and not the waxing crescent moon.

And it seems unlikely that the Nebra people wanted to make a memogram for both possibilities.

\subsection{A week count for the farming year} \label{subsec6.2}
It is possible that the Nebra people did not yet have a complete lunisolar calendar, but merely used the course of the moon to divide the farming year into weeks and to calculate important dates for agriculture by counting the weeks that had passed since sowing.

A later example of a list of agricultural dates is the list of dates by Hesiod given in Table \ref{tab17b}, which can be found in Ideler \cite[p.245-247]{ideler1}. The Gregorian dates given in Table \ref{tab17b} are based on the Julian dates calculated by Ideler for the astronomical phenomena, which we have simply converted into Gregorian dates.

The Nebra people may not have known as many rises and sets as Hesiod, but they would have had similar dates and could determine them by counting the weeks. In order for such a counting to work, it is important to schedule the sowing at a suitable phase of the moon so that the following weeks can be counted from the time of sowing.
In detail, this could have happened as follows:
\begin{Def}
We introduce the common term {\it round moon phase} for new moon and full moon.
\end{Def}
Then we use the following rule to determine a sowing date:
\begin{Rule} \label{rule6.4}
Determine the beta day. Then start sowing on the second round moon phase following the beta day.
  \end{Rule}

\begin{table}[t] 
\begin{center}
\begin{tabular}{|l|l|l|}
\hline
 Astronomical phenomenon & Gregorian date & Agricultural activity/season\\
 \hline
 morning setting Pleiades & Oct 27 & winter sowing time \\
 morning setting Hyades & Oct 31 & winter sowing time \\
 morning setting Orion & Nov 8 & winter sowing time \\
 evening rising Arctur & Feb 17 & beginning of spring \\
 evening setting Pleiades & Mar 29 & announcer morn. ris. Pleiades \\
 morning rising Pleiades & May 12 & harvest \\
 morning rising Orion & Jul 2 & threshing time \\
 morning rising Arrctur & Sep 11 & grape harvest\\
 \hline
\end{tabular}
\vspace{3mm}
\caption{Agricultural dates according to Hesiod \cite[p.245 - 247]{ideler1}.}
\label{tab17b} 
\end{center}
\end{table}

In which period would a sowing date determined according to Rule \ref{rule6.4} fall? To answer this question, we must first clarify what time intervals can be expected between round moon phases.

For the mean length of the synodic month, i.e. the time between a lunar phase and the return of this lunar phase, the formula
\begin{equation}
M = 29^d.5305881 - 0^d.0000002\cdot\tau
\end{equation}
is given in \cite[p.66]{ahnert2}. Here, $\tau$ is the time (in centuries) between 1974 (the year \cite{ahnert2} was published) and the time for which we want to determine $M$. For the past, $\tau$ is positive, for the future it is negative.

For -1600, this gives $\tau = 36$ and $M = 29.5305809$. The mean synodic month of the Bronze Age differs little from ours.

The actual length of the synodic month differs from the mean length by considerable fluctuations. In \cite[p.224]{weigzimm} $\pm 6$ hours are mentioned, Meeus \cite[p.354]{meeus} even provides examples with $\pm 7$ hours.

7 hours are in fractions of a day $7/24 = 0.292 \sim 0.3$ This gives us the minimum and maximum values for $M$
\begin{equation}
M_{\text{min}} = 29.23\quad,\quad M_{\text{max}} = 29.83\,.
\end{equation}

The earliest time for the second round moon phase after the beta day is when the first round moon phase after the beta day occurs immediately after the start of the day after the beta day. The second round moon phase after the beta day can then occur at the earliest at time
\begin{equation}
\text{beta day} + M_{\text{min}}/2 = \text{beta day} + 14.615 \text{days.}
\end{equation}

The latest possible time for the second round lunar phase after the beta day is obtained by adding $M_{\text{max}}$ to the beta day. This would put the start of sowing in the interval
\begin{equation}
[\text{beta day} + 14.615 \text{days}\;,\;\text{beta day} + 29.83 \text{days}]\;. \label{6.4}
  \end{equation}
With a beta day on February 12th (Gregorian) (see Table \ref{tab9}), the start of sowing would fall between February 26th and March 14th (Gregorian), if there is no leap year of the Gregorian calendar.

If the Nebra people only estimated the vertical positioning of the Pleiades under $\beta$ Aur and thus in an interval
$$[\text{beta day - 3days}\;,\;\text{beta day + 3days}]$$
the Pleiades appeared to them vertically under $\beta$ Aur (as in Table \ref{tab12}), then in Rule \ref{rule6.4} we must use the last day on which the Pleiades are vertically under $\beta$ Aur instead of the beta day itself. The interval (\ref{6.4}) would change to
\begin{equation}
[\text{beta day} + 17.615 \text{days}\;,\;\text{beta day} + 32.83 \text{days}]\;.
  \end{equation}

With this procedure, the sowing date sometimes falls on a new moon and sometimes on a full moon. Likewise, the week count sometimes begins on a new moon and sometimes on a full moon. If one sees this procedure as the background for the Disc image, then it becomes understandable why both phases of the moon were shown on the Disc.

We can well imagine that this type of week counting was already used when there was no astronomical criterion for a sowing date, but biological criteria (flowering of certain plants, return of migratory birds) were used to start sowing, because there is the same need to connect the signal to start sowing to a week count.

Round moon phases are so good for starting the week count because they cannot be confused as easily as the first and last quarter would be.

If one wants to count also the months that have passed since sowing, one only has to remember until autumn which phase of the moon was present when sowing. One could do this by leaving a black or white stone in a special place, for example.

The method given in Rule \ref{rule6.4} for determining a sowing date may in some years mean that sowing does not take place until mid-March (Gregorian).
Can such a delay be a problem?

During a river cruise\footnote{PLANTOURS cruise from Nancy to Bonn, MS SAN SOUCI, 14.07-21.07.2022.} we had the opportunity to discuss with a farmer, Mr. Th. Geiger\footnote{Thomas Geiger, Heuberger Hof 1, D-72108 Rottenburg /N, Germany.}, the question of whether delayed sowing could lead to a poor harvest in the autumn.

Mr. Geiger told us that a delay of 14 days is not a serious problem. Nature would catch up. And that would apply to the Bronze Age just as much as it does to our time today.

The main factors that can lead to a crop failure are
\begin{itemize}
\item Water shortage due to lack of rainfall,
  \item Nitrogen deficiency in the soil due to insufficient fallowing.
\end{itemize}
Compared to these factors, delayed sowing has a small impact.

Mr. Geiger also told us that one doesn't actually need an astronomical criterion to determine the right time to harvest. In the Bronze Age, grain was harvested at the "hard dough" stage of maturity.
The grain is then hard and can no longer be squeezed out, but can be bitten with good teeth. Fingernail impressions are irreversible. Any experienced farmer can determine these properties from the grains on the stalk. He doesn't need astronomy to determine the harvest time.

So it is not so problematic that in our new interpretation of the lunar symbols on the Sky Disc the criterion for "autumn" or "harvest time" has been eliminated.

Mr. Geiger also told us that, just as with harvesting, no astronomical criteria are actually required to determine the time for sowing. There are many biological characteristics that can be used to determine the time for sowing. This is also the basis for our view that a similar week counting system to the one described here could have been used before the use of astronomical characteristics. When the biological characteristics for sowing had occurred, people waited until the next round phase of the moon to sow so that the week counting system could begin with sowing.

\section{How could the astronomical things shown on the sky disc have been discovered by the Nebra people?} \label{sec7}

We now describe a possible way in which the Nebra people might have discovered the astronomical knowledge that we believe we have identified on the Sky Disc.

From our investigations in Appendix \ref{appB} it follows that with probability 1 the Disc image Figure \ref{AuxLine4} is a scaled representation of the Auriga line. Since the angle of inclination of the line $\alpha$ Aur -- $\epsilon$ Gem to the perpendicular from $\beta$ Aur to the horizon deviates only slightly from the actual value of this angle, Figure \ref{AuxLine4} shows the moment when the Pleiades are perpendicular to $\beta$ Aur.

This shows that W. Schlosser's hypothesis that the large symbol of 7 stars represents the Pleiades is correct. The Pleiades were therefore known to the Nebra people.

According to W. Schlosser's calculations and our own calculations in Table \ref{tab5}, the heliacal setting of the Pleiades in the period -2101 to -1601 occurred on a date that is suitable as a seeding time (even when using biological characteristics to determine the seeding time).

This often results in the Pleiades being visible in the sky before sowing. After sowing, however, they have disappeared. The Nebra people will have noticed this too. At some point they will have made a series of observations to explain the Pleiades' disappearance. This is how they discovered the heliacal setting of the Pleiades.

They probably used the heliacal setting of the Pleiades as an indicator for the sowing for some time. Since they could not predict the heliacal setting, they had to start observing the Pleiades long before the sowing and had to accept a long series of observations. In these observations they observed the apparent movement of the Pleiades below the stars of the Auriga line and the beta event.

They now only needed to start observing the Pleiades from the beta day and could thus shorten the observation series to 24 to 28 days (Table \ref{tab10c}).

However, they probably soon realized that they could also arrive at a usable sowing date by using round phases of the moon, to which they could then even add a week count. This is all the more likely if the Nebra people had already used a week count before using an astronomical sowing criterion (see Subsection \ref{subsec6.2}).

On the beta day, the Nebra people could also have built a self-harmonizing lunisolar calendar (see Subsection \ref{subsec6.1}). Whether they actually did that, we cannot prove at the moment.

All the procedures described in Section \ref{sec7} can be obtained by simply observing nature. No complicated mathematics or knowledge imported from another culture in the Mediterranean region is required.

\section{Are there still beta days today?}
To answer this question, we used the procedure described in Appendix \ref{appE} and Stelarium \cite{zottihoffmann,zottiwolf} to calculate another Table like Table \ref{tab7a} for January 1 and 2, 2020 (Gregorian) and the Mittelberg. The result is summarized in Table \ref{tab13}.

\begin{table}[t] 
\begin{center}
\begin{tabular}{|c|c|c|c|c|}
\hline
 date & time & $\text{Az}_{\eta}$ & app. sid. time & description \\
\hline
2020-01-01& 21:18:46 &$180^\circ0'$&    3:48:40  &         transit Pleiades\\
2020-01-01& 23:26:46 &$234^\circ24'$&   5:57:01  &
$\text{Az}_{\eta} = \text{Az}_{\alpha}$\\
2020-01-02&  1:45:46  &$268^\circ26'$&   8:13:23  &         hardly any
relative movement\\
&&&& to $\beta$ Aur (begin)\\
                  &&&&  $\text{Az}_{\beta}-\text{Az}_{\eta} = -52'$ \\                  
2020-01-02&  2:23:46  &$275^\circ45'$&   8:54:30 &          hardly any
relative movement\\ 
&&&& to $\beta$ Aur (end)\\
 &&&&  $\text{Az}_{\beta}-\text{Az}_{\eta} = -59'$ \\
2020-01-02&  5:04:46  &$305^\circ20'$&  11:35:57 &
$\text{Az}_{\eta} = \text{Az}_{\alpha}$\\
&&&&not visible, h$_\eta = 3^\circ33'$ \\
2020-01-02&  5:33:32  &$310^\circ53'$&  12:04:47 &          set Pleiades\\
2020-01-02& 13:00:03   &$49^\circ7'$&   19:32:33 &           rise Pleiades\\
2020-01-02& 20:07:03  &$147^\circ53'$&   2:40:42 &          hardly any
relative movement\\
&&&& to $\alpha$ Aur (begin)\\
2020-01-02& 21:07:03  &$176^\circ6'$&    3:40:52 &          hardly any
relative movement\\
&&&& to $\alpha$ Aur (end)\\
2020-01-02& 21:14:50  &$180^\circ0'$&    3:48:40  &         transit Pleiades\\
\hline
\end{tabular}
\vspace{3mm}
\caption{Important positions of the Pleiades during a full azimuth
  cycle on 2020-1-1 Greg. and 2020-1-2 Greg. on Mittelberg. ''Time''
  is CET. $\text{Az}_{\eta}$ = azimuth of $\eta$ Tau (Alcyone
/ Pleiades), h$_\eta$ = height of $\eta$ Tau  (Alcyone
/ Pleiades), $\text{Az}_{\alpha}$ = azimuth of $\alpha$ Aur
  (Capella). ''app. sid. time'' means ''apparent sidereal time''.}
\label{tab13}
\end{center}
\end{table}

We see that the movement of $\eta$ Tau (the brightest star of the Pleiades), which should pass from right to left under $\beta$ Aur, stops at about 1:45:46, $52'$ before the perpendicular from $\beta$ Aur to the horizon. At about 2:23:46 $\eta$ Tau starts moving again, but in a retrograde direction from left to right. This means that $\eta$ Tau no longer reaches a position that is exactly perpendicular to $\beta$ Aur.

To check whether at least one of the other 6 stars of the Pleiades is ever perpendicular to $\beta$ Aur, we used Stelarium to determine the azimuth angles of these stars for January 2, 2020 and for the times 1:45:46 and 2:23:46 (CET). Table \ref{tab19a} shows the results.

\begin{table}[t] 
\begin{center}
\begin{tabular}{|r|l|l|}
\hline
star & 1:45:46 & 2:23:46 \\
\hline
$\beta$ Aur & $267^\circ 34'$ & $274^\circ 46'$ \\
\hline
$\eta$ (25) Tau &  $268^\circ 26'$ & $275^\circ 45'$\\
27 Tau & $268^\circ 4'$ & $275^\circ 24'$ \\
17 Tau & $268^\circ 58'$ &  $276^\circ 15'$ \\
20 Tau & $268^\circ 59'$ & $276^\circ 16'$ \\
23 Tau & $268^\circ 23'$ & $275^\circ 51'$ \\
19 Tau & $269^\circ 11'$ & $276^\circ 27'$ \\
28 Tau & $268^\circ 8'$ & $275^\circ 27'$ \\
\hline
\end{tabular}
\vspace{3mm}
\caption{The azimuths of $\beta$ Aur and the 7 stars of the Pleiades on the Mittelberg on January 2nd, 2020 at the times 1:45:46 CET and 2:23:46 CET.}
\label{tab19a}
\end{center}
\end{table}

\begin{table}[t] 
\begin{center}
\begin{tabular}{|l|l|l|l|l|}
\hline
date & CET & Az $\eta$ Tau & Az $\beta$ Aur & Az$_\eta$ - Az$_\beta$ \\
\hline
2020-1-1 & 2:18:00 &  273.9220 & 273.0557 & 0.8663 \\
2045-1-1 & 2:18:00 &  274.4048 & 273.3941 & 1.0107 \\
2070-1-1 & 2:18:00 &  274.1358 & 273.0152 & 1.1206 \\
2095-1-1 & 2:18:00 &  273.8710 & 272.6460 & 1.2250 \\
2120-1-1 & 2:18:00 &  272.8558 & 271.5371 & 1.3187 \\
\hline
\end{tabular}
\vspace{3mm}
\caption{Azimuths of $\eta$ Tau and $\beta$ Aur (in degrees) and their differences between 2020 and 2120 on the Mittelberg at a time when there was no apparent relative motion of the Pleiades with respect to $\beta$ Aur.}
\label{tab20a}
\end{center}
\end{table}

We see that at both times all 7 stars of the Pleiades are to the right of the perpendicular of $\beta$ Aur. This means that there will be no beta event on the Mittelberg on January 2, 2020. The reason for this is that the circles k and K in Figure \ref{zenit} no longer intersect. The values of the azimuth differences Az$_\eta$ - Az$_\beta$ in Table \ref{tab20a} show that the distance to which the Pleiades can approach the vertical below $\beta$ Aur will continue to grow over the next 100 years.

At observation locations further south, however, the beta event can still be observed today. We checked this using Stellarium for Munich, Vienna and Athens. And even on the Mittelberg, the deviation from the perpendicular through $\beta$ Aur is so small that one can still get a good impression of what the beta event looks like when observing with the naked eye (see Prop. \ref{prop5.10}).

\section{The mathematical and astronomical abilities of the Nebra people}
We now want to compile the mathematical and astronomical abilities of the Nebra people that became apparent in our investigation. We found:
\begin{enumerate} 
\item The Nebra people were able to measure distances between stars in the sky and probably also angles (see Sec. \ref{sec524a} and App. \ref{appB}).
\item They were able to transfer star constellations to scale on the Sky Disc (Sec. \ref{sec524a}, App. \ref{appB}).
\item They were able to handle simple rational numbers such as $1, 2, 3\frac12, \frac12$ sufficiently to produce scaled images of measured figures. However, it cannot be said whether they had already calculated with such numbers using the four basic arithmetic operations.
\item It is also not possible to say whether they knew digits and calculated with digits. This is probably not the case, since they did not have any writing either.
\item The Nebra people had a certain idea of the mathematical concept of ''similarity''. They knew that proportions had to be preserved if one wanted to construct a scaled image of a figure.
  \item Their knowledge of similarity is also evident in the specification of horizon arcs. They must have determined the angle range swept by the sun over the course of a year using a gnomon. They knew that the angle traversed by the shadow is the same as the angle traversed by the sun. However, they were not yet able to represent this angle independently of the measurement method used. They recreated the situation on a gnomon on the Disc. Therefore, an observer must take a bearing along the intersecting lines in Fig. \ref{horizontwinkel} if he wants to check whether the sun has reached the position of the winter solstice.
\item It is noteworthy that the Nebra people knew about the gnomon, as did the designers of the ring sanctuaries of Goseck, P\"ommelte and Sch\"onebeck. According to Ideler \cite[p.247]{ideler1}, the Greeks did not know about the gnomon for a long time. It was brought to them by Anaximander of Miletus (610-547 BC), who had learned about it from the Babylonians.
\item Perhaps the Nebra people have already used a plumb bob to determine whether the Pleiades are perpendicular to $\beta$ Aur. However, using a measuring device at night also requires solving the following problems:
\begin{itemize}
\item One needs a shield against the wind.
\item One needs a weak lighting for the measuring device that does not blind the observer.
\end{itemize}
\item In \cite{schlosser3} W. Schlosser showed that many of the larger arcs on the sky disk are constructed circles. He was able to determine the centers and radii of these circular arcs.
  \end{enumerate}
   
\appendix
\section*{appendices}

\section{Conversion of Julian dates to Gregorian dates}
\label{appA}

\begin{table}[t] 
\begin{center}
\begin{tabular}{|c|c|}
\hline
 Sun at the spring equinox & $\Delta$ \\
\hline
   -1-3-23 Jul.  &   2 \\
   -601-3-28 Jul.  &   7 \\
   -800-3-28 Jul. & 7 \\
-1001-3-31 Jul.  &  10 \\
\hline
-1601-4-4 Jul.   &  14 \\
-1801-4-6 Jul.   &  16 \\
-1941-4-7 Jul.   &  17 \\
-2101-4-8 Jul.   &  18 \\
\hline
\end{tabular}
\vspace{3mm}
\caption{The values of $\Delta$ for converting Julian dates to Gregorian dates.}
\label{tab7}
\end{center}
\end{table}

For dates before October 15, 1582, Stellarium uses the Julian calendar. Furthermore, Stellarium applies the astronomical year count, which also includes a year zero. All dates that Stellarium provides us for the periods we are interested in are subject to these conditions.

If we want to estimate the seasonal conditions on such a Julian date, we have to convert it into a Gregorian date. We do this by using Stellarium to determine the Julian date $d$ of the day on which the sun is at the vernal equinox for the year $y$ in question. In the Gregorian calendar, the sun is at the vernal equinox on March 21st. We calculate the number of days $\Delta$ between $d$ and March 21st and can then convert a Julian date of the year $y$ into a Gregorian date of the year $y$ using the formula
$$\text{Gregorian date} = \text{Julian date} - \Delta .$$
Table \ref{tab7} lists $\Delta$-values for the years we are interested in.

This simple procedure can sometimes lead to deviations of $\pm 1$ day from Gregorian dates determined in other ways, since the sun sometimes stands at the vernal equinox on Gregorian March 20th. However, it is sufficient for estimating seasonal conditions.

\section{More precise determination of distances and angles in the second Auriga line} \label{appB}
\subsection{Distances and angles for the second Auriga line on the Sky Disc}\label{appB1}
For measurements on the Sky Disc we use 2 photos of the Disc. On the one hand, we use the image {\tt Nebra\_Scheibe\_white.jpg} from Wikipedia\\*[2mm]
\centerline{\url{https://commons.wikimedia.org/wiki/File:Nebra_Scheibe_white.jpg}}\\*[2mm]
which is also the basis for all images of the Disc in this paper. This picture was taken during an exhibition of the Disc in Basel in 2006. It could contain slight distortions that would reduce the accuracy of the measurements.

On the other hand, we use the photo 
\begin{center}
{\tt ``Die Himmelsscheibe von Nebra, , 979x936px, \textcopyright\; LDA Sachsen-Anhalt (Foto Juraj Lipt\'ak).jpg'' 
}
\end{center}
from the UNESCO World Documentary Heritage website \cite{unesco}. This photo was taken by the State Office for Monument Preservation and Archeology in Saxony-Anhalt. With it, distortions are to be expected to a lesser extent. We cannot show this photo for copyright reasons, but we can provide our measurement results here.

In order to determine distances and angles on the Sky Disc with greater accuracy than in Section \ref{sec524a}, we must resolve the following two questions:
\begin{enumerate}
\item Does the photo of the Sky Disc used contain distortions because the Sky Disc was tilted relative to the camera when it was taken?
  \item How can one determine the centers of the stars on the Sky Disc?
\end{enumerate}

\subsubsection{Measurements in the image {\tt Nebra\_Scheibe\_white.jpg}}\label{appB1.1}
\paragraph{a) Check for possible image distortions}
To check for possible distortions of an image of the Sky Disc, we use the angle $\alpha$ of the horizon arcs (see Figure \ref{horizontwinkel}). In \cite{schlosser2} W. Schlosser stated that this angle lies between $82^\circ$ and $83^\circ$ and that a more precise value is $82.7^\circ$. The same values are given in \cite[p.59]{melmich1} and \cite[p.60]{melmich1}, \cite[p.35]{melmich2}.

To exclude distortions, we measure this angle $\alpha$ on the image using the image processing program {\tt GIMP} \cite{gimp} and check whether it is sufficiently close to the values given by Schlosser, Meller and Michel.

\mycomment{
\begin{figure}[!ht]
  \begin{center}
    \includegraphics[width=\textwidth]{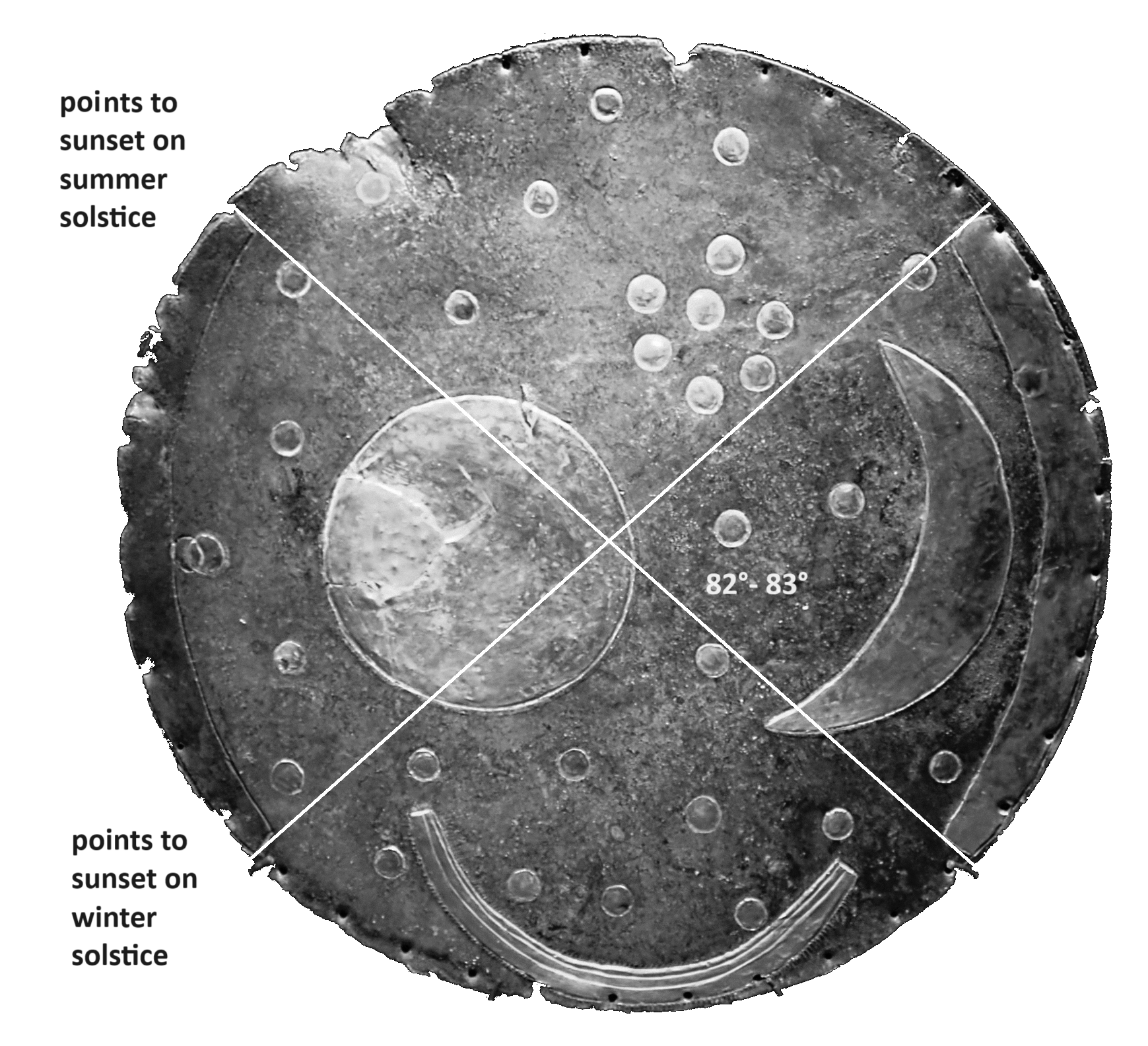}
\end{center}
\caption{The angle spanned by the horizon arcs. $82^\circ$-$83^\circ$ is the value of W. Schlosser.} 
\label{horizontwinkel}
\end{figure}
}

In order to assess measurement errors, we measure the angle $\alpha$ $n$ times and and calculate the arithmetic mean $\bar{a}$ of the measured values $a_i\,,\,(i=1,\cdots ,n)$, which we regard as the value for the horizon angle.

Then we determine a confidence interval for the mean $\bar{\alpha}$ using the method given by Sachs \cite[p.150, (4.25)]{sachs}.

Sachs' method can only be used if the arithmetic mean $\bar{\alpha}$ is normally distributed. However, according to Sachs \cite[p.151]{sachs} this is the case for $n\ge 30$ in a sufficiently good approximation. We therefore measure $\alpha$ $n=31$ times. Then Sachs' method can be applied. 

Since {\tt GIMP} cannot measure the angle $\alpha$ directly but only the inclination of a straight line relative to the horizontal, we first determine the partial angles $\alpha_1$ and $\alpha_2$ (see Figure \ref{horizontwinkel2}) .

\mycomment{
\begin{figure}[!ht]
  \begin{center}
    \includegraphics[width=\textwidth]{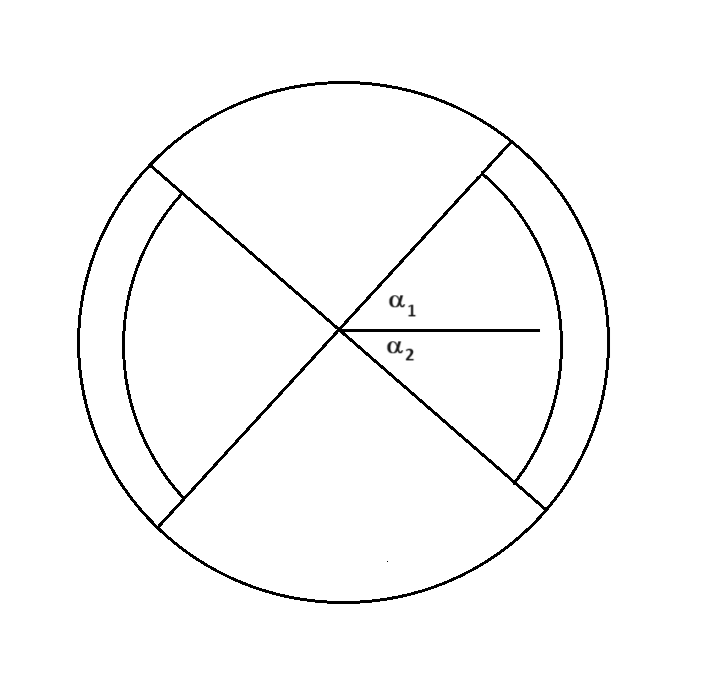}
\end{center}
\caption{The partial angles for the horizon arcs.} 
\label{horizontwinkel2}
\end{figure}
}

\begin{table}[t] 
\begin{center}
\begin{tabular}{|l|}
    \hline
    \multicolumn{1}{|c|}{Measurements $a_{1,i}$ for $\alpha_1$}\\
\hline
41.80, 41.54, 41.93, 41.78, 41.89, 41.63, 41.60, 41.77, 41.46, 41.89, 41.77,\\
41.76, 41.82, 41.79, 41.84, 41.96, 41.81, 41.80, 41.68, 41.82, 41.77, 41.79,\\
42.00, 41.76, 41.54, 41.85, 41.98, 41.83, 41.99, 41.77, 41.59 \\
   \hline
  \end{tabular}
\\*[3mm]
\begin{tabular}{|l|}
    \hline
    \multicolumn{1}{|c|}{Measurements $a_{2,i}$ for $\alpha_2$}\\
\hline
41.38, 41.38, 41.15, 41.53, 41.26, 41.16, 41.59, 41.25, 41.72, 41.49, 41.32,\\
41.29, 41.41, 41.44, 41.58, 41.50, 41.27, 41.46, 41.19, 41.51, 41.15, 41.48,\\
41.50, 41.37, 41.18, 41.38, 41.46, 41.53, 41.65, 41.64, 41.40\\
   \hline
  \end{tabular}
\\*[3mm]
\begin{tabular}{|l|}
    \hline
    \multicolumn{1}{|c|}{Values $a_i$ for $\alpha = \alpha_1 + \alpha_2$}\\
\hline
83.18, 82.92, 83.08, 83.31, 83.15, 82.79, 83.19, 83.02, 83.18, 83.38, 83.09,\\
83.05, 83.23, 83.23, 83.42, 83.46, 83.08, 83.26, 82.87, 83.33, 82.92, 83.27,\\
83.5, 83.13, 82.72, 83.23, 83.44, 83.36, 83.64, 83.41, 82.99\\
   \hline
   $\bar{a} = 83.1881$\\
   $s = 0.216601$\\
   $s_{\bar{a}} = 0.0389027$\\
\hline
  \end{tabular}
\vspace{3mm}
\caption{Data for determining the confidence interval (\ref{B7}) from {\tt Nebra\_Scheibe\_white.jpg}. All quantities in degrees.}
\label{tab33}
\end{center}
\end{table}

The results $a_{1,i}\, ,\,  a_{2,i}$ of these 31 measurements are shown in Table \ref{tab33}. By adding the $\alpha_1$ and $\alpha_2$ values of each pair of measurements, we calculate the $\alpha$ values $a_i$ and then the arthmetic mean $\bar{a}$, the empirical standard deviation $s$ of the individual measurement and the empirical standard deviation $s_{\bar{a}}$ of the mean according to

\begin{equation}
\bar{a} = \frac1n \sum_{i=1}^n a_i \quad,\quad s = \sqrt{\frac1{(n-1)} \sum_{i=1}^n (a_i - \bar{a})^2} \quad , \quad s_{\bar{a}} = \frac{s}{\sqrt{n}} \label{b1}
\end{equation}

Now we apply Sachs' method.
We choose a confidence level $1-p = 0.95$ and determine a constant $c$ from Student's t-distribution $F_{n-1}$ with $n-1 = 30$ degrees of freedom, so that $F_{n-1 }(c) = 1-p/2$. We obtain
\begin{eqnarray}
  p &=& 0.05 \\
  1-p/2 &=& 0.975 \\
  c &=& 2.04227 \label{b4}
\end{eqnarray}
The $c$ value was calculated using
\[
\text{Solve[CDF[StudentTDistribution[30],cc] == 0.975,cc][[1]]}
\]
with Mathematica \cite{mma}.

If one calculates
\begin{eqnarray}
  k &=& c\cdot s_{\bar{a}} \label{b5}\\
  k &=& 0.07945^\circ \label{B6}
\end{eqnarray}
then $\bar{a} - k \le \bar{\alpha} \le \bar{a} + k$ is the confidence interval we are looking for.

The value (\ref{B6}) results in the confidence interval
\begin{equation}
83.1086^\circ \le\bar{\alpha}\le 83.2675^\circ \label{B7}
\end{equation}
The true value of the mean value $\bar{\alpha}$ of the $\alpha$ values will lie in the confidence interval (\ref{B7}) with a probability of 95\%. Since the confidence interval is close to $83^\circ$, the image cannot contain large distortions.\\*[2mm]
\paragraph{b) Determination of distances and angles in the second Auriga line on the Sky Disc}

In order to calculate distances and angles in the second Auriga line with good accuracy, we must first find a method by which we can determine the centers of the small gold plates that represent stars. We solve this problems by fitting circles.

We loaded the photo of the Disc into the Mathematica program \cite{mma}, opened the coordinate tool by clicking on the image of the Disc and placed 8 points on the edge of each small star plate of the second Auriga line.

Mathematica provided a list {\tt points} of the coordinates of these 8 points of the star plates, measured in the photo's pixel coordinate system. We inserted these lists one after the other into the Mathematica tool {\tt RegionFit} and used
\begin{equation}
  {\tt RegionFit[points, "Circle", Method->"LMEDS"]} \label{B8}
\end{equation}
to calculate a fitting circle for each list, i.e. a circle that has the smallest distances from the 8 points.

Figure \ref{kreisfit1} shows the fitting circle and the 8 points of the star $\alpha$ Aur in very high magnification.

\begin{table}[p] 
\begin{center}
  \begin{tabular}{|r|r|r|}
    \hline
    \multicolumn{3}{|c|}{Circle around $\alpha$ Aur}\\
\hline
  x & y & d \\
 \hline
 957.9310 & 1344.9352 & 0.3598\\
 958.9655 & 1380.1076 & 0.0\\
 999.3103 & 1338.7283 & 3.7309\\
 999.3103 & 1381.1421 & 1.2715\\
 951.7241 & 1360.4525 & 0.0052\\
 978.6206 & 1334.5904 & 0.0\\
 1005.5172 & 1365.6249 & 0.0 \\
 980.6896 & 1390.4525 & 1.8091 \\
   \hline
   \multicolumn{3}{|l|}{$m = (978.7516 , 1361.6493)$}\\
   \multicolumn{3}{|l|}{$r = 27.0592$}\\
\hline
  \end{tabular}
  
  \vspace{3mm}
    \begin{tabular}{|r|r|r|}
    \hline
    \multicolumn{3}{|c|}{Circle around $\beta$ Aur}\\
\hline
  x & y & d \\
 \hline
 770.6753 & 1443.4349 & 0.4051\\
 769.6450 & 1406.3436 & 1.9750\\
 809.8273 & 1407.3739 & 0.0\\
 807.7666 & 1445.4955 & 0.0553\\
 765.5237 & 1422.8286 & 1.1911\\
 792.3119 & 1399.1314 & 0.0\\
 812.9182 & 1426.9499 & 4.0527\\ 
 787.1603 & 1451.6774 & 0.0\\
   \hline
   \multicolumn{3}{|l|}{$m = (790.6052 , 1425.4896)$}\\
   \multicolumn{3}{|l|}{$r = 26.4134$}\\
\hline
\end{tabular}
\hspace*{5mm}
  \begin{tabular}{|r|r|r|}
    \hline
    \multicolumn{3}{|c|}{Circle around $\theta$ Aur}\\
\hline
  x & y & d \\
 \hline
 674.7001 & 1302.5227 & 0.0989\\
 668.5665 & 1267.7654 & 0.0\\
 708.4351 & 1260.6095 & 0.0\\
 707.4129 & 1303.5450 & 0.0\\
 664.4774 & 1283.0995 & 0.2501\\
 690.0342 & 1258.5650 & 3.7708\\
 715.5910 & 1281.0550 & 2.9941\\ 
 687.9897 & 1312.7454 & 4.3185\\
   \hline
   \multicolumn{3}{|l|}{$m = (691.6422 , 1281.6896)$}\\
   \multicolumn{3}{|l|}{$r = 26.9513$}\\
\hline
\end{tabular}

  \vspace{3mm}
    \begin{tabular}{|r|r|r|}
    \hline
    \multicolumn{3}{|c|}{Circle around $\epsilon$ Gem}\\
\hline
  x & y & d \\
 \hline
 294.5875 & 1178.9275 & 1.0123\\
 296.6621 & 1145.7345 & 0.0918\\
 336.0788 & 1143.6600 & 0.0\\
 331.9296 & 1182.0394 & 2.6110\\
 290.4384 & 1162.3310 & 0.0\\
 311.1840 & 1137.4363 & 0.7674\\
 344.3770 & 1164.4056 & 0.3416\\
 313.2586 & 1189.3003 & 0.0\\
   \hline
   \multicolumn{3}{|l|}{$m = (317.2595 , 1162.7755)$}\\
   \multicolumn{3}{|l|}{$r = 26.8248$}\\
\hline
\end{tabular}
\hspace*{5mm}
  \begin{tabular}{|r|r|r|}
    \hline
    \multicolumn{3}{|c|}{Center star of the Pleiades Rosette}\\
\hline
  x & y & d \\
 \hline
 919.1412 & 1138.0797 & 0.0\\
 918.0761 & 1101.8679 & 0.0\\
 960.6783 & 1091.2173 & 8.4551\\
 958.5481 & 1138.0797 & 0.0140\\
 911.6858 & 1118.9087 & 0.0\\
 939.3772 & 1086.9571 & 5.2529\\
 970.2637 & 1115.7136 & 4.4499\\
 938.3121 & 1146.6001 & 0.0546\\
   \hline
   \multicolumn{3}{|l|}{$m = (938.8543 , 1119.3783)$}\\
   \multicolumn{3}{|l|}{$r = 27.1726$}\\
\hline
\end{tabular}
  \vspace{3mm}
\caption{Data of the stars of the Auriga line and the center of the Pleiades rosette: $x\,,\,y$ coordinates of the 8 edge points of the star plate, $d$ distance of the relevant edge point from the fitting circle of the star, $m\,,\, r$ center and radius of the fitting circle. Everything measured in pixels.}
\label{tab21a}
\end{center}
\end{table}

Table \ref{tab21a} lists for each star of the second Auriga line and for the center star of the Pleiades rosette the coordinates $x\,,\,y$ of the 8 points, the distance $d$ of each of these points from its fitting circle and the center $m$ and the radius $r$ of the fitting circle of the respective star.

All of these quantities are measured in pixels. $m$ and $r$ are the output of (\ref{B8}).

In order to be able to assess the quality of the fitting circles, we had Mathematica draw all 5 fitting circles into the photo of the Sky Disc (see Figure \ref{kreisfit}). The circles around the 4 stars of the Auriga line fit very well. Only the circle around the center star of the Pleiades rosette shows a larger deviation from the edge of the star.

\mycomment{
\begin{figure}[!ht]
  \begin{center}
    \includegraphics[width=8cm]{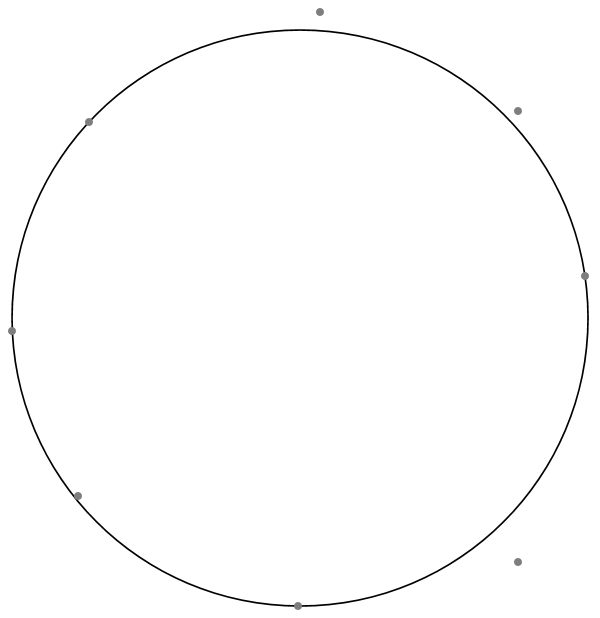}
\end{center}
\caption{The best-fitting circle for $\alpha$ Aur.} 
\label{kreisfit1}
\end{figure}
}

\mycomment{
\begin{figure}[!ht]
  \begin{center}
    \includegraphics[width=\textwidth]{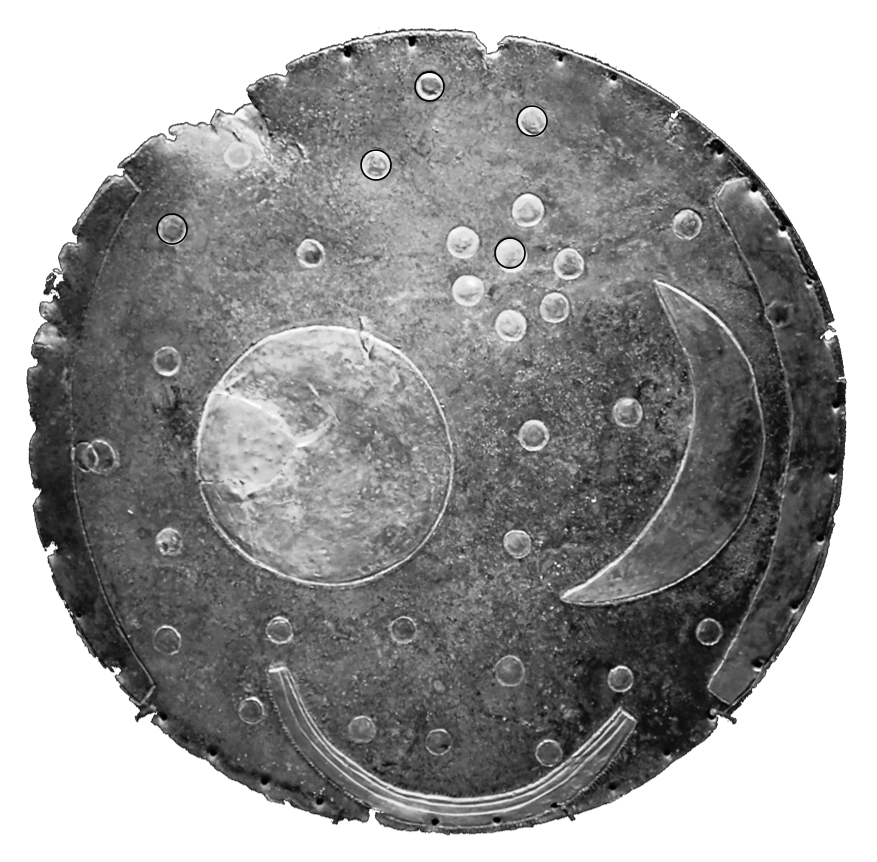}
\end{center}
\caption{The best-fitting circles around the central star of the Pleiades and the 4 stars of the Auriga line.} 
\label{kreisfit}
\end{figure}
}

\begin{table}[t] 
\begin{center}
\begin{tabular}{|l|l|l|}
\hline
 line & length (pixels) & relative length\\
\hline
$\beta$ Aur -- $\alpha$ Aur & 198.682 & 1.0 \\
$\beta$ Aur -- $\theta$ Aur & 174.563 & 0.87860166\\
$\epsilon$ Gem -- $\theta$ Aur & 392.814 & 1.97709658\\
$\epsilon$ Gem -- $\alpha$ Aur & 690.741 & 3.47660818\\
\hline
dist($\theta$ Aur , $\alpha$ Aur -- $\epsilon$ Gem) & 6.08887 & 0.0306463\\
\hline
\multicolumn{2}{|l|}{$\sphericalangle(\alpha \text{ Aur} \rightarrow\epsilon\text{ Gem} \,,\, \beta\text{ Aur}\rightarrow\text{center Pleiades})$} & $99.10772^\circ$\\
\hline
\end{tabular}
\vspace{3mm}
\caption{Distances and angles on the Sky Disc}
\label{tab22}
\end{center}
\end{table}

We denote the midpoints $m$ given in Table \ref{tab21a} by $m_\alpha$, $m_\beta$, $m_\theta$, $m_\epsilon$, $m_\text{plei}$ and see them as the centers of the corresponding star tile. ($m_\text{plei}$ represents the center star of the Pleiades Rosette.)

We then calculate the length of a line between two points $m_u = (u_1, u_2)$ and $m_v = (v_1, v_2)$ using the elementary formula
\begin{equation}
  d(m_u , m_v) = \sqrt{(v_1 - u_1)^2 + (v_2 - u_2)^2}\,.
  \end{equation}
Table \ref{tab22} shows the results for the center points considered above. We calculate the relative distances by dividing all absolute distances by the distance 198.682 pixels from $\alpha$ Aur and $\beta$ Aur.

We must also show that the star $\theta$ {Aur} lies, to a sufficient approximation, on the line through $\alpha$ {Aur} and $\epsilon$ {Gem}. Only then is the structure formed from the 4 stars of the Auriga line uniquely determined.

For this we use the Hessian normal form of a line
\begin{equation}
\frac{a x + b y - c}{\sqrt{a^2 + b^2}} = 0\;. \label{b10a}
\end{equation}
From the direction vector from $m_\alpha$ to $m_\epsilon$
\begin{equation*}
m_\epsilon - m_\alpha = (-661.492,-198.874)
\end{equation*}
we determine the normal vector of the line through $m_\alpha$ and $m_\epsilon$
\begin{equation}
\vec{n} = (-198.874,661.492)\quad , \quad\sqrt{198.874^2+661.492^2} = 690.741 \label{b11}
  \end{equation}
as well as the quantity $c$ in (\ref{b10a})
\begin{equation}
  c = \vec{n}\cdot m_\alpha = 706072.29 \label{b12}
\end{equation}
With (\ref{b11}), (\ref{b12}) we get for the Hessian normal form (\ref{b10a})
\begin{equation}
-0.287914 x +  0.957656 y - 1022.20 = 0. \label{b13}
\end{equation}
If we insert $m_\theta$ into (\ref{b13}), we obtain for the distance from $m_\theta$ to the line through $m_\alpha$, $m_\epsilon$
\begin{equation}
  d = 6.08887 \text{pixel}\quad, \quad d/198.682 = 0.0306463\;.
  \end{equation}
Thus, $\theta$ Aur lies to a good approximation on the line through $\alpha$ Aur and $\epsilon$ Gem.

To determine the angle between the vector $\alpha\text{ Aur}\rightarrow\epsilon\text{ Gem}$ and the vector $\beta\text{ Aur}\rightarrow\text{center Pleiades}$, we first calculate the unit vectors
\begin{eqnarray}
  \vec{a} &=& \frac{m_\epsilon - m_\alpha}{|m_\epsilon - m_\alpha|} \;=\; (-0.957656,-0.287914) \label{b10}\\
  \vec{b} &=& \frac{m_\text{plei} - m_\beta}{|m_\text{plei} - m_\beta|} \;=\; (0.435872,-0.900008)
\end{eqnarray}
Then we determine the angle $\gamma$ between $\vec{a}$ and $\vec{b}$ from the scalar product
\begin{eqnarray}
  \vec{a}\cdot\vec{b} &=& -0.158291 \\
  \gamma &=& \arccos(\vec{a}\cdot\vec{b}) \;=\; 99.10772^\circ\;.
\end{eqnarray}

If we compare the relative lengths in the Tables \ref{tab11} and \ref{tab22}, we see
\begin{Prop}
  If we round the relative lengths in Tables {\rm\ref{tab11}} and {\rm\ref{tab22}} to $2$ digits after the decimal point, the corresponding numbers agree up to $\pm 0.02$.
\end{Prop}
This is a good confirmation of the relative lengths in Table \ref{tab11} by the values in Table \ref{tab22} determined using a more precise method.

The values of the angle $\sphericalangle(\alpha \text{ Aur} \rightarrow\epsilon\text{ Gem} \,,\, \beta\text{ Aur}\rightarrow\text{center Pleiades})$  of both tables are also close to each other. However, since the best fitting circle of the center star of the Pleiades rosette is not as good as the fitting circles of the four stars of the Auriga line, we determine this angle again using a second method.
\mycomment{
\begin{figure}[!ht]
  \begin{center}
    \includegraphics[width=\textwidth]{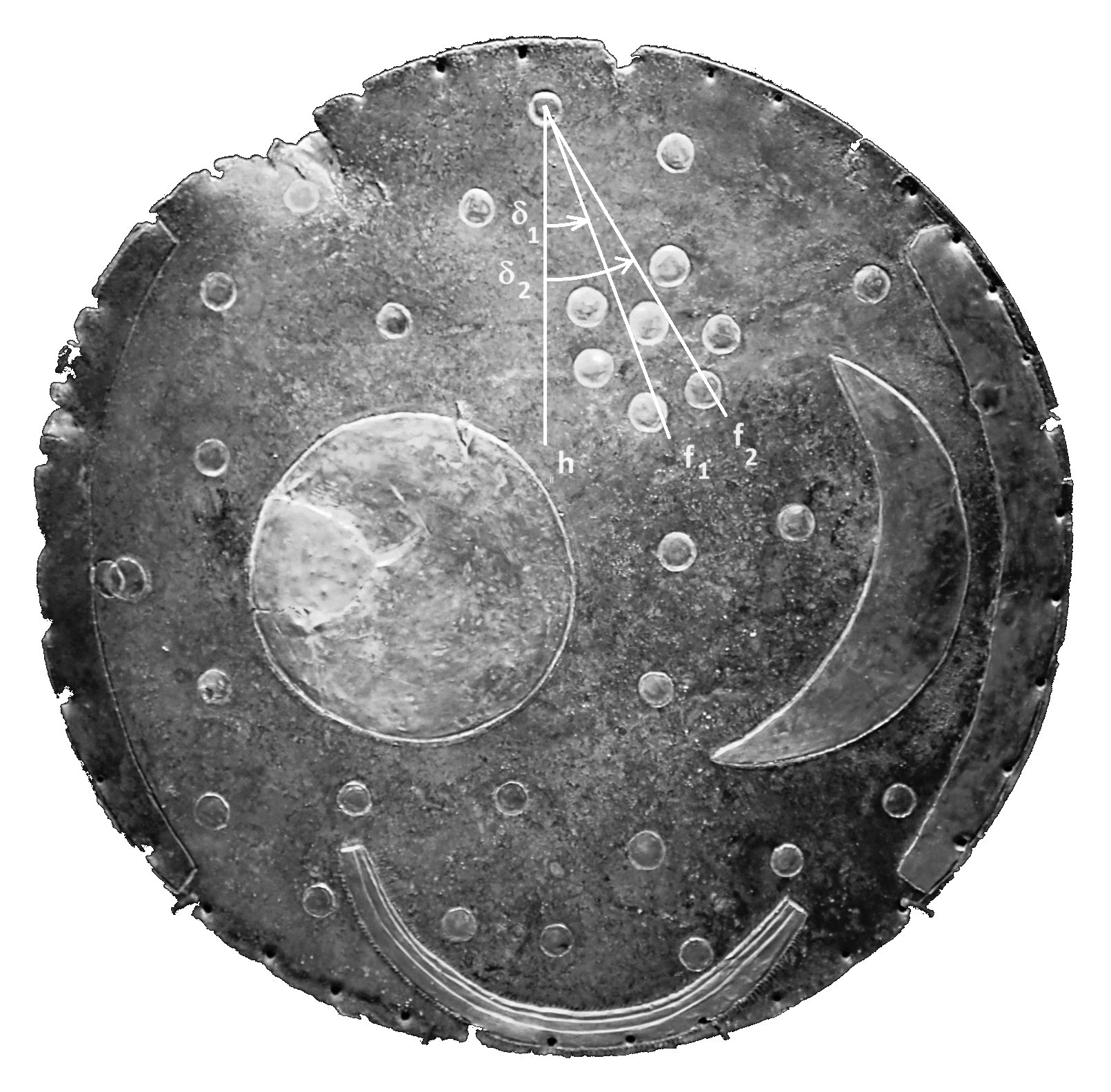}
\end{center}
\caption{For the second determination of the direction from $\beta$ Aur to the center of the Pleiades Rosette.} 
\label{betaplej}
\end{figure}
}

In the photo of the Sky Disk in GIMP \cite{gimp} we draw the two tangents $f_1$, $f_2$ to the center star of the Pleiades rosette from the center of $\beta$ Aur (see Figure \ref{betaplej}). We construct the center of the center star in GIMP as follows: We determined the coordinates (in pixel) $m_{\beta} = (790.6052 , 1425.4896)$ for this center (see Table \ref{tab21a}). From this we calculate the coordinates in GIMP
\begin{eqnarray}
  x &=&790.6052 \;\sim\; 791 \text{pixel} \\
  y &=&1584 - 1425.4896 = 158,5104 \;\sim\; 159 \text{pixel}\,.
\end{eqnarray}
The transformation of the y-coordinate is necessary because in GIMP the coordinate origin is not in the lower left corner of the image, as in Mathematica, but in the upper left corner of the image. Furthermore, the image used has 1584 pixel lines.

Now we draw a horizontal guide line to the y-value 159 and a vertical guide line to the x-value 791. Their intersection is the center point we are looking for. We mark it with the GIMP drawing pen (see Figure \ref{betabigg}).
\mycomment{
\begin{figure}[!ht]
  \begin{center}
    \includegraphics[width=8cm]{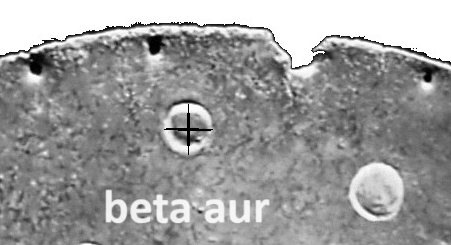}
\end{center}
\caption{The center of the star $\beta$ Aur.} 
\label{betabigg}
\end{figure}
}

We now draw the tangents $f_1$, $f_2$ and use GIMP to determine the angles $\delta_1$, $\delta_2$ of these tangents to the image vertical $h$. We calculate from $\delta_1$, $\delta_2$
\begin{eqnarray}
\delta &=& \frac{\delta_1 + \delta_2}2\,.
\end{eqnarray}
$\delta$ is also a good value for the angle of direction from $\beta$ Aur to the center star of the Pleiades.

We want to determine a confidence interval for the arithmetic mean $\bar{\delta}$ of $\delta$-values. As with the creation of Table \ref{tab33}, we carry out the measurement and calculation of $\delta_1$, $\delta_2$ and $\delta$ 31 times. Table \ref{tab23} shows the results $d_{1,i}$, $d_{2,i}$ and $d_i$.

\begin{table}[t] 
\begin{center}
\begin{tabular}{|l|}
    \hline
    \multicolumn{1}{|c|}{Measurements $d_{1,i}$ for $\delta_1$}\\
\hline
21.180, 20.980, 20.630, 20.730, 20.960, 20.790, 21.270, 21.000, 21.050, 20.870,\\
20.910, 21.180, 20.930, 21.140, 20.830, 21.110, 20.900, 20.700, 20.850, 21.080,\\
21.030, 21.100, 20.990, 20.900, 20.910, 20.980, 21.000, 21.190, 20.850, 20.950,\\ 20.780 \\
   \hline
  \end{tabular}
\\*[3mm]
\begin{tabular}{|l|}
    \hline
    \multicolumn{1}{|c|}{Measurements $d_{2,i}$ for $\delta_2$}\\
\hline
30.520, 30.860, 30.480, 30.380, 30.370, 30.720, 30.700, 30.620, 31.010, 30.600,\\
30.750, 30.610, 30.490, 30.440, 30.580, 30.790, 30.460, 30.220, 30.510, 30.870,\\
30.690, 30.780, 30.750, 30.750, 30.790, 30.870, 30.290, 30.600, 30.660, 30.580,\\ 30.600\\
   \hline
  \end{tabular}
\\*[3mm]
\begin{tabular}{|l|}
    \hline
    \multicolumn{1}{|c|}{Values $d_i$ for $\delta = (\delta_1 + \delta_2)/2$}\\
    \hline
    25.850, 25.920, 25.555, 25.555, 25.665, 25.755, 25.985, 25.810, 26.030, 25.735,\\
    25.830, 25.895, 25.710, 25.790, 25.705, 25.950, 25.680, 25.460, 25.680, 25.975,\\
    25.860, 25.940, 25.870, 25.825, 25.850, 25.925, 25.645, 25.895, 25.755, 25.765,\\
    25.690\\
   \hline
   \rule{0pt}{12pt}$\bar{d} = 25.7921$\\
   $s = 0.137433$\\ 
   $s_{\bar{d}} = 0.0246838$\\
\hline
  \end{tabular}
\vspace{3mm}
\caption{Data for determining the confidence interval (\ref{b18}) from {\tt Nebra\_Scheibe\_white.jpg}. All quantities in degrees.}
\label{tab23}
\end{center}
\end{table}
From this we calculate again $\bar{d}$, $s$ and $s_{\bar{d}}$ using (\ref{b1}).

We again choose a confidence level $1-p = 0.95$. This leads to the same constant $c = 2.04227$ as in (\ref{b4}), since $c$ only depends on $n$ and $p$.

If we now calculate $k = c\cdot s_{\bar{d}}$ according to (\ref{b5}), we get
\begin{equation}
k \;=\; 0.050411^\circ \,. \label{b22}
\end{equation}
From this follows for the confidence interval $\bar{d}-k \le \bar{\delta}\le\bar{d}+k$
\begin{equation}
  25.7417^\circ \le\bar{\delta}\le 25.8425^\circ\,. \label{b18}
\end{equation}

Now we want to calculate the limits within which
\begin{equation}
  \gamma := \sphericalangle(\alpha\text{ Aur}\rightarrow\epsilon\text{ Gem}\,,\,\beta\text{ Aur}\rightarrow\text{ center Pleiades}) \label{b18a}
  \end{equation}
moves when $\bar{\delta}$ moves in the confidence interval (\ref{b18}).

First, we determine the angle between the unit vector $\vec{a}$ according to (\ref{b10}) and the downward pointing image vertical.
Obviously, this angle is calculated by
\begin{equation}
  \arccos(\vec{a}\cdot (0, -1)) = 73.2669^\circ \label{b119}
\end{equation}
From Figure \ref{betaplej} we immediately see that we only have to add the angle (\ref{b119}) to the limits of the interval (\ref{b18}) to obtain the desired interval (\ref{b120}) for $\gamma$. This gives

\begin{Prop} \label{propb3}
  If the angle $\bar{\delta}$ moves in the confidence interval {\rm (\ref{b18})}, then the angle $\gamma$ moves in the interval
\begin{equation}
99.0086^\circ\le\gamma\le 99.1094^\circ \;,\label{b120}
\end{equation}
whereby the movements of both angles are in the same direction. Furthermore, the midpoint $\bar{d}$ of the confidence intervall corresponds to the midpoint
\begin{equation}
\gamma = 99.058999^\circ\;. \label{b121}
\end{equation}
of {\rm (\ref{b120})}.
\end{Prop}
In the interval (\ref{b120}) $\gamma$ will lie with probability of 95\%. Since (\ref{b120}) has a very small length and the value of $\gamma$ from Table \ref{tab22} is also in (\ref{b120}), we can say that both methods of determining $ \gamma$ led to approximately the same results.

\subsubsection{Measurements in the image {\tt ``Die Himmelsscheibe von Nebra, , \\
979x936px, \textcopyright\;LDA Sachsen-Anhalt (Foto Juraj Lipt\'ak).jpg''
}} \label{appB1.2}

In this section we carry out the same investigations on the image mentioned in the title as on the image Nebra\_Scheibe\_white.jpg in Section \ref{appB1.1}. However, as mentioned earlier, we cannot show the image in our paper for copyright reasons. To shorten the long name of the picture, we will also call it {\it Lipt\'ak's photo}.\\*[2mm]
\paragraph{a) Check for possible image distortions}

\begin{table}[t] 
\begin{center}
\begin{tabular}{|l|}
    \hline
    \multicolumn{1}{|c|}{Measurements $a_{1,i}$ for $\alpha_1$}\\
\hline
37.57, 37.84, 37.82, 38.05, 37.84, 37.73, 37.76, 38.02, 37.85, 37.92, 37.80, \\
38.02, 37.96, 37.97, 38.04, 37.73, 37.79, 37.80, 37.72, 37.90, 37.90, 37.62, \\
37.85, 37.64, 37.81, 37.93, 37.63, 37.89, 37.73, 37.54, 37.73\\
   \hline
  \end{tabular}
\\*[3mm]
\begin{tabular}{|l|}
    \hline
    \multicolumn{1}{|c|}{Measurements $a_{2,i}$ for $\alpha_2$}\\
\hline
43.96, 43.79, 43.96, 43.87, 44.24, 43.88, 44.08, 43.89, 43.99, 43.76, 44.00, \\
43.88, 43.93, 43.77, 44.06, 44.06, 43.95, 43.78, 44.03, 43.85, 43.97, 43.98, \\ 44.21, 43.83, 43.90, 43.85, 43.98, 43.89, 43.66, 44.07, 43.85\\
   \hline
  \end{tabular}
\\*[3mm]
\begin{tabular}{|l|}
    \hline
    \multicolumn{1}{|c|}{Values $a_i$ for $\alpha = \alpha_1 + \alpha_2$}\\
\hline
81.53, 81.63, 81.78, 81.92, 82.08, 81.61, 81.84, 81.91, 81.84, 81.68, 81.80,\\
81.90, 81.89, 81.74, 82.10, 81.79, 81.74, 81.58, 81.75, 81.75, 81.87, 81.60,\\
82.06, 81.47, 81.71, 81.78, 81.61,81. 78, 81.39, 81.61, 81.58 \\
   \hline
   $\bar{a} = 81.7523$\\
   $s = 0.171069$\\
   $s_{\bar{a}} = 0.0307250$\\
\hline
  \end{tabular}
\vspace{3mm}
\caption{Data for determining the confidence interval (\ref{b30}) from {\tt Lipt\'ak's photo}. All quantities in degrees.}
\label{tab23a}
\end{center}
\end{table}
First, we determine a confidence interval for the angle of the horizon arcs (Figure \ref{horizontwinkel}) to check that the image does not contain any major distortions. To do this, we measure each of the angles $\alpha_1$ and $\alpha_2$ (Figure \ref{horizontwinkel2}) $n = 31$ times and calculate the values $a_i$ for the angle $\alpha = \alpha_1 + \alpha_2$ from these measured values . We then use (\ref{b1}) to determine the arithmetic mean $\bar{a}$, the standard deviation $s$ and the standard deviation $s_{\bar{a}}$ to the $a_i$. Table \ref{tab23a} shows the results.

We again choose a confidence level $1-p = 0.95$. This means that we also get the same value (\ref{b28}) for the constant $c$ as in (\ref{b4}).
\begin{equation}
  c = 2.04227 \label{b28}
\end{equation}
We now calculate from $c$ and $s_{\bar{a}}$
\begin{equation}
k = c\cdot s_{\bar{a}} = 0.0627488^\circ\;.
\end{equation}
This results in the limits for the confidence interval $\bar{a}-k\le\bar{\alpha}\le\bar{a}+k$ for $\bar{\alpha}$
\begin{equation}
  81.6895^\circ \le \bar{\alpha} \le 81.8150^\circ\;. \label{b30}
\end{equation}

\mycomment{
\begin{figure}[!ht]
  \begin{center}
    \includegraphics[width=\textwidth]{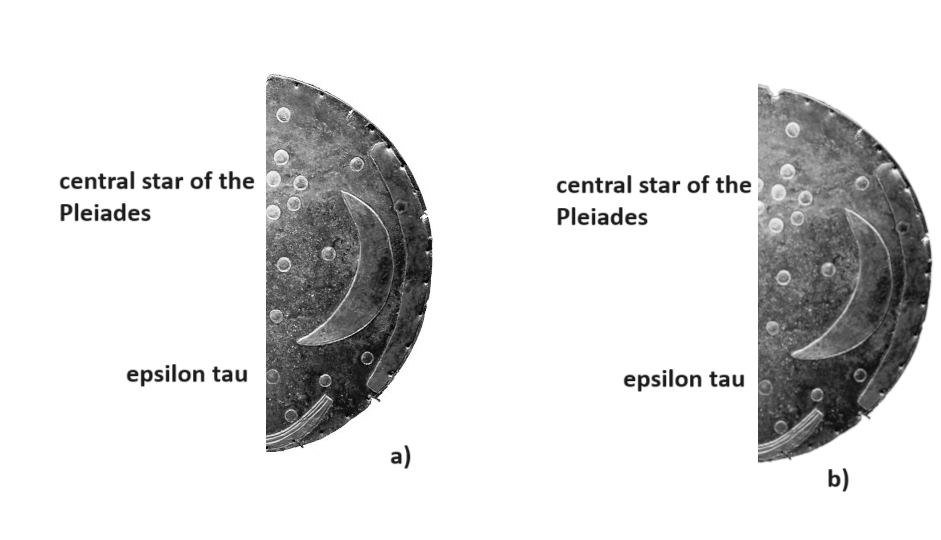}
\end{center}
\caption{The slope of the line $\epsilon$ Tau $\rightarrow$ central star of the Pleiades a) in Nebra\_Scheibe\_white.jpg and b) in the photo by Lipt\'ak, represented here by the image {\tt Nebra\_Scheibe\_white} rotated by $+3.47^\circ$.} 
\label{verdrehung}
\end{figure}
}Since $\bar{\alpha} = 81.7523^\circ$ is also close to the interval $82^\circ \le\alpha\le 83^\circ$, the photo considered in the current section cannot contain any major distortions.

Nevertheless, we want to investigate where the deviation between the $\bar{a}$-values in Table \ref{tab33} and Table \ref{tab23a} could come from.  For this purpose we want to measure the vertical and horizontal diameters of the two disc images using GIMP \cite{gimp}.

Before we can do this, however, we need to rotate Lipt\'ak's image slightly. In Figure \ref{verdrehung} a) one can see that in Nebra\_Scheibe\_white.jpg the star $\epsilon$ Tau and the central star of the Pleiades touch a vertical line, while in Lipt\'ak's picture the central star of the Pleiades deviates from a perpendicular tangent applied to $\epsilon$ Tau (see Figure \ref{verdrehung} b).
                                                                                                                                                                                                                               
To rotate Lipt\'ak's photo using GIMP, we first place a vertical guide line from leftfa to $\epsilon$ Tau. The guide remains fixed when the image is rotated. We click on the point where the star disk touches the auxiliary line, read the coordinates of the point of contact and enter them as the pivot point in the GIMP rotation tool. Then we rotate the disk photo with the mouse until the central stern of the Pleiades also touches the auxiliary line. This results in a rotation angle of $-3.47^\circ$.

After the orientation of the two disc images is the same, we cut off the edges of both images so that the new horizontal and vertical edges touch the image of the Sky Disc. Table \ref{tab24} shows the dimensions of the resulting photos.
\begin{table}[h] 
\begin{center}
\begin{tabular}{|l|c|}
\hline
 image & dimensions (pixels) \\
\hline
Nebra\_Scheibe\_white.jpg & $1510\times 1474$
\\
Lipt\'ak's photo & $977\times 934$\\
\hline
\end{tabular}
\vspace{3mm}
\caption{The dimensions of the cut photos.}
\label{tab24}
\end{center}
\end{table}

We calculate from Table \ref{tab24}
\begin{eqnarray}
  1474/1510 &=& 0.976159 \nonumber\\
  934/977   &=& 0.955988 \nonumber
\end{eqnarray}
The photo Nebra\_Scheibe\_white.jpg appears to be stretched a little more in the vertical direction than Lipt\'ak's photo. This explains why Nebra\_Scheibe\_white.jpg has a slightly larger horizon arc angle.\\*[0.2cm]
\paragraph{b) Determination of distances and angles in the second Auriga line on the Sky Disc}
The image considered in the present section has $979\times 936$ pixels, while the image in the previous Section \ref{appB1.1} has a size of $1611\times 1584$ pixels. We will therefore obtain completely different values for the coordinates of image points than in Section \ref{appB1.1}.

\begin{table}[p] 
\begin{center}
  \begin{tabular}{|r|r|r|}
    \hline
    \multicolumn{3}{|c|}{Circle around $\alpha$ Aur}\\
\hline
  x & y & d \\
 \hline
 616.6182 & 823.7906 & 0.0\\
 617.9266 & 842.7635 & 0.0\\
 638.8600 & 821.1700 & 0.1140\\
 641.4800 & 841.4600 & 0.2726\\
 612.0400 & 833.6000 & 1.8763\\
 625.7800 & 816.5900 & 1.0861\\
 644.1000 & 832.9500 & 0.0\\
 630.3571 & 850.6143 & 3.0822\\
   \hline
   \multicolumn{3}{|l|}{$m = (628.9931, 832.4688)$}\\
   \multicolumn{3}{|l|}{$r = 15.1145$}\\
\hline
  \end{tabular}
  
  \vspace{3mm}
    \begin{tabular}{|r|r|r|}
    \hline
    \multicolumn{3}{|c|}{Circle around $\beta$ Aur}\\
\hline
  x & y & d \\
 \hline
 502.0500 & 872.1800 & 0.0 \\
 508.4700 & 888.2300 & 0.0047\\
 524.5200 & 883.7300 &,0.0 \\
 523.2400 & 868.3300 & 0.1315\\
 499.4800 & 880.1000 & 2.1028\\
 513.6100 & 864.0500 & 0.0 \\
 529.0200 & 876.8900 & 2.4122\\
 516.8174 & 890.3684 & 1.2208\\
   \hline
   \multicolumn{3}{|l|}{$m = (513.9038, 876.7512)$}\\
   \multicolumn{3}{|l|}{$r = 12.7046$}\\
\hline
\end{tabular}
\hspace*{5mm}
  \begin{tabular}{|r|r|r|}
    \hline
    \multicolumn{3}{|c|}{Circle around $\theta$ Aur}\\
\hline
  x & y & d \\
 \hline
 437.9800 & 787.7800 & 1.2719\\
 437.9788 & 806.0241 & 0.0\\
 458.1800 & 785.1700 & 0.2771\\
 458.8349 & 802.7653 & 2.3719\\
 448.4068 & 781.2574 & 0.0\\
 432.1130 & 795.5960 & 1.8384\\
 449.0586 & 810.5864 & 0.0\\
 463.3972 & 793.6408 & 0.2873\\
   \hline
   \multicolumn{3}{|l|}{$m = (448.6167, 795.9245)$}\\
   \multicolumn{3}{|l|}{$r = 14.6686$}\\
\hline
\end{tabular}

  \vspace{3mm}
    \begin{tabular}{|r|r|r|}
    \hline
    \multicolumn{3}{|c|}{Circle around $\epsilon$ Gem}\\
\hline
  x & y & d \\
 \hline
 199.6600 & 754.4800 & 0.0\\
 194.5288 & 729.4381 & 0.1868\\
 215.0731 & 726.8701 & 1.9691\\
 222.7772 & 744.2043 & 1.3301\\
 191.9607 & 743.5623 & 0.3997\\
 203.5169 & 723.6600 & 0.0\\
 222.7772 & 732.0061 & 0.0\\
 211.8600 & 757.6900 & 1.4580\\
   \hline
   \multicolumn{3}{|l|}{$m = (207.9336, 739.8641)$}\\
   \multicolumn{3}{|l|}{$r = 16.7952$}\\
\hline
\end{tabular}
\hspace*{5mm}
  \begin{tabular}{|r|r|r|}
    \hline
    \multicolumn{3}{|c|}{Center star of the Pleiades Rosette}\\
\hline
  x & y & d \\
 \hline
 583.5900 & 701.4000 & 0.0\\
 582.9400 & 677.0100 & 0.0\\
 605.4147 & 673.1553 & 0.5807\\
 611.8348 & 696.9096 & 0.0\\
 578.4500 & 690.4900 & 0.0074\\
 597.0700 & 671.2300 & 0.0322\\
 612.4768 & 682.7854 & 0.2032\\
 597.0700 & 707.1799 & 0.6495\\
   \hline
   \multicolumn{3}{|l|}{$m = (596.0763, 688.8636)$}\\
   \multicolumn{3}{|l|}{$r = 17.6938$}\\
\hline
\end{tabular}
  \vspace{3mm}
\caption{Data of the stars of the Auriga line and the center of the Pleiades rosette: $x\,,\,y$ coordinates of the 8 edge points of the star plate, $d$ distance of the relevant edge point from the fitting circle of the star, $m\,,\, r$ center and radius of the fitting circle. Everything measured in pixels.}
\label{tab25}
\end{center}
\end{table}
We now carry out the same calculations for Lipt\'ak's photo as in Section \ref{appB1.1} for the image Nebra\_Scheibe\_white.jpg. First we determine the centers of the star plates we are interested in using fitting circles. Table \ref{tab25} shows the results. It is interesting to note that in Lipt\'ak's photo the fitting circle of the center star of the Pleiades fits much better to the 8 control points than in Nebra\_Scheibe\_white.jpg. This can be clearly seen in the residuals $d$ of this circle in the Tables \ref{tab25} and \ref{tab21a}.

We then calculate the distances of these centers, the distance of the star $\theta$ Aur from the line $\alpha$ Aur -- $\epsilon$ Gem and the angle between the lines $\alpha$ Aur $\rightarrow \epsilon$ Gem and $\beta$ Aur $\rightarrow$ center Pleiades using the same elementary geometric means as in Section \ref{appB1.1}. The results are listed in Table \ref{tab26}. Furthermore, we calculate all relative distances by dividing all absolute distances determined so far by the distance 123.315 pixels from $\alpha$ Aur and $\beta$ Aur.

\begin{table}[t] 
\begin{center}
\begin{tabular}{|l|l|l|}
\hline
 line & length (pixels) & relative length\\
\hline
$\beta$ Aur -- $\alpha$ Aur & 123.315 & 1.0 \\
$\beta$ Aur -- $\theta$ Aur & 103.901 & 0.84256623\\
$\epsilon$ Gem -- $\theta$ Aur & 247.126 & 2.00402568\\
$\epsilon$ Gem -- $\alpha$ Aur & 431.123 & 3.49611899\\
\hline
dist($\theta$ Aur , $\alpha$ Aur -- $\epsilon$ Gem) & 3.05339 & 0.0247609\\
\hline
\multicolumn{2}{|l|}{$\sphericalangle(\alpha \text{ Aur} \rightarrow\epsilon\text{ Gem} \,,\, \beta\text{ Aur}\rightarrow\text{center Pleiades})$} & $101.2184^\circ$\\
\hline
\end{tabular}
\vspace{3mm}
\caption{Distances and angles on the Sky Disc}
\label{tab26}
\end{center}
\end{table}

Finally, we calculate a second value for the angle $\gamma = \sphericalangle(\alpha \text{ Aur} \rightarrow\epsilon\text{ Gem} \,,\, \beta\text{ Aur}\rightarrow\text{center Pleiades})$ over a confidence interval. To do this, we measure each of the angles $\delta_1$ and $\delta_2$ (see Figure \ref{betaplej}) 31 times and calculate 31 values for $\delta = (\delta_1 + \delta_2)/2$ from these measured values (see Table \ref{tab27}). For the 31 $\delta$-values $d_i$, we use (\ref{b1}) to determine the arithmetic mean $\bar{d}$, the variance $s$ and the variance $s_{\bar{d}}$ of the mean. These values are also given in Table \ref{tab27}.

\begin{table}[t] 
\begin{center}
\begin{tabular}{|l|}
    \hline
    \multicolumn{1}{|c|}{Measurements $d_{1,i}$ for $\delta_1$}\\
\hline
20.910, 21.370, 21.100, 20.860, 21.360, 21.020, 21.270, 21.460, 21.110, 21.080, \\
20.960, 21.250, 21.010, 21.280, 21.320, 21.380, 21.320, 21.080, 21.290, 21.250, \\
20.860, 21.130, 21.210, 20.990, 21.130, 21.110, 21.070, 20.990, 21.260, 21.130, \\
21.340\\
   \hline
  \end{tabular}
\\*[3mm]
\begin{tabular}{|l|}
    \hline
    \multicolumn{1}{|c|}{Measurements $d_{2,i}$ for $\delta_2$}\\
\hline
30.610, 30.940, 30.580, 31.030, 30.760, 30.680, 30.730, 30.830, 30.520, 30.670,\\
30.250, 30.730, 30.990, 30.710, 30.920, 30.670, 31.100, 30.680, 30.510, 30.660,\\
30.560, 30.870, 31.060, 30.620, 30.630, 30.710, 30.500, 30.890, 30.680, 30.730,\\
30.730\\
   \hline
  \end{tabular}
\\*[3mm]
\begin{tabular}{|l|}
    \hline
    \multicolumn{1}{|c|}{Values $d_i$ for $\delta = (\delta_1 + \delta_2)/2$}\\
    \hline
25.760, 26.155, 25.840, 25.945, 26.060, 25.850, 26.000, 26.145, 25.815, 25.875,\\
25.605, 25.990, 26.000, 25.995, 26.120, 26.025, 26.210, 25.880, 25.900, 25.955,\\
25.710, 26.000, 26.135, 25.805, 25.880, 25.910, 25.785, 25.940, 25.970, 25.930,\\
26.035\\
   \hline
   \rule{0pt}{12pt}$\bar{d} = 25.9427$\\
   $s = 0.137476$\\ 
   $s_{\bar{d}} = 0.0246915$\\
\hline
  \end{tabular}
\vspace{3mm}
\caption{Data for determining the confidence interval (\ref{b37}) from {\tt Lipt\'ak's photo}. All quantities in degrees.}
\label{tab27}
\end{center}
\end{table}
We again choose the confidence level $1-p = 0.95$. This leads us in this calculation instead of (\ref{b22}) to
\begin{equation}
  k = 0.0504267
\end{equation}
and to the confidence intervall
\begin{equation}
  25.8923^\circ\le\bar{\delta}\le 25.9932^\circ\;. \label{b37}
\end{equation}
For the vector (\ref{b10}) we get
\begin{equation}
\vec{a} = (-0.976658, -0.214799)\;. \label{b38}
\end{equation}
We determine the angle between this unit vector $\vec{a}$ and the downward pointing image vertical.
\begin{equation}
  \arccos(\vec{a}\cdot (0, -1)) = 77.5963^\circ \;.
\end{equation}
We add this angle again to the limits of (\ref{b37}).
Thus, Proposition \ref{propb3} takes the form
\begin{Prop} \label{propb4}
  If the angle $\bar{\delta}$ moves in the confidence interval {\rm (\ref{b37})}, then the angle $\gamma$ according to {\rm (\ref{b18a})} moves in the interval
\begin{equation}
103.489^\circ\le\gamma\le 103.589^\circ \label{b40}
\end{equation}
whereby the movements of both angles are in the same direction. Furthermore, the midpoint $\bar{d}$ of the confidence intervall corresponds to the midpoint
\begin{equation}
\gamma = 103.539006^\circ\;. \label{b41}
\end{equation}
of {\rm (\ref{b40})}.
\end{Prop}

\subsection{Distances and angles for the Auriga line on the celestial sphere} \label{appB.2}
In this section, we calculate distances and angles directly on the celestial sphere itself. This way, we avoid distortions that arise when projecting the celestial sphere onto a plane (star map) and that could be contained in Figure \ref{vertical} and in Table \ref{tab12a}.

\mycomment{
\begin{figure}[!ht]
  \begin{center}
    \includegraphics[width=\textwidth]{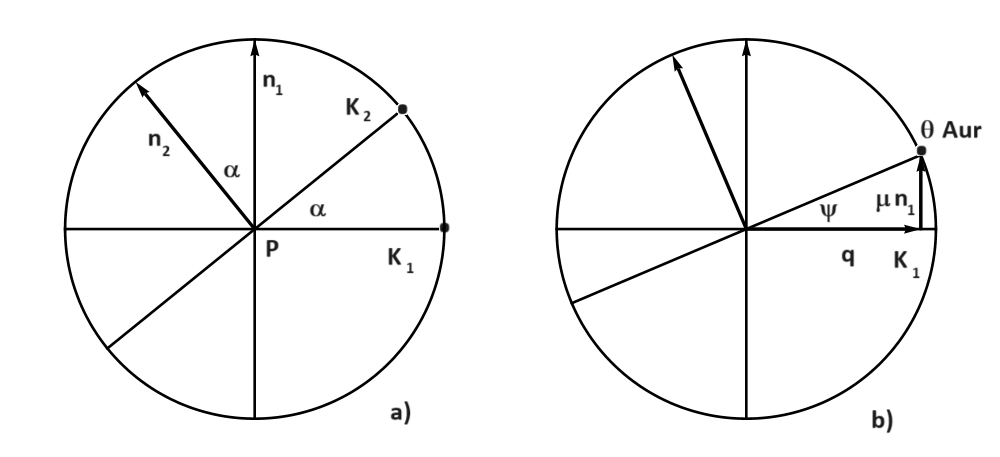}
\end{center}
\caption{Sketches for the calculations in Section \ref{appB.2}.} 
\label{diagramme}
\end{figure}
}

\begin{table}[t] 
\begin{center}
\begin{tabular}{|l|l|l|}
\hline
 star & azimuth $A$ & altitude $h$\\
\hline
$\alpha$ Aur & $249.9775^\circ$ & $58.3850^\circ$\\
$\beta$ Aur & $240.0972^\circ$ & $64.6370^\circ$ \\
$\theta$ Aur & $226.2044^\circ$ & $60.4202^\circ$ \\
$\epsilon$ Gem & $196.6126^\circ$ & $57.4953^\circ$ \\
$\eta$ Tau & $239.9961^\circ$ & $30.7554^\circ$ \\
\hline
\end{tabular}
\vspace{3mm}
\caption{Azimuthal coordinates of the stars of the Auriga line and the Pleiades center on -1601-2-28 Julian, 18:10:57 for the Mittelberg.}
\label{tab28a}
\end{center}
\end{table}

As a starting point, Stelarium provides us with the azimuthal coordinates of the stars $\alpha$ Aur, $\beta$ Aur, $\theta$ Aur, $\epsilon$ Gem and $\eta$ Tau on -1601-2-28Julian, 18:10:57 (see Table \ref{tab28a}).

From these values, we use the easy-to-prove formula
\begin{equation}
 \vec{e} = (\cos A \sin[90^\circ - h), \sin A \sin(90^\circ -h), \cos(90^\circ - h))
  \end{equation}
to calculate the unit vectors that point from the center of the celestial sphere to the position of the corresponding stars. We get
\begin{eqnarray}
  \vec{e}_\alpha & = & (-0.179483,-0.492525,0.85159) \nonumber\\
  \vec{e}_\beta &=&(-0.213546,-0.371326,0.903612) \nonumber\\
  \vec{e}_\theta &=& (-0.341639,-0.356313,0.869669) \label{b42}\\
 \vec{e}_\epsilon &=& (-0.514939,-0.153633,0.843347) \nonumber\\
 \vec{e}_\eta &=& (-0.429730,-0.744197,0.511374)\nonumber
  \end{eqnarray}
The distance $d(p_1, p_2)$ between two points $p_1$, $p_2$ on a unit sphere surface is defined as the length of the smallest great circle segment connecting these points. The length of this great circle segment is equal to the angle (measured in radians) between the unit vectors $\vec{e}_1$, $\vec{e}_2$ pointing from the center of the sphere to the positions of $p_1$, $p_2$. We can thus calculate $d(p_1, p_2)$ using the formula
\begin{equation}
  d(p_1 , p_2) = \arccos(\vec{e}_1 \cdot \vec{e}_2) \;. \label{b43}
\end{equation}
If we apply formula (\ref{b43}) to the first 4 vectors in (\ref{b42}), we get the first 4 values in column 2 of Table \ref{tab29a}. If we divide these values by  0.136325, we get the relative distances in the third column.
\begin{table}[t] 
\begin{center}
  \begin{tabular}{|l|l|l|}
    
\hline
 line & length (radiants) & relative length\\
\hline
$\beta$ Aur -- $\alpha$ Aur & 0.136325 & 1.0 \\
$\beta$ Aur -- $\theta$ Aur & 0.133461 & 0.97899\\
$\epsilon$ Gem -- $\theta$ Aur & 0.268772 & 1.97156\\
$\epsilon$ Gem -- $\alpha$ Aur & 0.481552 & 3.53239\\
\hline
dist($\theta$ Aur , $\alpha$ Aur -- $\epsilon$ Gem) & 0.00634605 & 0.046551\\
\hline
\multicolumn{2}{|l|}{$\gamma = \sphericalangle(\alpha \text{ Aur} \rightarrow\epsilon\text{ Gem} \,,\, \beta\text{ Aur}\rightarrow\eta\text{ Tau})$} & $102.758^\circ$\\
\hline
\end{tabular}
\vspace{3mm}
\caption{Distances and angles to the Auriga line on the celestial sphere.}
\label{tab29a}
\end{center}
\end{table}

The angle of intersection of two great circles is equal to the angle between the normal vectors of the great circle planes (see Figure \ref{diagramme}a). There, the two great circles $K_1$, $K_2$, whose planes have the normal vectors $n_1$, $n_2$, intersect at the point $P$. It is easy to see that the intersection angle $\alpha$ of the great circles is equal to the angle between the normal vectors.

Now let $K_1$ be the great circle through $\alpha$ Aur and $\epsilon$ Gem and $K_2$ be the great circle through $\beta$ Aur and $\eta$ Tau. Their planes have the normal vectors
\begin{eqnarray}
  \vec{n}_1 &=& \frac{\vec{e}_\alpha\times\vec{e}_\epsilon}{|\vec{e}_\alpha\times\vec{e}_\epsilon|} = (-0.614345,-0.619986,-0.488056) \label{b44}\\
  \vec{n}_2 &=& \frac{\vec{e}_\beta\times\vec{e}_\eta}{|\vec{e}_\beta\times\vec{e}_\eta|} = (0.865644,-0.500659,-0.00116513)
  \end{eqnarray}
The angle between these unit vectors is
\begin{equation}
\gamma = \arccos(\vec{n}_1\cdot\vec{n}_2) = 102.758^\circ \;.
\end{equation}

We still have to calculate the distance from $\theta$ Aur to the great circle $K_1$ through $\alpha$ Aur and $\epsilon$ Gem. For this purpose, we decompose the direction unit vector $\vec{e}_\theta$ into its projection $\vec{q}$ onto the plane of $K_1$ and its part perpendicular to $K_1$ $\mu\vec{n}_1$ (see Figure \ref{diagramme}b).
\begin{equation}
\vec{e}_\theta = \vec{q} + \mu \vec{n}_1 \;. \label{b47}
  \end{equation}
To determine $\mu$, we multiply (\ref{b47}) scalarly with $\vec{n}_1$ and note
$$\vec{q}\cdot\vec{n}_1 = 0\quad,\quad\vec{n}_1\cdot\vec{n}_1 = 1\;.$$
\begin{eqnarray}
  \vec{e}_\theta\cdot\vec{n}_1 &=& \mu \vec{n}_1\cdot\vec{n}_1 \\
  \vec{e}_\theta\cdot\vec{n}_1 &=& \mu  \;.
  \end{eqnarray}
From (\ref{b44}) and (\ref{b42}) we get
\begin{equation}
\mu = 0.00634601\;.
\end{equation}
Obviously $\sin\psi = \frac{\mu}1 = \mu$.
\begin{equation}
  \psi = \arcsin\mu = 0.00634605\;.
\end{equation}
This is the absolute distance from $\theta$ Aur to $K_1$. Dividing this by 0.136325 gives the relative distance, which is given in Table \ref{tab29a}.

\begin{table}[t] 
\begin{center}
  \begin{tabular}{|l|c|c|}
    
\hline
  & Tab.\ref{tab22}--Tab.\ref{tab29a} &  Tab.\ref{tab26}--Tab.\ref{tab29a}\\
\hline
$\beta$ Aur -- $\alpha$ Aur & 0.0& 0.0\\
$\beta$ Aur -- $\theta$ Aur & -0.100388& -0.136423\\
$\epsilon$ Gem -- $\theta$ Aur & 0.0055396& 0.0324687\\
$\epsilon$ Gem -- $\alpha$ Aur & -0.0557807& -0.0362699\\
\hline
$\gamma$ from center coordinates & $-3.65027^\circ$& $-1.53960^\circ$\\
\hline
$\gamma$ from the center of the confidence interval& $-3.69899^\circ$& $0.781013^\circ$\\
\hline
\end{tabular}
\vspace{3mm}
\caption{Differences between the relative distances and angles on the Sky Disc and those on the celestial sphere.}
\label{tab30a}
\end{center}
\end{table}

We now determine the deviations between the relative distances and angles on the Sky Disc and those on the celestial sphere. To do this, we subtract the corresponding relative distances from Table \ref{tab29a} from the relative distances given in Tables \ref{tab22} and \ref{tab26}. Similarly, we subtract the $\gamma$-value from Table \ref{tab29a} from the $\gamma$-values given in Tables \ref{tab22}, \ref{tab26} and in (\ref{b121}), (\ref{b41}). Table \ref{tab30a} shows the results.

The deviations of the relative distances and angles of the second Auriga line on the Sky Disc from the corresponding values on the celestial sphere are so small that the relative distances and angles on the Sky Disc can be considered as measured.

\begin{Rem}
All calculations in Appendix \ref{appB} were performed using {\tt Mathematica} \cite{mma}. {\tt Mathematica} determines which digits of the number are accurate for each number generated in such a calculation and tracks error propagation during the calculation. One can then use the {\tt Accuracy} tool to display how many reliable digits the results have to the right of the decimal point.

All results obtained in Appendix \ref{appB} have at least 13 reliable digits to the right of the decimal point, even if we give fewer digits in our text.
  \end{Rem}

\section{About the horizon arcs on the Sky Disc} \label{appC}
In \cite{schlosser2,schlosser3} W. Schlosser gave an interpretation of the horizon arcs on the Sky Disk. The angle spanned by them is equal to the azimuth range in which the sunset points move within a year when viewed from a certain observation point $P$ on Earth. If one knows the setting point of the sun on the summer solstice at $P$ and one directs one end of the left horizon arc towards this point, then the other end of the left horizon arc points to the setting point of the sun on the winter solstice (see Figure \ref{horizontwinkel}).

In this way, the Nebra people were able to determine the occurrence of the winter solstice, which was a very important date as the shortest day of the year.

Schlosser specified the Brocken that can be seen from Mittelberg as the marker of the sunset on the summer solstice on the Mittelberg (assuming that the Mittelberg was unforested at the time of the Sky Disk). D. Lorenzen \cite{lorenzen1,lorenzen2} raised objections to this because he had calculated using PeakFinder \cite{peakf1} software, which was newly developed in 2010, that at -1700 on the summer solstice the sunset, viewed from Mittelberg, was not exactly behind the Brocken but about $2^\circ$ to the left of the Brocken.

\mycomment{
\begin{figure}[!ht]
  \begin{center}
    \includegraphics[width=\textwidth]{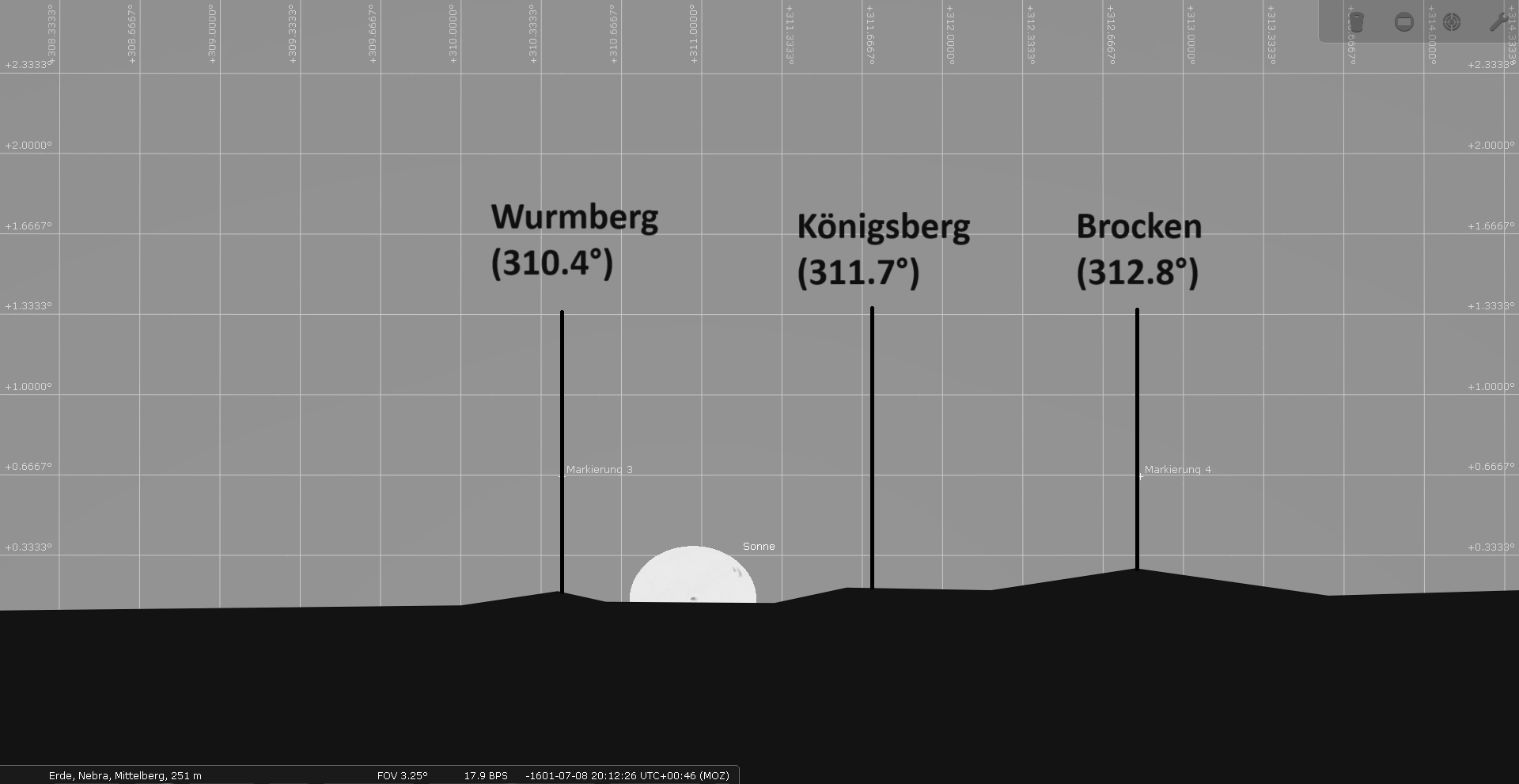}
\end{center}
\caption{Sunset at the summer solstice -1601 between Wurmberg and K\"onigsberg seen from Mittelberg.} 
\label{sunset}
\end{figure}
}

We checked this using PeakFinder and Stellarium. We simply determined the horizon line using PeakFinder and imported it into Stellarium. We then calculated the sunset on the summer solstice -1601 using Stellarium. It turned out that the sun sets in the depression between the Wurmberg and the K\"onigsberg to the left of the Brocken (Figure \ref{sunset}).

This confirms Lorenzen's result. But that doesn't mean that Schlosser's interpretation is no longer valid.

For the Nebra people, the Brocken probably did not have the same significance as it does for us today. They probably did not know that it is the highest mountain in the Harz. This meant that the depression between Wurmberg and K\"onigsberg was just as suitable as the Brocken itself as a bearing point.

But they could also have used a completely different landmark as a bearing point. Or they could have created an artificial bearing point using a gnomon (see Figure \ref{peilung}).

To do this, one sets up a gnomon and marks the end of the shadow at sunset in a period around the summer solstice. In the days near the summer solstice, the direction of the shadow at sunset no longer appears to change to someone observing only with the naked eye \cite[p. 23]{guenther}. This gives the direction to the point where the sun sets on the summer solstice.

If one now extends this line to the side of the gnomon facing away from the sun (100 meters is sufficient) and sets up a fixed table for the Sky Disk there, one can aim at the gnomon from this table and thus align the Disk.

\mycomment{
\begin{figure}[!ht]
  \begin{center}
    \includegraphics[width=\textwidth]{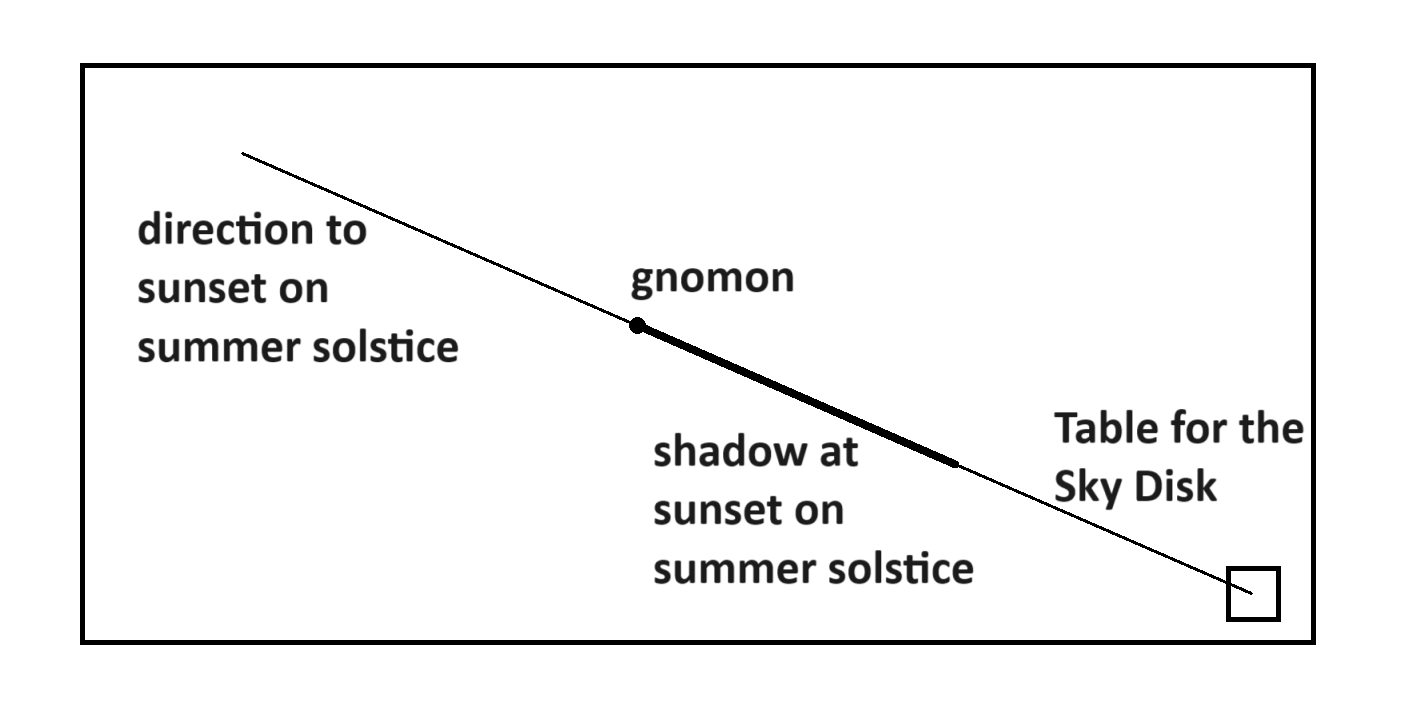}
  \end{center}
\caption{Finding arrangement for the position of the sunset at the summer solstice using gnomon.} 
\label{peilung}
\end{figure}
}

The fact that our calculations with PeakFinder and Stellarium led to the same result as Lorenzen's calculation with PeakFinder alone seems to indicate that PeakFinder is well suited for archaeoastronomical calculations. This is also seen as such in the PeakFinder description \cite{peakf2}. The reason for this is that the Archaeological Institute of the University of Basel has suggested incorporating archaeoastronomical functions into PeakFinder \cite[p.5]{peakf2}.

W. Schlosser calculated for the horizon arc angle of $82^\circ$--$83^\circ$ a strip of geographical latitude $\varphi$
$$51.8967^\circ \le \varphi \le 52.3959^\circ$$
in the northern hemisphere, on which the observation point $P$ must be located so that the method described at the beginning for determining the winter solstice works accurately.

Since our own measurements of the horizon arc angle $\bar{\alpha}$ in Appendix \ref{appB} resulted in $\bar{\alpha}$-values (Tables \ref{tab33} and \ref{tab23a}) and confidence intervals (\ref{B7}), (\ref{b30}) that do not lie in the interval between $82^\circ$ and $83^\circ$, we want to determine the ranges of latitudes belonging to these $\bar{\alpha}$-values here again. Our results are
\begin{Prop}
Let $\bar{\alpha} = 83.1881^\circ$ and $\bar{\alpha} = 81.7523^\circ$ be the $\bar{\alpha}$-values from the Tables {\rm\ref{tab33}} and {\rm\ref{tab23a}}. Consider an interval $\bar{\alpha}\pm 0.6^\circ$ around both values. Then the following latitude ranges belong to these intervals
\begin{eqnarray}
  83.1881^\circ \pm 0.6^\circ & \Rightarrow & 52.1581^\circ \le\varphi\le 52.6783^\circ \label{c1} \\
  81.7523^\circ \pm 0.6^\circ & \Rightarrow & 51.5075^\circ \le\varphi\le 52.0534^\circ \label{c2}
\end{eqnarray}
An observation location $P$ must be located in these areas so that the above method for determining the winter solstice works exactly.

Since the ring sanctuaries P\"ommelte and Sch\"onebeck have the geographical latitudes $\varphi = 51.997015^\circ$ and $\varphi =52.005628^\circ$ respectively, both lie in the interval {\rm (\ref{c2})}.
\end{Prop}

\begin{proof}
We start with the formulas
\begin{eqnarray}
  \tan H &=& \frac{\sin A}{\cos A \sin\varphi + \tan h \cos\varphi} \label{c3}\\
  \sin\delta &=& \sin\varphi \sin h - \cos\varphi \cos h \cos A \label{c4}
\end{eqnarray}
for the coordinate transformation $(A, h)\mapsto (H,\delta )$ from azimuthal coordinates to equatorial coordinates (\cite[p.94]{meeus}). The terms have the following meaning:
\begin{description}
\item[$H$] Hour angle of the sun
\item[$\delta$] Declination of the sun
\item[$A$] South azimuth of the sun
\item[$h$] Altitude of the sun
  \item[$\varphi$] Geographical latitude of the observation location $P$
\end{description}
At sunset, $h = 0$. This simplifies (\ref{c4}) to
\begin{equation}
  \sin\delta = - \cos A \cos\varphi\,. \label{c5}
\end{equation}
At the winter solstice, the sun at the observation point $P$ has the (unknown to us) azimuth $A = A_0$, at the summer solstice the azimuth $A = A_0 + \bar{\alpha}$, where $\bar{\alpha}$ is the horizon arc angle measured by us on the Sky Disk. Stellarium also provides the $\delta$-values for the points of the winter and summer solstices -1601 on the ecliptic
\begin{equation}
\delta_w = -23.8819^\circ \quad , \quad \delta_s = 23.8819^\circ \;. \label{c9}
\end{equation}
If we insert this into (\ref{c5}), we get the two equations
  \begin{eqnarray}
    -0.404853 &=& - \cos A_0 \cos\varphi \label{c6}\\
    0.404853 &=& - \cos( A_0 +\bar{\alpha}) \cos\varphi\;. \label{c7}
  \end{eqnarray}
  From these equations it follows that $\cos\varphi \not= 0$, i.e. $\varphi \not= \pm 90^\circ$.
  
  Given $\bar{\alpha}$ we have to determine $A_0$ and $\varphi$ from the system (\ref{c6}), (\ref{c7}).\\*[0.3cm]
(i) In Table \ref{tab33} we found $\bar{\alpha} = 83.1881^\circ$. The interval $\bar{\alpha} \pm 0.6^\circ$ to this value has the limits
\begin{equation}
82.5881^\circ \le\bar{\alpha}\le 83.7881^\circ\;. \label{c8}
\end{equation}
If we insert the lower limit of (\ref{c8}) into (\ref{c6}), (\ref{c7}) and add these two equations, we get
  \begin{equation}
    0 = (\cos A_0  + \cos( A_0 + 82.5881^\circ )) \cos\varphi
  \end{equation}
  and
    \begin{equation}
    0 = \cos A_0  + \cos( A_0 + 82.5881^\circ ) \label{c11}
    \end{equation}
    since $\cos\varphi \not= 0$.

    A picture of the function
\begin{equation}
  f(A_0) := \cos A_0  + \cos( A_0 + 82.5881^\circ ) \label{c12}
\end{equation}
shows that (\ref{c11}) has exactly one solution in $0^\circ \le A_0 \le 180^\circ$ (see Figure \ref{1loesung}).

\mycomment{
\begin{figure}[!ht]
  \begin{center}
    \includegraphics[width=8cm]{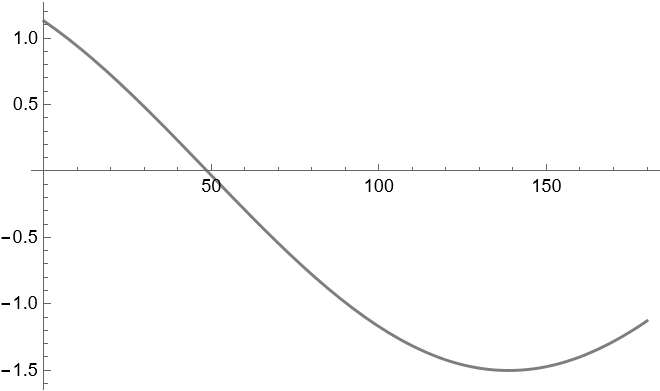}
  \end{center}
\caption{The function (\ref{c12}) over $0^\circ \le A_0 \le 180^\circ$.} 
\label{1loesung}
\end{figure}
}

We solve (\ref{c11}) with the tool {\tt Solve} from {\tt Mathematica} \cite{mma} and obtain
\begin{equation}
 A_0 = 48.7059^\circ \;. \label{c13}
  \end{equation}
If we insert (\ref{c13}) into (\ref{c6}), we get
\begin{equation}
  0.404853 = 0.659924 \cos\varphi \quad\Rightarrow\quad \varphi = 52.1581^\circ \;.
\end{equation}
Thus we have found the lower limit of the interval (\ref{c1}).

In the same way, we can calculate the upper bound of the interval (\ref{c1}) by substituting the upper bound of (\ref{c8}) into (\ref{c7}) and performing a calculation analogous to the one just performed.\\*[0.3cm]
(ii) The calculation of the interval (\ref{c2}) is also carried out in the same way, where we start by forming the interval $\bar{\alpha} \pm 0.6^\circ$ for the value $\bar{\alpha} = 81.7523^\circ$ from Table \ref{tab23a},
\begin{equation}
81.1523^\circ \le\bar{\alpha}\le 82.3523^\circ \;,
\end{equation}
and then carry out calculations analogous to those just done for its limits.
\end{proof}

\section{Procedures for {\tt Stellarium}}

\subsection{Determination of heliacal settings using {\tt Stellarium}}
\label{appD}
Let be given a star and its arcus visionis $\sigma$.
\begin{enumerate}
\item Start {\tt Stellarium} and set the observation location and year.
\item Stop the clock of {\tt Stellarium}.
\item Turn off ''Ground'' and ''Atmosphere'' (Main Tool Bar).
\item Turn on ''Azimuthal grid'' and ''Ecliptic (of date)'' (View $>$ Markings).
\item Choose any day of the year in question.
\item Change the time of day so that the star you are looking at is
  exactly on the line of the western horizon.
\item Enlarge the image so that a grid line is displayed for every
  integer degree value.
\item In the western sky, put a marker at the intersection of the
  ecliptic with the grid line that runs $\sigma$ degrees below the
  western horizon.
\item Return to normal magnification.
\item By changing the day, move the sun along the ecliptic until it is
  on the marker.
\item Move the considered star by changing the time of day until it is
  exactly on the western horizon.
\item Click on the Sun and read its depth below the horizon.
\end{enumerate}
If the depth of the Sun coincides with $\sigma$ with sufficient accuracy, then the day found is the heliacal setting of the star in question in the given year.

\subsection{Determination of beta days using {\tt Stellarium}}
\label{appC3}
Let be given the depth $h_0$ of the sun at which $\eta$ Tau should
appear vertically below $\beta$ Aur in the evening of the beta day.
\begin{enumerate}
\item Start {\tt Stellarium} and set the observation location and year.
\item Stop the clock of {\tt Stellarium}.
\item Turn off ''Ground'' and ''Atmosphere''.
\item Turn on ''Azimuthal grid'', ''Ecliptic (of date)'' and ''Altidude''.
\item Choose any day of the year in question.
\item Set the time of day to the transit time of $\eta$ Tau.
\item Select and center the star $\beta$ Aur. This makes $\beta$ Aur
  lie on the line of altitudes. $\eta$ Tau lies to the right of this
  line.
\item Now increase the time until $\eta$ Tau is also on the line of altitude. 
\item Enlarge the image so that a grid line is displayed for every
  integer degree value.
\item In the western sky, put a marker at the intersection of the
  ecliptic with the grid line that runs $h_0$ degrees below the
  western horizon.
\item Return to normal magnification.
\item By changing the day, move the sun along the ecliptic until it is
  on the marker.
\item Carry out steps (6), (7) and (8) again.
\item Click on the Sun and read its depth below the horizon.
\end{enumerate}
If the depth of the Sun coincides with $h_0$ with sufficient accuracy,
then the day found is the beta day in the given year.

\subsection{Determination of beta points using {\tt Stellarium}}
\label{appC4}
Let be given the depth $h_0$ of the sun at which $\eta$ Tau should
appear vertically below $\beta$ Aur in the evening of the beta day.
\begin{enumerate}
\item Start {\tt Stellarium} and set the observation location and year.
\item Stop the clock of {\tt Stellarium}.
\item Turn off ''Ground'' and ''Atmosphere''.
\item Turn on ''Azimuthal grid'', ''Ecliptic (of date)'' and
  ''Altidude''.
\item Torn on ''Use decimal degrees'' (Configuration $>$ Tools).
\item Choose any day of the year in question.
\item Set the time of day to the transit time of $\eta$ Tau.
\item Select and center the star $\beta$ Aur. This makes $\beta$ Aur
  lie on the line of altitudes. $\eta$ Tau lies to the right of this
  line.
\item Now increase the time until $\eta$ Tau is also on the line of
  altitude. Decrease the number of seconds until it is a multiple of 10.
\item Now read the azimuths Az$_\beta$, Az$_\eta$ of $\beta$ Aur and $\eta$ Tau. If $\eta$ Tau is to the right of $\beta$ Aur (Az$_\beta$ $<$ Az$_\eta$), then go 10s, 20s, 30s into the future until $\eta$ Tau is to the left of $\beta$ Aur (Az$_\eta$ $<$ Az$_\beta$). If, on the other hand, $\eta$ Tau is to the left of $\beta$ Aur, then go 10s, 20s, 30s into the past until $\eta$ Tau is to the right of $\beta$ Aur. We have thus found a seconds interval $[s_0, s_1]$ where $\eta$ Tau is at $s_0$ to the right and at $s_1$ to the left of $\beta$ Aur.
\item Bisect the interval $[s_0, s_1]$, where the division point $s^*$
  should be an integer value as close as possible to the center of the
  interval.
\item Among the intervals $[s_0, s^*]$ and $[s^*, s_1]$, choose the
  one on whose boundaries $\eta$ Tau lies on different sides of
  $\beta$ Aur. We call the new interval $[s_2, s_3]$.
\item Continue this iterval nesting until an interval $[s_n, s_{n+1}]$
  of length 1 is obtained, with $\eta$ Tau at $s_n$ to the left and at
  $s_{n+1}$ to the right of $\beta$ Aur.
\item Among the two time points $s_n$, $s_{n+1}$ choose the one where $|\text{Az}_{\eta} - \text{Az}_{\beta}|$ is the smallest.
\item Enlarge the image so that a grid line is displayed for every
  integer degree value.
\item In the western sky, put a marker at the intersection of the
  ecliptic with the grid line that runs $h_0$ degrees below the
  western horizon.
\item Enlarge the image further, where the intersection between the
  ecliptic and the $h_0$ line will move away from the marker. Place a
  new marker on this intersection.
\item Continue these magnifications and marker corrections until the
  maximum magnification is reached. The marker at the maximum
  magnification is the best approximation of the beta point.
\item Return to normal magnification.
\item Now select $\eta$ Tau, $\beta$ Aur, $\alpha$ Cyg and the marker
  of the beta point one by one and read the values needed for Table \ref{tab10}.
\end{enumerate}

\subsection{Determination of Tables \ref{tab7a} and \ref{tab13} using  {\tt Stellarium}}
\label{appE}

\begin{enumerate}
\item Start {\tt Stellarium} and set the observation location and year.
\item Stop the clock of {\tt Stellarium}.
\item Turn off ''Ground'' and ''Atmosphere''.
\item Turn on ''Azimuthal grid'' and ''Height''.
\item Choose any day of the year in question.
\item Click on the star Alcyone and read the time of its transit.
\item Set this time and adjust it until the azimuth $180^\circ0'$ is
  displayed for Alcyone.
\item Now increase the time in steps of hours or minutes until the
  next transit is reached.
\item Try to find the exact time at which Alcyone is exactly on the
  west or east horizon (setting or rising of the Pleiades).
\item When Alcyone approaches the position vertically below $\alpha$
  Aur or $\beta$ Aur, fix these stars with a double click. Then one can easily see when Alcyone is exactly vertically below these stars or when Alcyone is not making any relative movement in relation to these stars.
\end{enumerate}

\section{Our own observations of the sky}

\subsection{Observation locations}
We carried out our own observations of the sky at the observation locations given in Table \ref{tab21}.
We use as coordinates for the Mittelberg
\begin{table}[h] 
\begin{center}
\begin{tabular}{|l|c|c|c|}
\hline
location  & latitude & longitude  & height  \\
\hline
Nebra, Mittelberg &  $N51^\circ16'60.00"$ &    $E11^\circ31'13.00"$ &   251m\\
\hline
\end{tabular}
\end{center}
\end{table}

\begin{table}[t] 
\begin{center}
\begin{tabular}{|l|c|c|c|}
\hline
location  & latitude & longitude  & height  \\
\hline
Leipzig, Eichelbaumstr. & $N51^\circ18'16.15"$ & $E12^\circ19'32.83"$ &  119m \\
Leipzig, Sch\"onauer Str. & $N51^\circ18'26.15"$ &  $E12^\circ18'57.52"$ & 121m\\
Leipzig, Seebenischer Str. &  $N51^\circ17'40.49"$ &   $E12^\circ18'15.28"$ & 122m\\
\hline
\end{tabular}
\vspace{3mm}
\caption{Observation sites of our own sky observations.}
\label{tab21}
\end{center}
\end{table}

\subsection{The conjunction between Pleiades and full moon in autumn} \label{appE2}

We want to discuss the question of whether the Nebra people were able to observe the conjunction that takes place between the moon and the Pleiades in autumn. At this point the moon is said to be a good approximation of a full moon.

\begin{table}[t] 
\begin{center}
\begin{tabular}{|c|c|c|c|c|}
\hline
Date & Time & Moon  & angular distance  & height \\
Greg. && illuminated & Moon -- $\eta$ Tau & \\
\hline
2023-01-03 & 5:24:29 CET & 87.0\% & $3^\circ02'43.7"$& $-1^\circ25'06.67"$\\
2022-12-06 & 23:49:48 CET & 98.5\% & $2^\circ55'40.1"$ & $58^\circ43'01.68"$ \\
2022-11-09 & 15:26:22 CET &  98.6\% & $ 3^\circ24'25.3"$& $-9^\circ56'05.38"$ \\
2022-10-13 & 09:02:27 CEST & 88.1\% & $2^\circ58'31.1"$ & $21^\circ26'22.40"$ \\
2022-09-15 & 23:16:03 CEST & 68.9\% & $3^\circ29'16.6"$ & $14^\circ46'51.81"$ \\
\hline
\end{tabular}
\vspace{3mm}
\caption{Conjunctions between the Moon and $\eta$ Tau (Pleiades) in autumn and winter 2022/2023. (Calculated for Eichelbaumstra{\ss}e).}
\label{tab28}
\end{center}
\end{table}

First of all, it can be noted that from September to the end of the year, a conjunction between the Pleiades and the Moon often occurs in several months, during which the Moon is almost a full moon. Examples are shown in the Tables \ref{tab28}, \ref{tab29} calculated with {\tt Stellarium}. In Table \ref{tab28} there are conjunctions in 2022-11-09 and 2022-12-06, in which the Moon is almost full. During the conjunctions that actually fall in the autumn time in September and October, the Moon is still a bit away from being a full moon.

\begin{table}[t] 
\begin{center}
\begin{tabular}{|c|c|c|c|c|}
\hline
Date & Time  & Moon  & angular distance  & height \\
 Jul.& LMST & illuminated & Moon -- $\eta$ Tau & \\
\hline
-1601-09-06 & 21:22:38 & 89.3\% & $2^\circ36'6.38"$ &  $15^\circ10'51.9"$ \\
-1601-10-04 & 20:00:40 & 97.2\% & $6^\circ22'8.56"$ &  $18^\circ18'31.95"$ \\
-1601-10-31 & 18:57:37 &  98.4\% & $0^\circ52'47.42"$ &  $25^\circ33'6.35"$ \\ 
-1601-11-27 & 18:15:25 & 83.5\% & $5^\circ29'31.61"$ & $34^\circ44'42.25"$ \\
\hline
\end{tabular}
\vspace{3mm}
\caption{Conjunctions between the Moon and $\eta$ Tau (Pleiades) in autumn -1601. Calculated for Mittelberg for the time when the sun is $9^\circ$ below the horizon in the evening. LMST = Local Mean Solar Time. The height of the conjunction is the average of the height of the moon and the height of $\eta$ Tau.}
\label{tab29}
\end{center}
\end{table}

Table \ref{tab29} shows conjunctions for the Bronze Age -1601. Here too, there are conjunctions on the Julian dates -1601-10-04 and -1601-10-31, where the moon is almost full. The conjunction between the Pleiades and an almost full moon is not clearly determined and therefore not a good characteristic for the season "autumn".

However, a serious objection to the use of conjunctions between the Pleiades and a nearly full moon is that such conjunctions cannot be observed with the naked eye because the very bright full moon, when it is only a few degrees away from the Pleiades, outshines the Pleiades and they are then not visible at all.

We ourselves have observed the conjunction of the Pleiades and the Moon on 2022-10-13 as well as the daily approach of the Moon to the Pleiades since 2022-10-10. Table \ref{tab30} shows the azimuthal coordinates of the two celestial bodies at the observation times. The Pleiades are always very high, where they can normally be observed without any problems.

\begin{table}[t] 
\begin{center}
\begin{tabular}{|c|c|c|c|c|c|}
\hline
Date & Time  & Az Moon  & h Moon  & Az $\eta$ Tau & h  $\eta$ Tau\\
 Greg.& CEST  & & & & \\
 \hline
 2022-10-10 &  3:48 & $229.4581^\circ$ &  $33.0562^\circ$ &$181.5510^\circ$ & $62.8723^\circ$\\
2022-10-11 &  4:19 & $229.1098^\circ$& $40.1212^\circ$& $198.7065^\circ$& $61.9058^\circ$\\
2022-10-12 & 2:12 & $169.4165^\circ$& $54.1643^\circ$&$140.8417^\circ$&$58.3251^\circ$\\
2022-10-13 & 1:03 & $125.5653^\circ$&$48.5375^\circ$&$118.9302^\circ$&$50.4864^\circ$\\
\hline
\end{tabular}
\vspace{3mm}
\caption{Azimuth coordinates of the Moon and $\eta$ Tau during our observations in October 2022 (Eichelbaumstra{\ss}e).}
\label{tab30}
\end{center}
\end{table}

However, Table \ref{tab30b} shows that the visibility of the Pleiades deteriorated as the Moon approached them. Already on 2022-10-12, when the Moon was still $16^\circ$ from the Pleiades, the Pleiades were only visible as a faint flicker. At times they even disappeared. On 2022-10-13, when the Moon was only $4.7^\circ$ from the Pleiades, the Pleiades could no longer be seen. At this time, the Moon was no longer a full moon. Only about 90\% of its surface was illuminated. During a real full moon, the Pleiades would also not be visible.

\begin{table}[t] 
\begin{center}
\begin{tabular}{|l|}
  \hline
  The observation conditions were the same on all 4 nights. The air was clear,\\
  cold and had good transparency. There were no clouds. In the vicinity of the\\
  Pleiades, the constellations Taurus, Gemini, Auriga and Orion were visible.\\
  Because of the bright moonlight, however, only the brightest stars appeared.\\
  In Taurus, only $\alpha$, $\beta$ and $\zeta$ could be seen.\\
  \hline
\end{tabular}
\vspace{3mm}
\caption{Observation conditions for Table \ref{tab30b}.}
\label{tab37}
\end{center}
\end{table}

\begin{table}[t] 
\begin{center}
\begin{tabular}{|c|c|c|r|}
\hline
Date & Time & Moon  & angular distance  \\
Greg. & CEST & illuminated & Moon -- $\eta$ Tau \\
\hline
2022-10-10 &  3:48 & 99.9\% &  $42.1289^\circ$ \\
& &\multicolumn{2}{|l|}{Pleiades clearly visible.} \\
\hline
2022-10-11 &  4:19 & 98.4\% & $28.4702^\circ$\\
& &\multicolumn{2}{|l|}{Pleiades clearly visible.} \\
\hline
2022-10-12 & 2:12 & 95.0\% & $16.2770^\circ$\\
& &\multicolumn{2}{|l|}{Pleiades only visible as a faint flicker.} \\
& &\multicolumn{2}{|l|}{At times they disappeared.} \\
\hline
2022-10-13 & 1:03 & 89.7\% &$4.7262^\circ$\\
& &\multicolumn{2}{|l|}{Pleiades not visible.} \\
\hline
\end{tabular}
\vspace{3mm}
\caption{Visibility of the Pleiades during our observations October 2022 (Eichelbaumstra{\ss}e). All observations lasted approximately 15 minutes.}
\label{tab30b}
\end{center}
\end{table}

This means that the conjunctions of the Pleiades and a nearly full moon could hardly have been used by the Nebra people to indicate autumn or the harvest season.

\subsection{Heights of the celestial bodies observed}
\label{F.2}

In March and April 2020, we carried out astronomical observations of the stars in the sky around the Pleiades with the naked eye in order to gain experience with the visibility of the stars. We noted the times of the observations and later calculated the heights of the stars using {\tt Stellarium}. Tables \ref{tab15} to \ref{tab20} contain the data of these observations.

The first observation time for a star is the time at which we first noticed it. This time does not mean that the star first became visible at this time. However, a time marked with "V" is the setting time of the respective star. From this time onwards it was no longer visible.

\begin{table}[t] 
\begin{center}
\begin{tabular}{|c|c|c|c|c|c|}
\hline
\multicolumn{6}{|c|}{Heights of the stars} \\
\hline
Time (CET)  & Sun & $\alpha$ Tau  & $\eta$ Tau  &
Venus & Sirius \\
\hline
18:09&  $0^\circ0'$& & & & \\
18:22&  $-2^\circ2'$ & &   &                   $38^\circ42'$& \\
18:33 & $-3^\circ45'$& & &                     $37^\circ7'$  &   $20^\circ59'$\\
18:43&$  -5^\circ19'$  & &                     &$35^\circ39'$  &$  21^\circ19'$\\
18:50&$  -6^\circ24'$  &  &                    &$ 34^\circ36'$ &$   21^\circ30'$\\
18:52&$ -6^\circ43'$   & &                     &$  34^\circ18'$ &$   21^\circ33'$\\
19:11&$  -9^\circ39'$ &$ 48^\circ27'$  &$ 48^\circ45'$  &$ 31^\circ26'$ &$21^\circ52'$\\
20:11&$ -18^\circ41'$ &$ 40^\circ59'$  &$ 39^\circ50'$ &$ 22^\circ8'$ &$21^\circ13'$ \\
20:29&$ -21^\circ17'$ &$ 38^\circ28'$  &$ 37^\circ2'$  &$ 19^\circ20'$  &$ 20^\circ32'$  \\
\hline
magnitude &-26.75 &  0.85  &  2.85  &   -4.36  &   -1.45 \\
\hline
\end{tabular}
\vspace{3mm}
\caption{Observation 14.3.20, Sch\"onauer Str. ($\alpha$ Tau =
  Aldebaran, $\eta$ Tau = Alcyone (Pleiades)).}
\label{tab15}
\end{center}
\end{table}

\begin{table}[t] 
\begin{center}
\begin{tabular}{|c|c|c|c|c|c|c|c|}
\hline
\multicolumn{8}{|c|}{Heights of the stars} \\
\hline
Time (CET)  & Sun & $\alpha$ Tau & $\eta$ Tau  &
Venus & Sirius & $\alpha$ Aur  & $\beta$ Tau\\
\hline
18:16 &$ -0^\circ34'$&         &          &                    &        &     &\\
18:33 &$ -2^\circ42'$&         &          &$  38^\circ26'$&         &    &\\
18:38 &$ -3^\circ29'$&         &          &$  37^\circ41'$&$  21^\circ36'$&     &\\
18:51&$  -5^\circ30'$&$   48^\circ54'$&          &$  35^\circ45'$&$  21^\circ49'$&    &\\
18:52&$  -5^\circ39'$&$   48^\circ48'$&          &$  35^\circ36'$&$  21^\circ50'$&   &\\
18:54&$  -5^\circ58'$&$   48^\circ35'$&          &$  35^\circ18'$&$  21^\circ52'$&  &\\
18:57&$  -6^\circ26'$&$   48^\circ15'$&          &$  34^\circ51'$&$  21^\circ53'$& &\\
19:00&$  -6^\circ54'$&$   47^\circ56'$&$    48^\circ4'$&$   34^\circ23'$&$  21^\circ55'$& &\\
19:13&$  -8^\circ54'$&$   46^\circ26'$&$    46^\circ12'$&$  32^\circ24'$&$  21^\circ57'$&$ 74^\circ7'$&$  62^\circ37'$\\
19:22&$ -10^\circ16'$&$   45^\circ20'$&$    44^\circ52'$&$  31^\circ1'$&$   21^\circ54'$&$ 72^\circ44'$&$ 61^\circ40'$\\
19:31&$ -11^\circ39'$&$   44^\circ12'$&$    43^\circ31'$&$  29^\circ38'$&$   21^\circ47'$&$  71^\circ20'$&$  60^\circ37'$\\ 
20:10&$ -17^\circ26'$&$   38^\circ55'$&$    37^\circ33'$&$  23^\circ34'$&$   20^\circ40'$&$  65^\circ13'$&$  55^\circ34'$\\
\hline
magnitude & -26.75 &  0.85   &2.85   &-4.40 &   -1.45 &  0.05  &
1.65 \\
\hline
\end{tabular}
\vspace{3mm}
\caption{Observation 18.3.20, Sch\"onauer Str. ($\alpha$ Tau =
  Aldebaran, $\eta$ Tau = Alcyone (Pleiades), $\alpha$ Aur = Capella).}
\label{tab16}
\end{center}
\end{table}

\begin{table}[t] 
\begin{center}
\begin{tabular}{|c|c|c|c|c|c|c|}
\hline
\multicolumn{7}{|c|}{Heights of the stars} \\
\hline
Time (CET)  & Sun & $\alpha$ Tau & $\eta$ Tau  &
Venus & $\alpha$ Aur  & $\beta$ Tau\\
\hline
18:23&$   -0^\circ5'$& & & & & \\
18:32&$   -1^\circ29'$& & &$                  39^\circ47'$& & \\
18:45&$   -3^\circ30'$& & &$                  37^\circ50'$& & \\
18:59&$   -5^\circ41'$& & &$                  35^\circ43'$& & \\
19:02&$   -6^\circ9'$& & &$                   35^\circ16'$& & \\
19:03&$   -6^\circ18'$&$  45^\circ44'$& &$           35^\circ7'$& & \\
19:08&$   -7^\circ4'$&$   45^\circ7'$& &$            34^\circ21'$&$  72^\circ28'$&$     61^\circ28'$\\
19:10&$   -7^\circ22'$&$  44^\circ52'$& &$           34^\circ2'$&$   72^\circ9'$&$      61^\circ14'$\\
19:19&$   -8^\circ45'$&$  43^\circ44'$&$      42^\circ58'$&$    32^\circ39'$&$  70^\circ45'$&$     60^\circ11'$\\
19:30&$  -10^\circ25'$&$  42^\circ17'$&$      41^\circ17'$&$    30^\circ57'$&$  69^\circ2'$&$      58^\circ49'$\\
19:37&$  -11^\circ28'$&$  41^\circ20'$&$      40^\circ13'$&$    29^\circ52'$&$  67^\circ56'$&$     57^\circ55'$\\
19:46&$  -12^\circ48'$&$  40^\circ5'$&$       38^\circ50'$&$    28^\circ28'$&$  66^\circ31'$&$     56^\circ43'$\\
20:14&$  -16^\circ53'$&$  36^\circ4'$&$       34^\circ28'$&$    24^\circ6'$&$   62^\circ9'$&$      52^\circ46'$\\
20:23&$  -18^\circ9'$&$   34^\circ44'$&$      33^\circ4'$&$     22^\circ42'$&$  60^\circ45'$&$     51^\circ27'$\\
\hline
magnitude & -26.75 &  0.85   &2.85   &-4.46 &  0.05  &
1.65 \\
\hline
\end{tabular}
\vspace{3mm}
\caption{Observation 22.3.20, Sch\"onauer Str. ($\alpha$ Tau =
  Aldebaran, $\eta$ Tau = Alcyone (Pleiades), $\alpha$ Aur = Capella).}
\label{tab17}
\end{center}
\end{table}

\begin{table}[t] 
\begin{center}
\begin{tabular}{|c|c|c|c|c|c|}
\hline
\multicolumn{6}{|c|}{Heights of the stars} \\
\hline
Time (CET)  & Sun & $\alpha$ Tau & $\eta$ Tau  &
Venus & Sirius \\
\hline
22:12&$  -30^\circ37'$&$   16^\circ39'$&$        15^\circ6'$&$                  6^\circ52'$&$     8^\circ28'$\\
22:27&$  -31^\circ57'$&$   14^\circ19'$&$        12^\circ55'$&$                 4^\circ47'$&$     6^\circ38'$\\
22:37&$  -32^\circ46'$&$   12^\circ47'$&$        11^\circ28'$&$                 3^\circ25'$&$     5^\circ23'$\\
22:50&$  -33^\circ43'$&$   10^\circ47'$&$        9^\circ38'$&$                  1^\circ40'$&$     3^\circ41'$\\
22:56&$  -34^\circ8'$&$    9^\circ52'$&$         8^\circ48'$ V&$                0^\circ57'$&$     2^\circ54'$\\
23:01&$  -34^\circ27'$&$   9^\circ7'$& &$                                 0^\circ14'$ V&$   2^\circ13'$\\
23:11&$  -35^\circ1'$&$    7^\circ36'$& & &$                                          0^\circ52'$ V\\
23:29&$  -35^\circ52'$&$   4^\circ56'$& & & \\
23:46&$  -36^\circ25'$&$   2^\circ28'$ V& & & \\
\hline
magnitude & -26.75 &  0.85   &2.85   &-4.45 & -1.45 \\
\hline
\end{tabular}
\vspace{3mm}
\caption{Observation 24.3.20, Seebenischer Str. ($\alpha$ Tau =
  Aldebaran, $\eta$ Tau = Alcyone (Pleiades)). V = setting of the star.}
\label{tab18}
\end{center}
\end{table}

\begin{table}[t] 
\begin{center}
\begin{tabular}{|c|c|c|c|c|c|c|}
\hline
\multicolumn{7}{|c|}{Heights of the stars} \\
\hline
Time (CEST)  & Sun & $\alpha$ Tau & $\eta$ Tau  &
Venus & $\beta$ Tau & $\alpha$ Aur\\
\hline
20:42&$   -8^\circ27'$&$  32^\circ38'$& &$         32^\circ4'$& &$          58^\circ34'$\\
20:47&$   -9^\circ10'$&$  31^\circ52'$& &$         31^\circ18'$&$  48^\circ36'$&$    57^\circ47'$\\
20:55&$  -10^\circ19'$&$  30^\circ39'$& &$         30^\circ3'$&$   47^\circ23'$&$    56^\circ34'$\\
21:14&$  -12^\circ58'$&$  27^\circ43'$&$     25^\circ52'$&$   27^\circ5'$&$   44^\circ28'$&$    53^\circ39'$\\
21:35&$  -15^\circ48'$&$  24^\circ27'$&$     22^\circ37'$&$   23^\circ50'$&$  41^\circ12'$&$    50^\circ29'$\\
22:02&$  -19^\circ13'$&$  20^\circ13'$&$     18^\circ30'$&$   19^\circ43'$&$  36^\circ58'$&$    46^\circ29'$\\
\hline
magnitude & -26.74 &  0.85   &2.85   &-4.57 & 1.65 & 0.05 \\
\hline
\end{tabular}
\vspace{3mm}
\caption{Observation 5.4.20, Sch\"onauer Str. ($\alpha$ Tau =
  Aldebaran, $\eta$ Tau = Alcyone (Pleiades), $\alpha$ Aur =
  Capella). The Moon was also visible in the eastern sky with magnitude -12.14 (2 days before full Moon, 92\% illuminated surface). Venus was close to the Pleiades. The observation was impaired by a lot of interfering light.}
\label{tab19}
\end{center}
\end{table}

\begin{table}[t] 
\begin{center}
\begin{tabular}{|c|c|c|c|c|c|c|}
\hline
\multicolumn{7}{|c|}{Heights of the stars} \\
\hline
Time (CEST)  & Sun & $\alpha$ Tau & $\eta$ Tau  &
Venus & Sirius& Moon\\
\hline
21:47&$  -11^\circ34'$&$   9^\circ45'$&$         8^\circ42'$&$                 21^\circ53'$&$     2^\circ48'$&$     18^\circ11'$\\  
21:57&$  -12^\circ44'$&$   8^\circ15'$&$         7^\circ20'$  V&$              20^\circ24'$&$     1^\circ26'$&$     16^\circ43'$\\
22:19&$  -15^\circ10'$&$   4^\circ58'$& &$                               17^\circ11'$&$    -1^\circ37'$&$     13^\circ31'$\\
22:36&$  -16^\circ54'$&$   2^\circ30'$ V& &$                             14^\circ46'$& &$               11^\circ6'$\\
\hline
magnitude & -26.73 &  0.85   &2.85   &-4.54 & -1.45 & -7.73 \\
\hline
\end{tabular}
\vspace{3mm}
\caption{Observation 26.4.20, Seebenischer Str. ($\alpha$ Tau =
  Aldebaran, $\eta$ Tau = Alcyone (Pleiades)). V = setting of the
  star. Moon age: 3.4 days after the new Moon, 12.3\% illuminated. The
  Moon was close to the Pleiades. Its conjunction with the Pleiades
  took place on April 25, 2020, 21:41 CEST. The setting of Sirius went unnoticed.}
\label{tab20}
\end{center}
\end{table}

\newpage

\vspace{0.5cm}

\vspace{0.5cm}
\section*{Image sources}
\noindent Many of our graphics are based on the following sources:\\*[2mm]
\noindent [source 1] {\tt Nebra\_Scheibe\_white.jpg}\\
         \url{https://commons.wikimedia.org/wiki/File:Nebra_Scheibe_white.jpg?uselang=de}\\
         Author: Johannes Kalliauer\\
         License:
         \begin{itemize}
         \item GNU Free Documentation License, Version 1.2\\
           \url{https://commons.wikimedia.org/wiki/Commons:GNU_Free_Documentation_License,_version_1.2}
         \item Creative-Commons-Lizenz ''Namensnennung - Weitergabe unter gleichen Bedingungen 3.0 nicht portiert''\\
           \url{https://creativecommons.org/licenses/by-sa/3.0/deed.de}
         \end{itemize}

\noindent [source 2] {\tt Pflugwalle.jpg}\\
         \url{https://commons.wikimedia.org/wiki/File:Pflugwalle.jpg?uselang=de}\\
 Author: Matthias S\"u{\ss}en\\
         License:
         \begin{itemize}
         \item GNU Free Documentation License, Version 1.2\\
           \url{https://commons.wikimedia.org/wiki/Commons:GNU_Free_Documentation_License,_version_1.2}
         \item Creative-Commons-Lizenz ''Namensnennung - Weitergabe unter gleichen Bedingungen 3.0 nicht portiert''\\
           \url{https://creativecommons.org/licenses/by-sa/3.0/deed.de}
         \end{itemize}

         \noindent [source 3] {\tt Himmelstafel.Tal-Qadi.1024.png}\\
         \url{https://commons.wikimedia.org/wiki/File:Himmelstafel.Tal-Qadi.1024.png}\\
          Author: Markus Bautsch\\
         License:
         \begin{itemize}
         \item Creative Commons Attribution-Share Alike 4.0 International license.\\
           \url{https://creativecommons.org/licenses/by-sa/4.0/deed.en}
         \end{itemize}

         \noindent [source 4] {\tt Sidereal\_day\_(prograde).png}\\
         \url{https://commons.wikimedia.org/wiki/File:Sidereal_day_(prograde).png}\\
          Author: Gdr (Wikipedia user)\\
         License:
         \begin{itemize}
         \item GNU Free Documentation License, Version 1.2\\
           \url{https://commons.wikimedia.org/wiki/Commons:GNU_Free_Documentation_License,_version_1.2}
         \item Creative Commons Attribution-Share Alike 3.0 Unported license\\
           \url{https://creativecommons.org/licenses/by-sa/3.0/deed.en}
         \end{itemize}
       $\;$\\*[2mm]
         All of the above licenses allow:
         \begin{itemize}
         \item To share - copy and redistribute the material in any medium or format for any purpose, even commercially.
           \item To adapt - remix, transform, and build upon the material for any purpose, even commercially. 
           \end{itemize}
         In detail, our graphics were created as follows:\\*[2mm]
         {\bf Figure \ref{bigobj}:}\\
         Starting point [source 1]. Conversion to black and white. Inserting labels.\\*[2mm]
         {\bf Figure \ref{bear}:}\\
         Own work. Stellarium screenshot. Inserting lines.\\*[2mm]
         {\bf Figure \ref{pflug}:}\\
         Own work. Stellarium screenshot. Inserting lines and labels.\\*[2mm]
         {\bf Figure \ref{pflug2}:}\\
         Starting point [source 1]. Conversion to black and white. Inserting lines and labels.\\*[2mm]
         {\bf Figure \ref{walle}:}\\
         Starting point [source 2]. Paint over the plow with black. Replace the background with white. Mirror on a vertical axis.\\*[2mm]
         {\bf Figure \ref{talqadi}:}\\
         Starting point [source 3]. Conversion to black and white. Inserting labels.\\*[2mm]
         {\bf Figure \ref{pflug_steil}:}\\
         Own work. Stellarium screenshot. Inserting lines and labels.\\*[2mm]
         {\bf Figure \ref{auxiliaryline2}:}\\
         Starting point [source 1]. Conversion to black and white. Inserting lines and labels.\\*[2mm]
         {\bf Figure \ref{auxiliaryline}:}\\
         Own work. Stellarium screenshot. Inserting lines and labels.\\*[2mm]
         {\bf Figure \ref{stellarday}:}\\
         Starting point [source 4]. Conversion to black and white.\\*[2mm]
         {\bf Figure \ref{arcvis}:}\\
         Own work. Stellarium screenshot. Inserting lines and labels.\\*[2mm]
         {\bf Figure \ref{lineAP}:}\\
         Own work. Stellarium screenshot. Inserting lines and labels.\\*[2mm]
         {\bf Figure \ref{fig13}:}\\
         Own work. Stellarium screenshot. Inserting lines and labels.\\*[2mm]
         {\bf Figure \ref{fig14}:}\\
         Own work. Stellarium screenshot. Inserting lines and labels.\\*[2mm]
         {\bf Figure \ref{fig15}:}\\
         Own work. Stellarium screenshot. Inserting lines and labels.\\*[2mm]
         {\bf Figure \ref{HorEklipt}:}\\
         Own Work. Created with Mathematica.\\*[2mm]
         {\bf Figure \ref{azimutkreis}:}\\
         Own Work. Created with GIMP.\\*[2mm]
         {\bf Figure \ref{zenitkriterium}:}\\
         Own Work. Created with Mathematica.\\*[2mm]
         {\bf Figure \ref{zenit}:}\\
         Own Work. Created with Mathematica.\\*[2mm]
         {\bf Figure \ref{AuxLine4}:}\\
         Starting point [source 1]. Conversion to black and white. Inserting lines and labels.\\*[2mm]
         {\bf Figure \ref{vertical}:}\\
         Own work. Stellarium screenshot. Inserting lines and labels.\\*[2mm]
          {\bf Figure \ref{ausgleich1}:}\\
          Own Work. Created with Mathematica.\\*[2mm]
          {\bf Figure \ref{ausgleich2}:}\\
          Own Work. Created with Mathematica.\\*[2mm]
          {\bf Figure \ref{ausgleich3}:}\\
          Own Work. Created with Mathematica.\\*[2mm]
          {\bf Figure \ref{stellar1}:}\\
          Own Work. Created with Mathematica.\\*[2mm]
          {\bf Figure \ref{stellar2}:}\\
          Own Work. Created with Mathematica.\\*[2mm]
          {\bf Figure \ref{stellar3}:}\\
          Own Work. Created with Mathematica.\\*[2mm]
          {\bf Figure \ref{horizontwinkel}:}\\
          Starting point [source 1]. Conversion to black and white. Inserting lines and labels.\\*[2mm]
          {\bf Figure \ref{horizontwinkel2}:}\\
          Own Work. Created with GIMP.\\*[2mm]
          {\bf Figure \ref{kreisfit1}:}\\
          Own Work. Created with Mathematica.\\*[2mm]
          {\bf Figure \ref{kreisfit}:}\\
          Starting point [source 1]. Conversion to black and white. Inserting the fitting circles.\\*[2mm]
          {\bf Figure \ref{betaplej}:}\\
          Starting point [source 1]. Conversion to black and white. Inserting lines and labels.\\*[2mm]
          {\bf Figure \ref{betabigg}:}\\
          Starting point [source 1]. Conversion to black and white. Cutting out the image section. Inserting lines and labels.\\*[2mm]
          {\bf Figure \ref{verdrehung}a:}\\
          Starting point [source 1]. Conversion to black and white. Cutting out the image section. Adding labels.\\*[2mm]
          {\bf Figure \ref{verdrehung}b:}\\
          Starting point [source 1]. Rotating the image by $3.47^\circ$. Cutting out the image section. Conversion to black and white. Adding labels.\\*[2mm]
          {\bf Figure \ref{diagramme}:}\\
          Own Work. Created with GIMP.\\*[2mm]
           {\bf Figure \ref{sunset}:}\\
           Own work. Stellarium screenshot using a horizon line generated with Peakfinder. Inserting lines and labels.\\*[2mm]
           {\bf Figure \ref{peilung}:}\\
           Own Work. Created with GIMP.\\*[2mm]
           {\bf Figure \ref{1loesung}:}\\
           Own Work. Created with Mathematica.\\*[2mm]
           The author of this paper, the copyright holder of all figures in this paper, hereby publishes them under the following licenses:\\*[2mm]
Figure \ref{talqadi} is published under the
           Creative Commons Attribution-Share Alike 4.0 International license.\\
           \url{https://creativecommons.org/licenses/by-sa/4.0/deed.en}\\*[2mm]
           All other figures are published under the following licenses:
           \begin{itemize}
         \item GNU Free Documentation License, Version 1.2\\
           \url{https://commons.wikimedia.org/wiki/Commons:GNU_Free_Documentation_License,_version_1.2}
         \item Creative Commons Attribution-Share Alike 3.0 Unported license\\
           \url{https://creativecommons.org/licenses/by-sa/3.0/deed.en}
         \end{itemize}
           $\;$\\*[2mm]



\end{document}